\documentclass[pra,aps,amsmath,amssymb,superscriptaddress,twocolumn,longbibliography]{revtex4-1}
\usepackage{graphicx,multirow}
\usepackage{color,outlines}
\usepackage[sc,osf]{mathpazo}
\usepackage{amsthm}
\makeatletter
\renewcommand*\env@matrix[1][*\c@MaxMatrixCols c]{%
  \hskip -\arraycolsep
  \let\@ifnextchar\new@ifnextchar
  \array{#1}}
\makeatother

\usepackage{physics}
\usepackage{outlines}
\usepackage{hyperref}

\newcommand{\beq}{\begin{equation}}
\newcommand{\eeq}{\end{equation}}
\newtheorem{thm}{Theorem}

\newtheorem{definition}{Definition}
\newtheorem{prop}{Proposition}
\newtheorem{conj}{Conjecture}

\usepackage{tikz}
\usetikzlibrary{shapes.geometric, arrows}
\usetikzlibrary{decorations.pathmorphing}
\tikzset{snake it/.style={decorate, decoration=snake}}

\renewcommand{\arraystretch}{1.3}
\usepackage{booktabs}
\usepackage{makecell}

\begin{document}

\title{Construction and local equivalence of dual-unitary operators: from dynamical maps to quantum combinatorial designs} 
\author{Suhail Ahmad Rather}
\email{suhailmushtaq@physics.iitm.ac.in}
\affiliation{Department of Physics, Indian Institute of Technology
Madras, Chennai, India~600036}
\author{S. Aravinda}
\email{aravinda@iittp.ac.in}
\affiliation{ Department of Physics, Indian Institute of Technology Tirupati, Tirupati, India~517619} 
\author{Arul
Lakshminarayan}
%\email{arul@physics.iitm.ac.in} 
\affiliation{Department of Physics, Indian Institute of Technology
Madras, Chennai, India~600036}

\begin{abstract}

While quantum circuits built from two-particle dual-unitary (maximally entangled) operators serve as  minimal models of typically nonintegrable many-body systems, the construction and characterization of dual-unitary operators themselves are only partially understood. A nonlinear map on the space of unitary operators was proposed in PRL.~125, 070501 (2020) that results in operators being  arbitrarily close to dual unitaries.
Here we study the map analytically for the two-qubit case describing the basins of attraction, fixed points, and rates of approach to dual unitaries. 
A subset of dual-unitary operators having maximum entangling power are 2-unitary operators or perfect tensors, and are equivalent to four-party  absolutely maximally entangled states. It is known that they only exist if the local dimension is larger than $d=2$.
We use the nonlinear map, and introduce stochastic variants of it, to construct explicit examples of new dual and 2-unitary operators. A necessary criterion for their local unitary equivalence to distinguish classes is also introduced and used to display various concrete results and a conjecture in $d=3$.  
It is known that orthogonal Latin squares provide a ``classical combinatorial design"
for constructing permutations that are 2-unitary.  
We extend the underlying design from classical to genuine quantum
ones for general dual-unitary operators and give an example of what might be the smallest sized genuinely quantum design of a 2-unitary in $d=4$.
%
%While there are no general methods to determine if two unitary operators are local unitary (LU) equivalent, we introduce a necessary criterion for their equivalence to distinguish classes. 
%We show that in the case of two-qutrits (local dimension $d=3$), there is only one 2-unitary permutation, while there are 11 classes of dual-unitary permutations, upto LU equivalence and exchange of subsystems. Numerical evidence is provided that for qutrits there is only one 2-unitary operator class, even when not resricted to permutations.
%%
%For the case of two-ququads ($d=4$), we show that while
%there is still only one class of 2-unitary permutations, a non-permutation 2-unitary operator is explicitly constructed that is not LU equivalennt to them, thus providing a genuine quantum design of such operators. We also display results for $d=5$, where we show that there are at least two LU classes of 2-unitary permutations. 
%  
\end{abstract}

\maketitle

\tableofcontents 

\section{Introduction}

In recent years a major research trend uses tools of quantum information theory to understand the puzzles of quantum many-body physics. The typically complex entanglement structure of many-body states drives cross fertilization across various fields of research in physics. Particularly, in the areas of condensed matter physics and string theory, quantum information theory continues to play an exciting role in creating new avenues of understanding \cite{amico2008entanglement,zeng2015quantum,jahn2021holographic,kibe2021holographic}. 

Quantum computers allow the realization of the vision of Feynman \cite{feynman2018simulating} on the efficient simulation of physical systems \cite{georgescu2014quantum}.
In the present era of Noisy-Intermediate-Scale-Quantum (NISQ) \cite{preskill2018quantum} computing such simulations become realistic \cite{ippoliti2021many}. The universality of quantum computing allows simulation of any quantum system, where quantum circuits are built using unitary operators or gates acting on single particle and two particle subsystems. 
 %The unprecedented control that QC  provides helps to  experimentally simulate the new physics predicted by the theory. 
A quantum computer is itself a controllable quantum many-body system. Traditional approaches involve studying the properties of systems based on the Hamiltonian evolution and spectra.  At this juncture it is important to understand properties of quantum many-body systems from the quantum circuit formalism and to contrast that with the traditional studies.

Quantum advantage %over the classical computation involving the 
using random unitary circuits has been explored in recent experiments using Google's ``sycamore" processor \cite{QSupreme} and ``Zuchongzhi 2.0" \cite{zhu2021quantum}. Similar models are used in studies of entanglement evolution in  many-body quantum systems  in which the random unitary gates act on nearest neighbours \cite{Nahum2018, nahum2017quantum,khemani2018operator,von2018operator,chan2018solution}. Quantum circuits, in arbitrary local dimensions, without any random interactions have been proposed as elegant minimal models that can span the gamut of integrable to fully chaotic quantum many-body systems \cite{Akila2016,Bertini2019,gutkin2020exact,BraunPRE2020}. These quantum circuits have a special ``duality" property that the evolution operator in the spatial as well as in the temporal direction of the circuit are governed by unitary dynamics. The origin of this duality lies in the two-particle unitary gates being {\it dual-unitary} \cite{Bertini2019}. This duality facilitates an analytical treatment of many quantities such as two-point correlation functions, spectral form-factor, operator entanglement, out-of-time-order correlators and the exact study the entanglement dynamics \cite{Bertini2018,Bertini2019,bertini2019entanglement,gutkin2020exact,
bertini2019operator,kos2020chaos,lerose2020influence,garratt2020many,flack2020statistics,
bertini2019operatorII,piroli2020exact,claeys2020maximum,klobas2021exact, ippoliti2022fractal,ippoliti2021postselection,klobas2021entanglement,claeys2022emergent,zhou2022maximal}. The two-particle unitary from which the circuit is built, plays a significant role in the following and is an operator in $d^2$ dimensional space where the local dimension is $d$ and is typically denoted as $U$.

On the other hand entanglement of unitary operators, similar to states, has been studied from the early days of quantum information theory\cite{Zanardi2000,Zanardi2001,Wang2002,Wang2003,Nielsen2003,Vidal2002,Hammerer2002,Collins2001,Eisert2000,cirac2001}. Quantities such as operator entanglement \cite{Zanardi2001}, entangling power \cite{Zanardi2000} and complexity of an operator \cite{dowling2008geometry} are few measures that quantify the nonlocal properties of unitary operators. Operator entanglement measures how entangled an operator is when viewed as a vector in a product vector space.  It has been identified that an unitary operator is dual-unitary if and only if it has maximum possible operator entanglement \cite{SAA2020}. Entangling power quantifies the average entanglement produced by a bipartite unitary operator acting on an ensemble of pure product states. A special subclass of dual-unitary operators are those having the maximum possible entangling power allowed by local dimensions. These are the same operators that have been referred to variously as 2-unitary \cite{Goyeneche2015} or perfect tensors \cite{Pastawski2015}. 

The exactly calculable two-point correlation functions in dual-unitary circuits enable the characterization of the many-body system in terms of an ergodic hierarchy, from ergodic to Bernoulli through mixing \cite{Bertini2019,ASA_2021}. It is identified that the dual-unitary circuit is Bernoulli, when correlations instantly decay, if and only if the two-particle unitary operator has maximum entangling power \cite{ASA_2021}. Additionally, a sufficient condition for the many-body circuit to show the mixing behavior is derived as a function of entangling power \cite{ASA_2021}. These results establish a close connection between the entangling properties of two particle unitary operators from which the many-body quantum circuits are built and the dynamical nature of the  many-body systems.

Let a bipartite pure quantum state's Schmidt form be  $\ket{\psi}_{AB} = \sum_{i=0}^{d-1} \sqrt{\lambda_i} \ket{ii} /\sqrt{d} $, where $d$ is the local Hilbert space dimension. 
Setting $\lambda_i = 1 \; \text{for all} \; i$ in this expansion results in $\ket{\Phi}=\sum_{i=1}^{d}\ket{ii}/\sqrt{d}$ which is a maximally entangled state, in fact the one that is closest to $\ket{\psi_{AB}}$. Any set of orthonormal bases in the subspaces constructs such maximally entangled states.

In contrast the construction of maximally entangled unitary operators do not follow from orthonormal operator bases. Express an unitary operator in operator Schmidt form, $U = \sum_{j=1}^{d^2} \sqrt{\lambda_j} X_j \otimes Y_j $, $\tr (X_j^{\dagger}X_k) = \tr (Y_j^{\dagger} Y_k) = \delta_{jk}$ and $\lambda_j \geq 0$. $U$ is maximally entangled or dual-unitary iff $\lambda_i = 1\;\text{for all}\; i$.
However the constraint of unitarity is stronger than the constraint of normalization on the state and this imposes complex conditions on the Schmidt matrices $X_j$ and $Y_j$. Thus simply assigning $\lambda_i = 1 \; \text{for all}\; i$ does not retain unitarity although it does result in a maximally entangled operator. This makes it hard to analytically construct dual-unitary operators. Construction of maximally entangled unitaries was discussed at least as early as in Ref.~\cite{Tyson_2003}.

If the local dimension is $d=2$, namely for qubits, all possible dual-unitary operators can be parametrized using the Cartan decomposition \cite{Bertini2019}. There are no 
2-unitary or perfect tensors in this case. While a complete parametrization of dual-unitary operators for $d>2$ is not known, many classes and examples have been examined and used thus far. The {\sc swap} or flip operator is a simple example of a dual-unitary operator. The discrete Fourier transform in $d^2$ dimensions maximizes operator entanglement \cite{Tyson2003,Bhargavi2019} and is hence also a dual-unitary. However, the {\sc swap} has zero entangling power, while the Fourier transform has a finite value.
Diagonal and block diagonal operators, along with the {\sc swap} gate can be used to construct dual-unitary operator \cite{claeys2020ergodic,ASA_2021}. These have limited entangling power and in partcular cannot reach the maximum value \cite{ASA_2021}. The dual-unitary operators introduced recently in Ref. \cite{singh2021diagonal} are also bounded by the entangling power of diagonal unitary operators.

%At this juncture, we would like to investigate the procedure to construct the dual-unitary operator with the range of entangling power. 
A numerical iterative algorithm which produces unitary operators arbitrarily close to being dual-unitary has been presented in \cite{SAA2020}. This algorithm can yield dual-unitary operators with a wide range of entangling powers, especially exceeding the bound corresponding to block-diagonal based constructions. Remarkably, the numerical algorithm can also yield exact analytical forms for dual-unitary operators \cite{ASA_2021} and several other examples, including new 2-unitaries, are displayed further below. A slightly modified algorithm has been used to positively settle an open problem on the existence of four-party absolutely maximally entangled (AME) states of local dimension six \cite{SRatherAME46}, see also Ref.~\cite{AME46_conf} for an elaborate discussion of the solution. AME states are genuinely entangled multipartite pure states which have maximal entanglement in all bipartitions \cite{Helwig_2012}. Thus an AME state of $N$ qudits each of dimesion $d$ denoted as AME($N,d$) has all subsystems of size $\lfloor N/2 \rfloor$ maximally mixed.
%Further constructions of dual-unitary operators with a range of entangling powers, include those based on permutation matrices, quantized coupled cat maps, and complex Hadamard  matrices. %The goals of this paper are manifold. 

The numerical algorithm acts generically as a dynamical map in the space of unitary operators. 
These are thus high dimensional dynamical systems which deserve to be studied in their own right. 
In this work, we study the fixed point structure of the map. In particular for the case of two qubits, an explicit analytical form of the map is derived. This enables deriving dynamical characteristics such as the rate of approach to attractors which are dual-unitary operators. A variety of dynamical behaviours have been observed: (i) power-law approaches to the {\sc swap} gate, (ii) exponential approach to other dual unitary gates, with a rate that diverges for the maximally entangling case of the {\sc dcnot} gate.

%
%The early usage of matrix reshaping is by Sudarshan, Mathews and Rau in the characterization of open quantum system dynamics \cite{sudarshan1961stochastic}. Later extended  \cite{oxenrider1985matrix} for the investigation of the complete positivity of dynamical maps. The celebrated mixed state entanglement witness introduced by A. Peres is  based on one of the matrix reshaping, called partial transpose \cite{peres1996separability}. Later, another matrix reshaping, called realignment is also used to witness the mixed state entanglement \cite{chen2002matrix,rudolph2003some}.  
 
Due to operator-state isomorphism, a 2-unitary operator is equivalent to a four-party AME state \cite{Goyeneche2015}. There are various ways of constructing AME states in which quantum combinatorial designs are used \cite{Goyeneche2015,GRMZ_2018}. Since 2-unitaries are a subset of dual unitaries, a less restrictive combinatorial design underlying dual-unitary operators that are permutations was found in \cite{ASA_2021}. In this work, we extend such combinatorial designs to dual-unitary operators going beyond permutations. We define {\it stochastic} dynamical maps capable of generating such structured dual-unitary operators. 

Apart from dual-unitary operators, the maps presented in \cite{SAA2020} can be used to generate infinitely many 2-unitaries for $d>2$. An important question now arises if the 2-unitaries so obtained are local unitarily (LU) equivalent to each other and to 2-unitary permutations of the same size. Two bipartite unitary operators $U$ and $U'$ on $\mathcal{H}_d \otimes \mathcal{H}_d$ are said to local unitarily equivalent, denoted $U\overset{\text{LU}}{\sim}U'$,  if there exist single qudit unitary gates $u_1, u_2,v_1, v_2$ such that 
\beq
U'=(u_1 \otimes u_2)U (v_1 \otimes v_2).
\label{eq:U1U2equi}
\eeq
However, as far as we know, there is no general procedure to identify LU equivalent unitary operators, apart from the case of two-qubits \cite{cirac2001}. This problem becomes acute when the operators concerned have the same entangling powers, such as the 2-unitaries. In this work, we address this question by proposing a necessary criterion for LU equivalence between bipartite unitary operators
that can potentially work also in the case of 2-unitaries.

This leads us to conjecture that {\it all} two-qutrit 2-unitaries are LU equivalent to each other. From an exhaustive search, we find that the special subset of $72$ possible 2-unitary permutations of size 9 are LU equivalent to each other. For local Hilbert space dimension $d=4$, we find that this still continues to hold: indeed there are $6912$ 2-unitary permutations, and we find that these can generated from any one of them by local permuations. Thus upto LU equivalence, we find that there is {\it only one} 2-unitary permutation in $d=3,4$ and implies that there is only one unique orthogonal Latin square of size $3$ and $4$. Note that the connection between 2-unitary permutations and orthogonal Latin squares has been known 
for some time \cite{Clarisse2005}.

Although there is only one 2-unitary permutation upto LU equivalence,  further below, we give an explicit example of a 2-unitary of size 16 which is not LU equivalent to any 2-unitary permutation of the same size. In other words, we give an explicit example of an AME state of four ququads which is not LU connected to an AME state of four ququads with minimal support. Minimal support four-party AME states have $d^2$ nonvanishing coefficients in some product orthonormal basis, which is the smallest number possible \cite{Goyeneche2015}. These new examples of AME states can be used to construct new error-correcting codes as was done in \cite{SRatherAME46} and can provide insights about the most general underlying combinatorial designs 2-unitaries possess. For $d=5$ we show that there are two LU inequivalent 2-unitary permutations or equivalently, two LU inequivalent AME states of minimal support. This contradicts {\em Conjecture 2} in Ref.~\cite{Adam_SLOCC_2020} which implies that all four party AME states of minimal support are LU equivalent for all local dimensions $d$.
% The The numerical algorithm is used to construct several new examples of dual-unitary and AME states in $d=3$ and $d=4$. A 2-unitary operator for $d=4$ is obtained from the numerical algorithm and it is shown that the unitary is local unitarily inequavalent to any permutation 2-unitary operator. 

\subsection{Summary of principal results and structure of the paper}

In view of the length of this paper, we summarize some of the main results here.
\begin{enumerate}
\item{Dynamical maps for generation of dual-unitary operators, Sections ~(\ref{sec:dynamicalmaps}, \ref{Sec:weyldynamics})}
\begin{enumerate}
\item
It is shown in Proposition \ref{prop:loc_orbit_map}  that the map preserves local unitary equivalence. 

\item
Explicit form of the map for the two-qubit case  is derived. This is used to show that all dual unitaries are period-two points, and conversely all period-two points are dual unitaries. It is shown that convergence of the map to dual unitaries is typically exponential. 

\end{enumerate}
\item{Quantum designs and new classes of 2-unitaries and AME states, Sec. ~(\ref{sec:combi}--~(\ref{Sec:permu})}
\begin{enumerate}
%\item
%There is a known connection between dual-unitary permutations and combinatorial objects. We generalize this approach to dual-unitary gates which need not be permutations.
\item
A necessary criterion for LU equivalence between bipartite unitary gates is proposed, and is particularly useful for establishing inequivalence between 2-unitary operators and AME states.
%based on the entanglement distribution they produce when acted on uniformly distributed product states. 
%This criteria in 
\item
An AME state of four qudits each of local dimension 4, AME$(4,4)$, is constructed such that it is not LU equivalent to any known AME state obtained from classical orthogonal Latin squares. It is likely to be the simplest genuine  orthogonal quantum Latin square construction.  

\item
For local dimensions $d=3,$ and $d=4$, it is shown that there is only one LU class of AME states constructed from orthogonal Latin squares (OLS). However, we show that there are more than one such LU class for $d>4$.

\end{enumerate}
\end{enumerate}

The paper is structured as follows. In Sec.~(\ref{sec:def}), the basic terminologies used in the current work is defined. In  Sec. ~(\ref{sec:dynamicalmaps}), the non-linear iterative maps from which dual-unitary and 2-unitary operators are produced is described, and  their fixed point structure is discussed. Stochastic generalizations are introduced to result in specially structured operators. In Sec.~(\ref{Sec:weyldynamics}) the iterative map is studied in explicit forms for the case of qubits. Here we analytically estimate the power-law or exponential approach to dual-unitary operators. In Sec (\ref{sec:combi}) combinatorial designs corresponding to dual-unitary operators is discussed. In Sec.~ (\ref{Sec:spldual}) the question of local unitary equivalence of 2-unitary operators is discussed, specially for the small cases of $d=3$, and $4$. In Sec. (\ref{Sec:permu}) permutations small dimensions is studied in detail via their entangling powers and gate-typicality. Classification of LU classes for dual-unitary and T-dual unitary permutation operators is given for $d=2,3$. Finally we conclude in Sec.~(\ref{sec:conc}).

\section{Preliminaries and definitions \label{sec:def}} 

In this section, we mostly recall some relevant quantities and measures.

\subsection{Operator entanglement and entangling power \label{sec:op_ent}}
%The entanglement properties of an unitary operator $U \in \mathcal{H}_d \otimes \mathcal{H}_d$ can be defined by considering its Choi-Jamiolkowski isomorphism  to the state $\ket{U} \in  \mathcal{H}_d^{\otimes 2} \otimes \mathcal{H}_d^{\otimes 2}$. 

Any operator $X \in \mathcal{B}(\mathcal{H}_d)$ is mapped to the state $\ket{X} \in \mathcal{H}_d \otimes \mathcal{H}_d $ as  
\begin{equation}
	\ket{X} := (X \otimes I) \ket{\Phi},
	\label{eq:iso}
\end{equation}
where $\{\ket{i}\}_{i=1}^d$ is an orthonormal basis in $\mathcal{H}_d$ and $\ket{\Phi} := \frac{1}{\sqrt{d}}\sum_{i=1}^d \ket{ii}$ is the generalized Bell state.
A bipartite unitary operator $U = \sum_{ij}^{d^2} \alpha_{ij} e_i^A \otimes e_j^B \in  \mathcal{B}(\mathcal{H}_d \otimes \mathcal{H}_d) $, is mapped to $\ket{U} = \sum_{ij}^{d^2} \alpha_{ij} \ket{e_i}_A \otimes \ket{e_j}_B \in \mathcal{H}_{d^2} \otimes \mathcal{H}_{d^2}$,  
where $e^{A,B}_{i}$ are a pair of operator bases in $\mathcal{B}(\mathcal{H}_d)$.

The entanglement of an unitary operator $U$ is the entanglement of the state $\ket{U}=\sum_{j=1}^{d^2} \sqrt{\lambda_j}\ket{X^A_j}\ket{ Y^B_j}$.  
%The invariance of the Schmidt coefficients under the CJ isomorphism enables to define the entanglement of unitary operator in terms of Schmidt co-efficients.  
The Schmidt decomposition of $U$ is given by
\begin{equation}
U = \sum_{j=1}^{d^2} \sqrt{\lambda_j} X_j^A \otimes Y_j^B,
\label{eq:sch}
\end{equation}
with $\tr{{X_j^A}^\dagger X_k^A} = \tr{{Y_j^B}^\dagger Y_k^B} = \delta_{jk}, \quad  \sum_{j = 1}^{d^2}\lambda_j =d^2$.
The {\it operator entanglement} of $U$ is defined in terms of linear entropy as 
\begin{equation}
	E(U) = 1 - \frac{1}{d^4} \sum_{j=1}^{d^2} \lambda_j^2,
	\label{eq:EU}
\end{equation}
where $0 \leq E(U) \leq 1 - 1/d^2$. 

Another related, but distinct, entanglement facet of an unitary operator $U$ is its {\it entangling power}, $e_p(U)$. It is defined as the average entanglement  produced due to its action on pure product states distributed according to the uniform, Haar measure, 
\begin{equation}
	e_p(U) = C_d\overline{\mathcal{E}(U \ket{\phi_1} \otimes \ket{\phi_2})}^{\ket{\phi_1},\ket{\phi_2}},
\end{equation}
where $\mathcal{E}(\cdot)$ can be any entanglement measure, and $C_d$ is a constant scale factor. Considering $\mathcal{E}(\cdot)$ to be the linear entropy, the entangling power can be directly calculated using operator entanglement \cite{Zanardi2001} as follows. Let $S$ be the {\sc swap} operator such that 
\beq
S\ket{\phi_A}\ket{\phi_B} = \ket{\phi_B}\ket{\phi_A},
\eeq
for all $\ket{\phi_A} \in \mathcal{H}_d,\,\ket{\phi_B} \in \mathcal{H}_d$.
% \quad S(u_A\otimes u_B)S = u_B\otimes u_A, \forall u_A,  u_B \in \mathcal{H}_d \otimes \mathcal{H}_d$. 
We choose $C_d$ such that the scaled entangling power $0\leq e_p(U) \leq 1$ is given by 
\begin{equation}
	e_p(U) = \frac{1}{E(S)} [E(U) + E(US) - E(S)],
	\label{eq:ep}
\end{equation}
where $E(S) = 1 - 1/d^2$. 

Note that the swap operator is such that it has the maximum possible operator entanglement, however the entangling power $e_p(S)=0$. For any operator $U$, $e_p(U)=e_p(US)=e_p(SU)$. The so-called gate-typicality $g_t(U)$ \cite{Bhargavi2017} distinguishes these and is defined as
\begin{equation}
	g_t(U) = \frac{1}{2E(S)} [E(U) - E(US) + E(S)],
	\label{eq:gt}
\end{equation}
and also ranges from $0$ to $1$, with $g_t(S)=1$ and vanishes for all local gates.

\subsection{Matrix reshaping \label{sec:matshap}}
A bipartite unitary operator $U$ on $\mathcal{H}_d \otimes \mathcal{H}_d$ can be expanded in product basis as 
\begin{equation} 
	U  = \sum_{i \alpha j \beta} \mel{i\alpha}{U}{j\beta} \ket{i\alpha}\bra{j\beta}
	  \label{eq:umat}
\end{equation} 
There are four basic matrix rearrangements of $U$ that we use in this work : 
\begin{enumerate}
\item
Realignment operations:
\begin{align}
\quad R_1 : \mel{i\alpha}{U}{j\beta} & = \mel{\beta\alpha}{U^{R_1}}{ji} \\ 
  \quad R_2 : \mel{i\alpha}{U}{j\beta} & = \mel{ij}{U^{R_2}}{\alpha\beta}
 \end{align}
 \item
 Partial Transpose operations:
 \begin{align}
 \quad \Gamma_1 : \mel{i\alpha}{U}{j\beta} & = \mel{j\alpha}{U^{\Gamma_1}}{i\beta} \\ 
\quad \Gamma_2 : \mel{i\alpha}{U}{j\beta} & = \mel{i\beta}{U^{\Gamma_2}}{j\alpha}  
\end{align}

\end{enumerate}

The relation between entanglement and matrix reshapings becomes clear on considering the state $\ket{U}$ as now a 4-party state $\ket{\psi} \in \mathcal{H}_d\otimes \mathcal{H}_d\otimes \mathcal{H}_d\otimes \mathcal{H}_d $: 
\begin{equation}
 \ket{\psi}_{ABCD} = (U_{AB} \otimes I_{CD}) \ket{\Phi}_{AC}\ket{\Phi}_{BD}, 
 \label{eq:4st}
\end{equation}
where $\ket{\Phi} = \frac{1}{\sqrt{d}}\sum_i^d \ket{ii}$ is the generalized Bell state. The reduced state corresponding to the three possible partition $AB|CD$, $AC|BD$ and $AD|BC$ are given by
\begin{equation}
  \rho_{AB} = \frac{1}{d^2} UU^\dagger, \,
  \rho_{AC} = \frac{1}{d^2}U^{R_2}U^{R_2\dagger}, \,
  \rho_{AD} = \frac{1}{d^2}U^{\Gamma_2}U^{\Gamma_2 \dagger}. 
\label{eq:redst}
\end{equation}
%\begin{equation}
% \begin{split}
%  \rho_{AB} &= \frac{1}{d^2} UU^\dagger\\
%  \rho_{AC} &= \frac{1}{d^2}U^{R_1}U^{R_1\dagger}\\
%  \rho_{AD} &= \frac{1}{d^2}U^{\Gamma_2}U^{\Gamma_2 \dagger}
% \end{split}. 
%\label{eq:redst}
%\end{equation}
Using $(X \otimes Y)^{R_2} = \op{X}{Y^*}$, where * refers to complex conjugation in the computational basis. It is easy to see $U^{R_2}U^{R_2\dagger} = \sum_1^{d^2} \lambda_j \op{X_j}{X_j}$. The Schmidt value $\lambda_j$ are the singular values of $U^{R_2}$ (which are the same as the singular values of $U^{R_1}$). The operator entanglement $E(U)$ can be interpreted as the linear entropy of entanglement of the bipartition $AC|BD$, and can be expressed in terms of $U^{R}$ as 
\begin{equation}
	E(U) = 1 - \frac{1}{d^4} \Tr{(U^RU^{R\dagger})^2}.
	\label{eq:euR}
\end{equation}

Similarly, the operator entanglement $E(US)$ is the linear entropy of the bipartition $AD|BC$ and is
\begin{equation}
	E(US) = 1 - \frac{1}{d^4} \text{Tr}{(U^{\Gamma}U^{\Gamma \dagger})^2}.
	\label{eq:EUS}
\end{equation}
Whenever the subscripts on $R$ and $\Gamma$ have been dropped they can refer equally to either of the two operations. Note that the singular values of $U^R$ and $U^{\Gamma}$ are all local unitary invariants (LUI). 

We recall definitions of some special families of unitary operators and also introduce some new families of unitary operators.
\begin{definition}[Dual unitary \cite{Bertini2019}]\label{def:dual}
	If the realigned matrix $U^R$ of unitary operator $U$ is also unitary, then $U$ is called a dual unitary. 
\end{definition}

\begin{definition}[T-dual unitary \cite{ASA_2021}] \label{def:Tdual}
If the partial transposed matrix $U^{\Gamma}$ of a unitary operator $U$ is also unitary, then $U$ is called a T-dual unitary. 
\end{definition}

\begin{definition}[2-unitary \cite{Goyeneche2015}] \label{def:2-uni}
	A unitary $U$ for which both $U^R$ and $U^{\Gamma}$ are also unitary is called 2-unitary. 
\end{definition}

\begin{definition}[Self dual unitary] \label{def:selfdual}
Unitary operator $U$ for which $U^R = U$ is called a self-dual unitary.
\end{definition} 

%\begin{definition}[Self T-dual unitary] \label{def:selftdual}
%Unitary operator $U$ for which $U^{\Gamma} = U$ is called a self T-dual unitary. 
%\end{definition} 
Note that for a 2-unitary $E(U)= E(US)=E(S)=1-1/d^2$ are maximized and thus from Eq.~(\ref{eq:ep}) $e_p(U)=1$, the maximum possible value. Thus the corresponding four-party state given by Eq.~(\ref{eq:4st}) is maximally entangled along all three bipartitions and is an absolutely maximally entangled state of four qudits; AME($4,d$).

In the mathematics literature, the class of unitary operators which remain unitary under `block-transpose' have been studied since 1989 \cite{ocneanu1988quantized,krishnan1996biunitary,jones1999planar,reutter2016biunitary,Nechita2017,kodiyalam2020planar,nechita2021sinkhorn} . Referred to as biunitaries, they are dual unitary upto multiplication by {\sc swap}, and is the result of the $\Gamma_1$ operation above. However, the term ``biunitary" seems to be used interchangably for both dual and T-dual unitary operators and subsequently no special studies of 2-unitaries, that are both dual and T-dual seems to exist.

T-dual and dual unitary have very different entanglement properties, as reflected in their two most prominent representatives: the identity and the {\sc swap} gate. However, they are related in the sense that every T-dual unitary $U$ has a dual partner $US$ (or $SU$). 
Note that if $U$ is 2-unitary, so also are $U^R$ and $U^{\Gamma}$. For example, the realignment of $U^R$ is $U$ itself, while $(U^R)^{\Gamma}=U^{\Gamma} S$, which is evidently unitary given that $U$ is 2-unitary.

\section{Dual-unitary  and 2-unitary operators from nonlinear iterative maps
\label{sec:dynamicalmaps}}
Complete parametrization of dual unitary operators for arbitrary local Hilbert space dimension $d$ is not known in general except in the two-qubit case \cite{Bertini2019}. Several (incomplete) analytic constructions of families of dual unitary operators have been proposed based on complex Hadamard matrices \cite{gutkin2016classical}, diagonal \cite{claeys2020ergodic}, and block-diagonal unitary matrices \cite{ASA_2021}. Here we briefly review the non-linear maps introduced in \cite{SAA2020} to generate unitary operators which are arbitrarily close to dual unitaries.
\subsection{Dynamical map for dual unitaries \label{subsec:MRmap}}
The following map is defined on the space of bipartite unitary operators,
$$
\mathcal{M}_R:\mathcal{U}(d^2) \longrightarrow \mathcal{U}(d^2).$$
One complete action of $\mathcal{M}_R$ on a seed unitary $U_0$ consists of the following two steps:
\begin{enumerate}
\item[(i)] 
{\em Linear part}: Realignment of $U_0$: $U_0 \xrightarrow{R}U_0^R$, 
\item[(ii)]
{\em Non-linear part}: Projection of $U_0^R$ to the nearest unitary matrix $U_1$ given by its polar decomposition (PD) \cite{Fan1955, Keller1975}; $U_0^R=U_1 \sqrt{U_0^{R\, \dagger} U_0^R}$.
\end{enumerate}
Note that $U_0^R$ must be of full rank for the map to be well defined as the polar decomposition of rank deficient matrices is not uniquely defined.
We write one complete action of the map on $U_0$ as
\[\mathcal{M}_R[U_0]:=U_1.\]
After $n$ iterations,
\[\underbrace{\mathcal{M}_R \circ \mathcal{M}_R \circ \cdots \circ \mathcal{M}_R}_{n \; \text{times}}[U_0]:=\mathcal{M}_R^n[U_0]=U_n.\]
%such that $U_n$ is arbitrarily close to a dual unitary for sufficiently large $n$. 
For arbitrary seeds the map  has been observed to converge to dual-unitaries almost certainly \cite{SAA2020}, and this is  made plausible by the following observations on the fixed points of the map $\mathcal{M}_R$.

%The dynamical maps discussed here have the following important property.
%{\em Seed unitary dressed with local gates on both sides:}

An important property of the map is that {\em it preserves the local orbit} of seed unitary $U_0$ in the following sense.
\begin{prop}
\beq
\text{If}\; U_0' \stackrel{\text{LU}}{\sim} U_0,\; \text{then}\;  U_1'=\mathcal{M}_R[U_0'] \stackrel{\text{LU}}{\sim} U_1= \mathcal{M}_R[U_0].
\label{eq:orbitpreserve}
\eeq
\label{prop:loc_orbit_map}
\end{prop}
\begin{proof}
We start from the identity \cite{ASA_2021},
\beq
\begin{split}
U_0^{' R }=& 
[(u_1 \otimes u_2) U_0 (v_1 \otimes v_2)]^R\\
&=(u_1 \otimes v_1^T) U_0^R (u_2^T \otimes v_2).
\end{split}
\eeq
where $T$ is the usual transpose.
Using the polar decomposition of $U_0^{R}$ we get 
\beq
\begin{split}
U_0^{' R } &= (u_1 \otimes v_1^T) U_1 \sqrt{U_0^{R\, \dagger} U_0^R}(u_2^T \otimes v_2)\\ & = (u_1 \otimes v_1^T) U_1 (u_2^T \otimes v_2) \sqrt{U_0^{' R\, \dagger} U_0^{' R}}\\&  \equiv U_1' \sqrt{U_0^{' R\, \dagger} U_0^{' R}}.
\end{split}
\eeq
Explicitly
\beq
\mathcal{M}_R[U_0']=(u_1\otimes v_1^T) \mathcal{M}_R[U_0] (u_2^T \otimes v_2).
\label{eq:Map_Locals}
\eeq
\end{proof}
Thus the changes in the operator entanglement under the $\mathcal{M}_R$ map are unaffected by local unitary operations.

Analogous to the $\mathcal{M}_R$ map for dual unitaries, one can define $\mathcal{M}_{\Gamma}$ map to generate T-dual unitary operators. Action of $\mathcal{M}_{\Gamma}$ on $U_0$ is defined as $\mathcal{M}_{\Gamma}[U_0]:=U_1$, where $U_1$ is closest unitary to $U_0^{\Gamma}$ given by its polar decomposition. 
%(P.D),
%\beq
%U_0 \xrightarrow{\Gamma} U_0^{\Gamma} \xrightarrow{\text{P.D}} U_1.
%\eeq
%For a T-dual unitary operator $\mathsf{U}$, $\mathcal{M}_{\Gamma}^2[\mathsf{U}]=U$ i.e., T-dual unitary operators are period-2 fixed points of the $\mathcal{M}_{\Gamma}$ map. 
Such an algorithm was independently studied in \cite{Nechita_2017} to generate a special class of random structured bipartite unitary operators.

\subsubsection{Dual unitaries as fixed points} Action of the map on a dual unitary $\mathsf{U}$ is
\[\mathcal{M}_R[\mathsf{U}]=\mathsf{U}^R,\]
as $\mathsf{U}^R$ is also unitary. As the realignment operation is an involution, $(X^R)^R=X$, it follows that $\mathcal{M}_R[{\mathsf{U}}^R]=\mathsf{U}$ and 
\beq
\mathcal{M}_R^2[\mathsf{U}]=\mathsf{U},
\eeq
{\it i.e.}, dual unitaries are period-2 fixed points of the map. Note that self-dual unitaries ($\mathsf{U}^R=\mathsf{U}$) are fixed points of the $\mathcal{M}_R$ map itself. 

For the two-qubit case ($d^2=4$), we prove that dual-unitaries are the \emph{only} fixed points of the $\mathcal{M}_R^2$ map, or equivalently the perod-2 orbits of $\mathcal{M}_R$. 
However, for the two-qutrit case ($d^2=9$), there are fixed points of the $\mathcal{M}_R^2$ map other than dual unitaries; see appendix \ref{app:nondualFP} for an explicit example. In this case $U_0^R=U_1 \sqrt{U_0^{R\, \dagger} U_0^R}$, $U_1^R=U_0 \sqrt{U_1^{R\, \dagger} U_1^R}$, and the pair $U_0$ and $U_1$ conspire such that they are the nearest unitary to the other's realignment. Generic seeds are neither of this kind, nor do they seem to end up in such pairs.
%Starting from such a non-dual period-2 orbit of the $\mathcal{M}_R$ map,  the Schmidt coefficients are preserved. This suggests that $U_0$ and $U_1$ may be locally equivalent, as we demonsrate for the example in \ref{app:nondualFP}.

For $d^2>9$ we have not been able to find such non-dual fixed points. The reason that the map does not converge to such fixed points is because of large dimensionality a random sampling of seed unitaries over the corresponding unitary group $\mathcal{U}(d^2)$  with $(d^2-1)$ number of parameters is unable to find appropriate seed unitaries which lead to such fixed points. One might also expect higher order fixed points of the map which makes the map a novel dynamical system in its own right but for the purposes of this work we will not focus on such directions.

\subsection{Dynamical map for 2-unitaries}
The set of 2-unitary operators is a common intersection of dual and T-dual unitaries. In order to generate 2-unitary operators a slightly modified map $\mathcal{M}_{\Gamma R}$ is used by incorporating also the partial transpose operation. Schematically the action of the map on seed unitary $U_0$ is

%In order to obtain 2-unitaries we define a slightly modefied map $M_{TR} : U(N^2) \rightarrow U(N^2) $. Firt the linear map is breaked into two parts as $L_{RT}: U \rightarrow U^R \rightarrow U_{RT}$. The nonlinear map $NL_{TR} : U_{RT} \rightarrow V$, where $V$ is agin the closest unitary to the matrix $U_{RT}$ and is been given by its polar decomposition. 
\beq
\mathcal{M}_{\Gamma R}:\;\;U_0 \xrightarrow{R} U_0^R \xrightarrow{\Gamma} \left(U_0^{R}\right)^{\Gamma}:=U_0^{\Gamma R} \xrightarrow{\text{PD}} U_1.
\label{eq:MTRmap}
\eeq
%An $n$ number of repeated actions of $\mathcal{M}_{\Gamma R}$ map on seed unitary $U_0$ are denoted as
%\[\underbrace{\mathcal{M}_{\Gamma R} \circ \mathcal{M}_{ \Gamma R} \circ \cdots \circ \mathcal{M}_{\Gamma R}}_{n \; \text{times}}[U_0]:=\mathcal{M}_{\Gamma R}^n[U_0]=U_n.\]
It has been pointed out previously that sampling $U_0$ from the Circular Unitary Ensemble (CUE), for small local dimensions ($d^2=9,\;16$), $U_n=\mathcal{M}_{\Gamma R}^n[U_0]$ is arbitrarily close to being 2-unitary, with a significant probability; around $95\%$ for $d^2=9$ and $20\%$ for $d^2=16$ \cite{SAA2020}. To generate 2-unitary operators in larger dimensions, one may need to start with an appropriate initial seed unitary, not just sampled from the CUE. This was done in \cite{SRatherAME46} to generate a 2-unitary of order $36$, which settled the long-standing 
problem of the existence of absolutely maximally entangled states of 4 parties in local dimensions six.
 
%The reasons that random seeds do not lead to 2-unitaries in higher dimensions needs explorations.
The search for 2-unitaries in the unitary group $\mathcal{U}(d^2)$ with $d^4$ parameters can be 
viewed as an optimization problem for maximizing entangling power in a high dimensional space with a complex landscape. Random seeds can get attracted to the many local extrema which are typically saddles. This makes the search for finding global extrema increasingly hard in higher dimensions. A glimpse of the difficulties involved is discussed in Ref. \cite{Grzegorz_thesis} where details about Hessians of the entangling power and especially its maximization in $d=6$ are presented.
%which is directly related to the existence of Absolutely Maximally Entangled State of four quhexes AME($4,6$).
\subsubsection{2-unitaries as fixed points}
An action of the $\mathcal{M}_{\Gamma R}$ map on a 2-unitary $\mathsf{U}$ is given by
\[\mathcal{M}_{\Gamma R}[\mathsf{U}]=\mathsf{U}^{\Gamma R},\]
as $\mathsf{U}^{\Gamma R}:=\left(\mathsf{U}^R\right)^{\Gamma}=U^{\Gamma}S$ is also unitary.
%\footnote{This is true for any T-dual unitary operator; If $U$ is T-dual then $U^{\Gamma R}:=\left(U^R\right)^{\Gamma}$ is dual unitary.}
%The last equation follows from the fact that 
The combined rearrangement $\Gamma R$ is not an involution like $R$ or $\Gamma$, but is equivalent to the identity operation when composed thrice. Note that  $\Gamma R$ operation on the set of 4 symbols which label indices of the product basis states is $\left\lbrace{1, 2,3,4}\right\rbrace \xrightarrow{R}\left\lbrace{1, 3, 2,4}\right\rbrace \xrightarrow{\Gamma} \left\lbrace{1,4,2,3}\right\rbrace$. 
%For simplicity we write this combined action as 
%\[\left\lbrace{1, 2,3,4}\right\rbrace \xrightarrow{\Gamma R} \left\lbrace{1,4,2,3}\right\rbrace\]
%\item
Thus, iterating $\Gamma R$ thrice results in
\begin{multline*}
\left\lbrace{1,2,3,4}\right\rbrace \xrightarrow{\Gamma R} \left\lbrace{1,4,2,3}\right\rbrace \xrightarrow{\Gamma R} \left\lbrace{1,3,4,2}\right\rbrace \xrightarrow{\Gamma R} \left\lbrace{1,2,3,4}\right\rbrace.
\end{multline*}
%\end{itemize}
Therefore, 2-unitaries are period-3 fixed points of the $\mathcal{M}_{\Gamma R}$ map:
\beq
\mathcal{M}_{\Gamma R}^3[\mathsf{U}] =\mathsf{U}.
\eeq
%Numerically we observed that there are period 3 fixed points of the $\mathcal{M}_{\Gamma R}$ map other than 2-unitaries also.

\subsection{Structured dual unitaries from stochastic maps \label{Sec:spldual}}
Analytic constructions of dual unitaries are obtained by multiplying T-dual unitaries with the swap $S$. The families of T-dual unitary operators that have been analytically constructed so far, mostly have a block-diagonal structure, or are permutations (which can also be block-diagonal) \cite{ASA_2021,prosen2021many}. 
%An interesting class of dual-unitary operators is that of dual-unitary permutation matrices of size $d^2$. 
Dual-unitary permutations preserve dual-unitarity under multiplication (both left as well as right ) by arbitrary diagonal unitaries \cite{ASA_2021}. We refer to this property of dual unitary permutations as an enphasing symmetry,
and is a kind of gauge freedom enjoyed by these matrices. This symmetry is also present in the uniform block-diagonal constructions. The iterative algorithms discussed above, in general, do not lead to dual unitaries with this symmetry. In these cases, no special structure of dual is usually evident,  as illustrated in Fig.~(\ref{fig:MRmapoutput}).  
\begin{figure}
\includegraphics[scale=0.5]{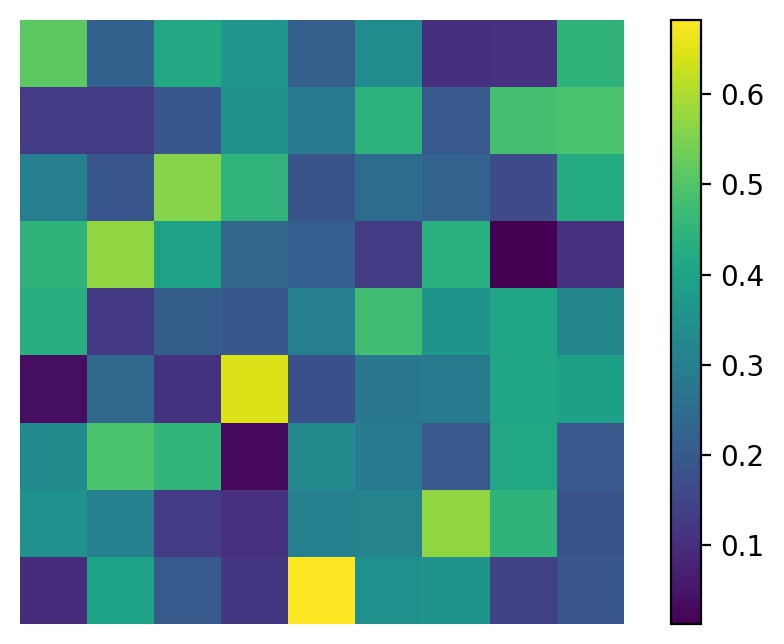}
\caption{Action of $\mathcal{M}_{R}$ map, 1000 times, on a random seed unitary of size $d^2=9$ results in an approximate dual unitary which have typically all $d^4$ entries non-zero. Shown are the absolute value of the matrix elements in one such instance.}
\label{fig:MRmapoutput}
\end{figure}

In this section we demonstrate that modified algorithms can be defined that are capable of resulting in dual and 2-unitaries with block-diagonal structures or enphased permutations, and hence afford some degree of control or design. This is achieved by incorporating in the algorithm random diagonal unitaries which preserve the dual-unitary property of structured matrices.

%We define stochastic or non-deterministic iterative algorithms which converge to dual unitaries with different possible block structures. 

One such algorithm $\mathbb{M}_{R}$ which converges to dual unitaries with the enphasing symmetry is defined as
\beq
\mathbb{M}_{R}:\;\;U_0 \xrightarrow{R} U_0^R  \xrightarrow{\text{PD}} U_1' \rightarrow U_1=D_1 \; U_1' \; D_2\;,
\label{eq:MRkickmap}
\eeq
where $D_1,\,D_2$ are diagonal unitaries with random phases. Note that the map is no longer deterministic, as $U_0$ does not uniquely determine $U_1$.  The map converges (in all cases that we have encountered for $d=2,\,3,\,4$) to dual unitaries that remain dual-unitary on multiplication by arbitrary diagonal unitaries.

Starting from a random seed unitary $U_0$, the map converges to dual unitaries with different block structures as shown in Fig.~(\ref{fig:block_d_3}) for $d^2=9$. For the sake of convenience we have shown the non-zero elements of the corresponding T-dual unitary to the dual-unitary obtained from the map. It is known that a block-diagonal unitary of size $d^2$ is T-dual if the size of each block is a multiple of $d$ \cite{ASA_2021}. The map indeed yields T-dual unitaries which are block-diagonal and size of each block is multiple of $d$ as shown in Fig.~(\ref{fig:block_d_3}) for $d^2=9$. The resulting dual unitaries are of the following form (up to multiplication by $S$):
\begin{enumerate}
\item[(i)]
$
U=\oplus_{i=1}^3 u_i,\;u_i \in \mathcal{U}(3).$
\item[(ii)]
$U=u_1 \oplus u_2,\;u_1\in \mathcal{U}(6),\,u_2\in \mathcal{U}(3).$
\end{enumerate}
Due to their peculiar structure these dual unitaries remain dual-unitary under multiplication by random diagonal unitaries. This is easy to see for the uniform block case as compared to the non-uniform case in Fig.~(\ref{fig:block_d_3}). 
In the nonuniform case, the $6 \times 6$ block cannot be replaced by an arbitrary unitary matrix. In fact $6 \times 6$ unitary matrix acting on $\mathcal{C}^2 \otimes  \mathcal{C}^3$ should satisfy T-dual unitarity \cite{ASA_2021}. If we require in addition that the duality (or T-duality equivalently) is preserved under multiplication by diagonal unitaries, a subset is picked, an example being shown in Fig.~(\ref{fig:block_d_3}). Note the peculiar structure of the $6 \times 6$ block in the non-uniform case. It consists of three $2 \times 2$ unitary matrices arranged in such a way that multiplication by arbitrary diagonal unitaries preserves T-duality.
 \begin{figure}
\includegraphics[scale=0.5]{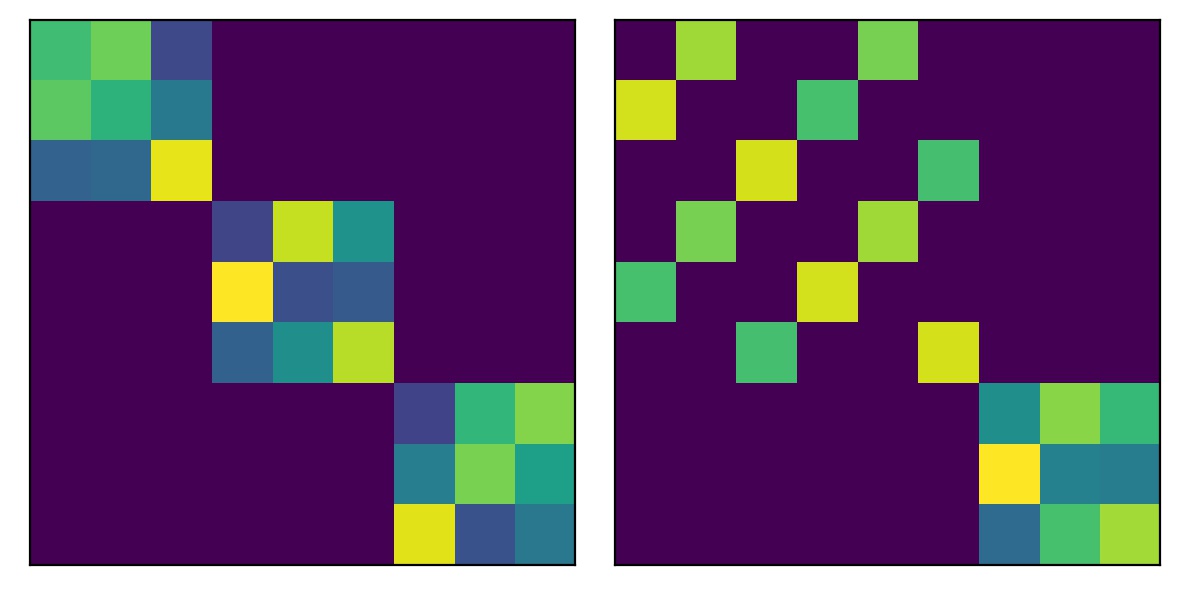}
\caption{(Color online) Structured T-dual unitaries obtained using the $\mathbb{M}_R$ map for $d=3$. (Left) T-dual unitary consisting of three blocks (unitary matrices) of size 3. (Right) T-dual unitary with a $6 \times 6$ block and a $3 \times 3$ block.}
\label{fig:block_d_3}
%\includegraphics[scale=0.4]{Figures_quant_design/M_R_map_output.png}
%\includegraphics[scale=0.4]{Figures_quant_design/M_R_kicked_map_output.png}
%\caption{(Left) Color plot of the output matrix of $M_R$ map starting from a generic initial condition. (Right) Color plot of the output matrix of the modified $M_R$ map starting from a generic initial condition.}
\end{figure}
 
We have checked that similar structured matrices are obtained for $d=4$ and $5$.  For $d^2=16$, the map yields dual unitaries which are of the following forms (up to multiplication by $S$),
\begin{enumerate}
\item[(i)]
$U=\oplus_{i=1}^4 u_i$, $u_i \in \mathcal{U}(4)$,
\item[(ii)]
$U=\oplus_{i=1}^2 u_i$, $u_i \in \mathcal{U}(8)$, and
\item[(iii)]
$U=\oplus_{i=1}^2 u_i$, $u_1 \in \mathcal{U}(4)$, $u_2 \in \mathcal{U}(12)$.
\end{enumerate}
These block structures are compatible with the analytical constructions of dual unitary operators based on block-diagonal unitaries. 

To obtain structured 2-unitaries, we define $\mathbb{M}_{\Gamma R}$ map as follows,
\beq
U_0 \xrightarrow{R} U_0^R \xrightarrow{\Gamma} \left(U_0^{R}\right)^{\Gamma}:=U_0^{\Gamma R} \xrightarrow{\text{PD}} U_1' \rightarrow U_1=D_1 \; U_1' \; D_2\;.
\label{eq:MTRkickmap}
\eeq
For $d^2=9$, it is observed that for a random seed unitary if the map converges to 2-unitary then it is a 2-unitary permutation matrix up to multiplication by diagonal unitaries as shown in Fig.~(\ref{fig:MTRkickmapoutput}). There is only one non-zero element in each row positioned in such a way that the whole arrangement of non-zero entries in a 2-unitary permutation matrix is directly related to orthogonal Latin squares which we elaborate in next sections. The map is not as efficient as it's deterministric counterpart $\mathcal{M}_{\Gamma R}$ in yielding 2-unitaries from random seed unitaries. However, it demonstrates that one can obtain structured 2-unitary operators of desired symmetry and can be used to gain insights about the most general constructions of such special unitary operators,
%Obtaining structured 2-unitary operators has implications for the support of the corresponding absolutely maximally entangled state constructed from it.
\begin{figure}
\includegraphics[scale=0.5]{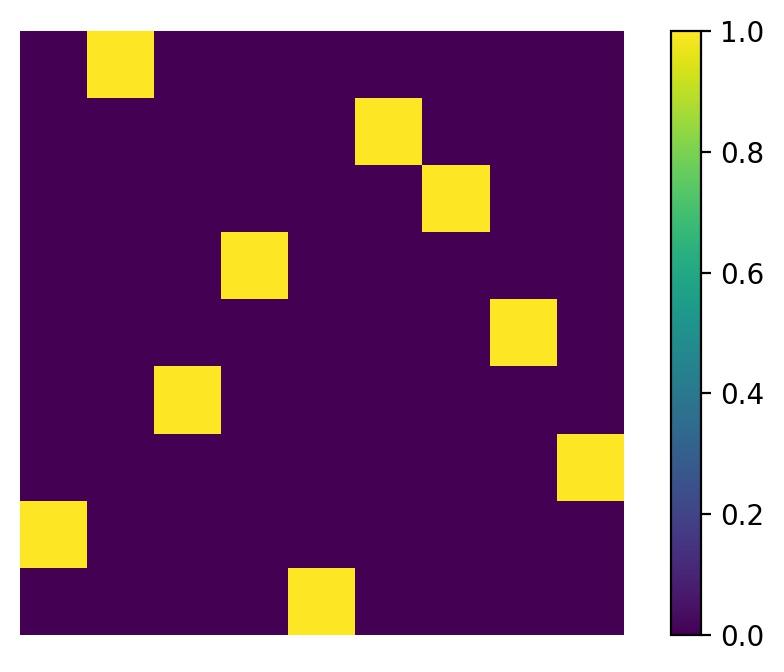}
\caption{Action of $\mathbb{M}_{\Gamma R}$ map on a random seed unitary of order 9. The map converges to 2-unitary permutation matrix (up to phases). The only non-zero element in each row or column is shown by a yellow square.}
\label{fig:MTRkickmapoutput}
\end{figure}
We have observed that multiplying at each step of the map with random, but structured, unitaries other than diagonal unitaries can also yield structured dual matrices, provided duality is preserved under such operations.

\section{Dynamical map in the two-qubit case \label{Sec:weyldynamics}}

The $\mathcal{M}_R$ map is now studied explicitly and analytically in the case of two qubit unitary operators. The nonlocal part of the operators is well-known in this case.  As the map has been shown to be covariant under local unitary transformations, see Eq.~(\ref{eq:orbitpreserve}), it is sufficient to consider its action on the nonlocal part.
The subset of dual-unitary matrices is known explicitly in this case and we can calculate the rate at which arbitrary seeds approach the dual set. We will find that those that approach the {\sc swap} gate $S$ do so algebraically slowly, while generically the approach is exponential. 

Any unitary operator in $\mathcal{U}(4)$ can be written as $(u_1 \otimes u_2) U (v_1 \otimes v_2)$, where $u_i$ and $v_i$ are single qubit unitaries in $\mathcal{U}(2)$, and 
\begin{equation}
U = \exp\left[-i \sum_{k=1}^3 c_k \,(\sigma_k \otimes \sigma_k)\right].
\label{eq:cartan}
\end{equation}  
Here $\sigma_k$ are Pauli matrices, $c_k \in \mathbb{R}$ are \textit{Cartan coefficients}, and $U$ is the nonlocal part of the canonical Cartan form \cite{KBG01,KC01,Zhang2003}. 
The so-called ``Weyl chamber" \cite{Zhang2003} is a tetrahedron formed by considering the subset of $c_k$'s,
\begin{align}
& 0 \leq |c_3| \leq c_2 \leq c_1 \leq \frac{\pi}{4} 
%\quad \text{or} \\
%& \frac{\pi}{4} < c_1 \leq \frac{\pi}{2}, 0 \leq c_3 \leq c_2 %\leq \frac{\pi}{2} - c_1 ,
\label{eq:weyl}
\end{align}
The $c_k$'s in the Weyl chamber which uniquely identify local unitarily inequivalent gates, is also termed as the gate's information content \cite{Kus2013}. The nonlocal part of $U$ we refer below simply as the Cartan form.  
For two qubit dual unitaries \cite{Bertini2019},
\beq
c_1=c_2=\frac{\pi}{4},\;\; c_3 \in \left[0, \frac{\pi}{4} \right],
\label{eq:Cartan_dualrange}
\eeq
and provides the complete parametrization of the nonlocal part. An equivalent parametrization is not known in higher dimensions.

\subsection{$\mathcal{M}_R$ map in the Weyl chamber} 
While the map has been defined on general unitary matrices, 
the overall phase has no impact on entanglement and the map can be 
defined as an action on $\mathcal{SU}(4)$, with $\det(U)=1$, to itself
by removing the phase at each step. This turns out to be very useful
for the qubit case.

Consider a seed unitary in Cartan form as  

\begin{align}
&U_0= \exp\left[-i \sum_{k=1}^3 c_k^{(0)} \,(\sigma_k \otimes \sigma_k)\right] = \\
     &    \left( \begin{array}{cccc}
    e^{-ic_3^{(0)}}c_-^{(0)}       & 0 & 0  & -ie^{-ic_3^{(0)}}s_-^{(0)} \\
    0       & e^{ic_3^{(0)}}c_+^{(0)} & -ie^{ic_3^{(0)}}s_+^{(0)} & 0 \\
    0       & -ie^{ic_3^{(0)}}s_+^{(0)} & e^{ic_3^{(0)}}c_+^{(0)} & 0 \\
    -ie^{-ic_3^{(0)}}s_-^{(0)}       & 0 & 0  & e^{-ic_3^{(0)}}c_-^{(0)}  \nonumber
\end{array} \right)
\label{eq:Carmatrix}
\end{align}
where,
\beq
\begin{split}
c_\pm^{(n)} &= \cos(c_1^{(n)}\pm c_2^{(n)}),\\
s_\pm^{(n)}&= \sin(c_1^{(n)}\pm c_2^{(n)}),\; n=0,1,\cdots.
\end{split}
\label{eq:candsplusminus}
\eeq

Note that $U_0 \in \mathcal{SU}(4)$ and we would like the subsequent iterations to also satisfy this property: it also becomes easy to identify the Cartan coefficients $c_i$ at every step.
A crucial property of the map is that it preserves the matrix form of $U_0$ such that $U_1$ has exactly the same structure; see appendix (\ref{app:map}). 

Let 
\begin{equation}
    U_n = 
%     \frac{1}{(\alpha_n^2-\beta_n^2)(\delta_n^2-\gamma_n^2)}   
     \left( \begin{array}{cccc}
    \alpha_n & 0 & 0 & \beta_n \\
    0 & \delta_n & \gamma_n & 0 \\
    0 & \gamma_n & \delta_n & 0 \\
    \beta_n & 0 & 0 & \alpha_n 
    \end{array} \right)\; \in \;\mathcal{SU}(4)\,,
\label{eq:carabc0}
\end{equation}
where
\begin{equation} \begin{split} & \alpha_n = e^{-ic_3^{(n)}}c_-^{(n)},\, \beta_n = -ie^{-ic_3^{(n)}}s_-^{(n)},\\ & \gamma_n = -ie^{ic_3^{(n)}}s_+^{(n)},\, \delta_n = e^{ic_3^{(n)}}c_+^{(n)},
\end{split}
\label{eq:alphadefn}
\end{equation}

% by 
%\begin{equation}
%    U_1 =     \left( \begin{array}{cccc}
%    a_1 & 0 & 0 & b_1 \\
%    0 & d_1 & c_1 & 0 \\
%    0 & c_1 & d_1 & 0 \\
%    b_1 & 0 & 0 & a_1 
%    \end{array} \right).
%\label{eq:carabc1}
%\end{equation}
The mapping between matrix elements of $U_{n+1}$ and $U_n$ is given by  
\begin{equation}
\begin{split}
\alpha_{n+1}=&\frac{e^{-i\,\frac{\chi_{n+1}}{4}}}{2}\left[ \frac{\alpha_n+\delta_n}{|\alpha_n+\delta_n|}+\frac{\alpha_n-\delta_n}{|\alpha_n-\delta_n|}\right]\\
\beta_{n+1}=&\frac{e^{-i\,\frac{\chi_{n+1}}{4}}}{2}\left[ \frac{\alpha_n+\delta_n}{|\alpha_n+\delta_n|}-\frac{\alpha_n-\delta_n}{|\alpha_n-\delta_n|}\right]\\
\gamma_{n+1}=&\frac{e^{-i\,\frac{\chi_{n+1}}{4}}}{2}\left[ \frac{\beta_n+\gamma_n}{|\beta_n+\gamma_n|} -  \frac{\beta_n-\gamma_n}{|\beta_n-\gamma_n|}\right]\\
\delta_{n+1}=&\frac{e^{-i\,\frac{\chi_{n+1}}{4}}}{2}\left[ \frac{\beta_n+\gamma_n}{|\beta_n+\gamma_n|} +  \frac{\beta_n-\gamma_n}{|\beta_n-\gamma_n|}\right],
\end{split}
\label{eq:abcdmap}
\end{equation}
where
\beq
\chi_{n+1}=\text{Arg}[(\alpha_n^2-\delta_n^2)(\beta_n^2-\gamma_n^2)].
\label{eq:chi}
\eeq

%\begin{equation}
%\begin{pmatrix}
%\alpha_1 \\
%\beta_1 \\
%\gamma_1 \\
%\delta_1 
%\end{pmatrix} = 
%\begin{pmatrix}
%\alpha_+ & 0 & 0 & \alpha_- \\
%\alpha_- & 0 & 0 & \alpha_+ \\
%0 & \beta_- & \beta_+ & 0 \\
%0 & \beta_+ & \beta_- & 0 
%\end{pmatrix} 
%\begin{pmatrix}
%\alpha_0 \\
%\beta_0 \\
%\gamma_0 \\
%\delta_0 
%\end{pmatrix},
%\label{eq:abcdmap}
%\end{equation} 
%where 
%\begin{equation} 
%\begin{split} 
%\alpha_{\pm} &=  \frac{|\alpha_0-\delta_0| \pm |\alpha_0+\delta_0|}{2|\alpha_0-\delta_0||\alpha_0+\delta_0|} \\
%\beta_{\pm} &=  \frac{|\beta_0-\gamma_0| \pm |\beta_0+\gamma_0|}{2|\beta_0-\gamma_0||\beta_0+\gamma_0|}
%\end{split}.
%\end{equation}
The dynamical system is thus a 4-dimensional complex map on the
manifold $S^4 \times S^4$. There are constraints originating from the unitarity condition:  $|\alpha_i|^2+|\beta_i|^2=1$, $|\gamma_i|^2+|\delta_i|^2=1$, $\mbox{Re} (\alpha_i \beta_i^*)=0$ and  $\mbox{Re}( \gamma_i \delta_i^*)=0$, and the $\mathcal{SU}$ condition: $(\alpha_n^2-\beta_n^2)(\delta_n^2-\gamma_n^2)=1$. 
The nonlinear nature of the map is clear as the entries
of the above transformation are themselves functions of other variables.

Rather than the high-dimensional complex map in Eq.~(\ref{eq:abcdmap}), using Eq.~(\ref{eq:alphadefn}) one obtains a 3-dimensional real map in terms of the Cartan coefficients. 
%{\blu As $\left\lbrace\sigma_k \otimes \sigma_k\right\rbrace_{k=1}^3$ commute with each other these can be simultaneously diagonalised. Cartan coefficients of $U=\exp(-i \sum_k c_k \sigma_k \otimes \sigma_k)$ are related to the eigevalues of the Hermitian matrix $\sum_k c_k \sigma_k \otimes \sigma_k$ and a correct procedure to extract Cartan coefficients is explained in detail in Ref.~\cite{Childs_2003}}. 
Defining $\theta_{\pm}^{(n)}=\text{Arg}(\alpha_n \pm \delta_n)$ and $\phi_{\pm}^{(n)}=\text{Arg}(\beta_n \pm \gamma_n)$, the complex map in Eq.~(\ref{eq:abcdmap}) simplifies to
\beq
\begin{split}
&c_1^{(n+1)}  = \frac{1}{4}(-\theta_{+}^{(n)}+\theta_{-}^{(n)}-\phi_{+}^{(n)}+\phi_{-}^{(n)}),\\
&c_2^{(n+1)} = \frac{1}{4}(\theta_{+}^{(n)}-\theta_{-}^{(n)}-\phi_{+}^{(n)}+\phi_{-}^{(n)}),\\
&c_3^{(n+1)}  =\frac{1}{4}(-\theta_{+}^{(n)}-\theta_{-}^{(n)}+\phi_{+}^{(n)}+\phi_{-}^{(n)}).
\label{eq:cartan_gen2}
\end{split}
\eeq
Numerically it is observed that the Cartan coefficients of $U_{n+1}=\mathcal{M}_R[U_n]$ obtained from the above 3-dimensional map agree for all even $n$ with those calculated using the numerical algorithm presented in Refs.~\cite{Childs_2003,Kus2013} and also satisfy Eq.~(\ref{eq:weyl}). However for odd $n$ although $c_3^{(n)}$ values still agree but $c_1^{(n+1)}$ and $c_2^{(n+1)}$ values differ from the numerical value  by $\pi/2$. In order to obtain the desired Cartan coefficients satisfying Eq.~(\ref{eq:weyl}) from the above 3-dimensional map, one needs to replace $c_2^{(n+1)}$ by $\pi/2-c_2^{(n+1)}$ for all odd $n$.

%Cartan coefficients obtained from the above set of equations are identical to those calculated using the algorithm presented in Ref.~\cite{Kus2013} to find the Cartan coefficients of any $U \in \mathcal{U}(4)$.

Before we simplify the 3-dimensional map given by Eq.~(\ref{eq:cartan_gen2}), we first prove the following two theorems.

\begin{thm}
For two-qubit gates of the form Eq.~(\ref{eq:cartan}) self-dual unitaries ($U=U^R$) are the only fixed points of the $\mathcal{M}_{R}$ map.
\end{thm}

\begin{thm}
For two-qubit gates of the form Eq.~(\ref{eq:cartan}) dual unitaries are the only fixed points of the $\mathcal{M}_{R}^2$ map i.e., $\mathcal{M}_{R}^2[U_0]=U_0$ iff $U_0$ is a dual-unitary.
\end{thm}

Proofs of the above theorems are presented in appendix \ref{app:fixed_point_thms}.

Consider a two-qubit seed unitary $U_0$ parametrized by the Cartan parameters $c_1^{(0)},c_2^{(0)},c_3^{(0)}$. Under the $\mathcal{M}_R$ map $U_0$ is mapped to $U_1$ which is parametrized by $c_1^{(1)},c_2^{(1)},c_3^{(1)}$. The $\mathcal{M}_R$ map can be viewed most economically as a 3-dimensional dynamical map on the Cartan parameters,
\begin{equation} 
\begin{split} 
U_n & \xrightarrow{\mathcal{M}_R} U_{n+1},\\
(c_1^{(n)},c_2^{(n)},c_3^{(n)}) & \xrightarrow{\mathcal{M}_R} (c_1^{(n+1)},c_2^{(n+1)},c_3^{(n+1)}).
\end{split}
\end{equation}

\subsection{Deriving the map for special initial conditions}
%We will first derive the map for some special seed unitaries and then make some quantitative arguments about the map in general case.
Although we are unable to derive explicit maps in terms of these parameters for general $c_i^{(0)}$, we are able to do so for special values. We show that these converge to the desired fixed points, $c_1^{(n\rightarrow \infty )}=\pi/4,\;c_2^{(n\rightarrow \infty )}=\pi/4\,,$
and $c_3^{(n\rightarrow \infty )} \in [0,\pi/4]$, which is the set of dual-unitary operators. This is depicted in Fig.~(\ref{fig:weylMRmap}) for few random realizations evolved under the map for $n=10$ steps. 
\begin{figure}
\includegraphics[scale=0.5]{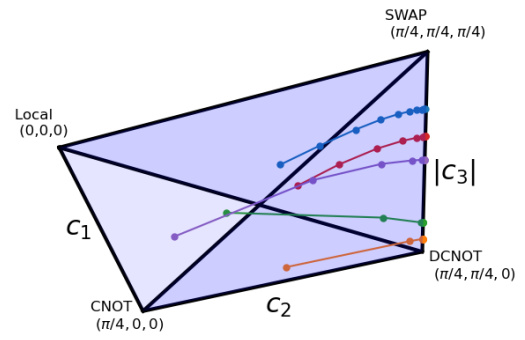}
\caption{Trajectories of five random realizations of two-qubit gates are shown inside the Weyl chamber under action of the map for $n=10$ steps. The edge joining the {\sc swap} gate and the {\sc dcnot} gate corresponds to dual unitaries to which the map converges.}
\label{fig:weylMRmap}
\end{figure}
For the general case we argue why this happens and also derive the rate of exponential approach. 

\subsubsection{XY family: plane $c_3=0$}

The first special case is when $c_3^{(0)}=0$ and $0<c_2^{(0)}\leq c_1^{(0)}$.
In this case, using Eq.~(\ref{eq:abcdmap}), we can see that a {\sl single} application of the $\mathcal{M}_{R}$ map results in the following unitary:
\[U_1=\left(
\begin{array}{cccc}
1 & 0 & 0 & 0\\
0 & 0 & -i &  0 \\
0 & -i  & 0 &  0 \\
 0 & 0 & 0 & 1
\end{array}
\right),
\]
with $\det(U_1)=\det(U_0)=1$. Thus the map preserves the $\mathcal{SU}$ property of seed unitaries. Cartan parameters for $U_1$ are: $c_1^{(1)}=c_2^{(1)}=\pi/4,\,c_3^{(1)}=0$, and therefore $U_1$ is dual-unitary. This gate is LU equivalent to the gate {\sc dcnot} \cite{Kus2013} which is $S \times \mbox{{\sc cnot}}$. Explicitly,
\begin{align}
U_{\text{DCNOT}}& =(H \otimes I)\,U_1 \,(D_1 \otimes D_1\, H)  \nonumber \\
&=\begin{pmatrix}
1 & 0 & 0 & 0 \\
0 & 0 & 0 & 1 \\
0 & 1 & 0 & 0 \\
0 & 0 & 1 & 0   
\end{pmatrix},
\end{align}
 where $H=\frac{1}{\sqrt{2}}\begin{pmatrix}
1 & 1 \\
1 & -1 
\end{pmatrix}$ is the Hadamard gate and $D_1=P_{\frac{\pi}{2}}=\begin{pmatrix}
1 & 0 \\
0 & i 
\end{pmatrix}$ is a phase gate.

Thus the entire interior of the base of the Weyl chamber, $c_3=0$ plane, is mapped to the same dual-unitary gate $U_1$ in just one step and the rate at which it happens is infinite. 
\subsubsection{XXX family: $c_1=c_2=c_3$} 
Let $c_1=c_2=c_3 = c \in [0,\pi/4]$ in Eq.~(\ref{eq:cartan}), the single parameter family of unitary operators $U$,
\beq
U=\exp(-i\,c\sum_{i=1}^3 \sigma_i \otimes \sigma_i). 
\label{eq:xxxuni}
\eeq
This forms an edge of the Weyl chamber, the one that connects local unitaries to the {\sc swap} gate $S$.
Unitaries of this form are useful in many contexts such as in the trotterization of integrable isotropic (XXX) Heisenberg Hamiltonian \cite{Vanicat_2018}. They are also, modulo phases, the fractional powers of the {\sc swap} gate $S$ as $U = \exp \left(-2\,i\,c\, S\right)$.
 
If we choose the seed unitary $U_0$ from this family with $c=c^{(0)}$, it follows from Eq.~(\ref{eq:alphadefn}) that $\beta_0=0$. Action of the map on $U_0$ gives
%From Eq.~(\ref{eq:abcdmap}), the action of $\mathcal{M}_R$ map on intial unitary $U_0$  is  
\beq
U_1=\begin{pmatrix}
\alpha_1 & 0 & 0 & \beta_1 \\
0 & 0 & \gamma_1 & 0 \\
0 & \gamma_1 & 0 & 0 \\
\beta_1 & 0 & 0 & \alpha_1
\end{pmatrix} ,
% \overset{\text{LU}}{\sim} \begin{pmatrix}
%\alpha_1 & 0 & 0 & 0 \\
%0 & \delta_1 & \gamma_1 & 0 \\
%0 & \gamma_1 & \delta_1 & 0 \\
%0 & 0 & 0 & \alpha_1 
%\end{pmatrix}
\eeq

for which $\delta_1=0$. Note that $U_1$ is not exactly of the same form as $U_0$ for which $\beta_0=0$. In fact for all even (odd) $n$, $U_n$ is such that $\beta_n=0$ ($\delta_n=0$). For even $n$, $\beta_n=0$ implies $c_1{(n)}=c_2{(n)}$ both being equal to $c_3^{(n)}$ and thus $U_n$ belongs to the same family. However, for odd $n$ it is observed that although $c_2^{(n)}=c_3^{(n)} \leq \pi/4$, $c_1^{(1)}=\pi/2-c_2^{(1)} \geq \pi/4$ and thus $c_1^{(1)}$ does not satisfy Eq.~(\ref{eq:weyl}). 
%However a local unitarily transformed one given by $-i\,(\sigma_x \otimes \sigma_x)U_n$ for odd $n$ will have $c_1^{(n)}=c_2^{(n)}=|c_3^{(n)}| \leq \pi/4$. 
Note that $U_1$ with Cartan coefficients $c_1^{(1)}=\pi/2-c^{(1)},c_2^{(n)}=c_3^{(1)}=c^{(1)}$ is {\em not} LU equivalent to a gate with Cartan coefficients $c_1^{(n)}=c_2^{(n)}=c_3^{(1)}=c^{(1)}$, although in the part of the Weyl chamber we are restricted attention to, they are the same points.

%where $\alpha_1= -i\exp(i\,c^{(1)})\sin(2\,c^{(1)})$, $\gamma_1=\exp(-i\,c^{(1)}) $, and $\beta_1=\exp(i\,c^{(1)})\cos(2\,c^{(1)})$.  
%Thus taking seed unitary $U_0$ from XXX family, action of $\mathcal{M}_R$ map on $U_0$ is such that for even $n$ $U_n$ also belongs to XXX family and at odd $n$ it belongs to XXX family up to a local unitary transformation. 
As a consequence of this, the 3-dimensional map given by Eq.~(\ref{eq:cartan_gen2}) becomes a 1-dimensional map.  
%it maps $U_0$ to $U_1$ which is also in the same family. This may not be evident from the form of $U_1$ but a locally transformed one $(\mathbb{I} \otimes \sigma_x) U_1 (\sigma_x \otimes \mathbb{I})$ is exactly of the same form, where $\sigma_x=\begin{pmatrix}
%0 & 1 \\
%1 & 0
%\end{pmatrix}$. This property of the map is easy to see if one considers the $\mathcal{M}_R^2$ map, as $\beta_2=0$. In terms of Cartan parameters the map is 1-dimensional and all 3 Cartan parameters being converge to $\pi/4$ corresponding to the {\sc swap} gate.

Let $c^{(n)}$ be the Cartan coefficient parametrizing $U_n=\mathcal{M}_R^n[U_0]$, and
% \beq
% c^{(n+1)}=\frac{\pi}{4}-\frac{1}{2}\arctan\left[\frac{2\,\cos(2\,c^{(n)})}{\sin(2\,c^{(n)})+\sqrt{3\cos ^2(2\,c_n^{(n)})+1}}\right].
% \label{eq:XXXcnmap}
% \eeq
%It is easy to check that $c^{(n)}=\pi/4$ is the fixed point of the map. 
%we will show that it is global attractor for all $c^{(n)} \in (0,\pi/4]$. 
%For the sake of convenience we define, 
\beq
x_n=1/\tan(2\,c^{(n)}).
%x_n=\tan\left[\frac{\pi}{2}-2\,c^{(n)}\right].
\label{eq:xndef}
\eeq
The complications attendant on the ranges of $c_i$ do not affect this variable. In terms of $x_n$ (for a derivation see App.~(\ref{app:CartanAlg})), the map takes a simple algebraic form, 
\beq
x_{n+1}=\frac{2\,x_n}{1+\sqrt{4\,x_n^2+1}}.
\label{eq:XXXunmap}
\eeq
The unique fixed point of the map is $x^*=0$ corresponding to the {\sc swap} gate, and the map is a contraction as shown in the  App.~(\ref{app:CartanAlg}).
Therefore, in the limit of large $n$, $x_n \rightarrow x^*=0$. In this limit, Eq.~(\ref{eq:XXXunmap}) can be approximated as
\beq 
x_{n+1} \approx x_n(1-x_n^2).
\label{eq:u_nlarge}
\eeq
Thus in the vicinity of the fixed point, the difference equation may be approximated by the differential equation
%\beq
$dx_n/dn=-x_n^3$.
%\label{eq:u_ndiffeqn}
%\eeq
This is simple to solve and gives the large $n$ approximation to the map above as \beq
x_n\approx 1/\sqrt{2n}.
%\frac{1}{\sqrt{2n}}.
\eeq
%{\red AL: Figure to be  moved to Appendix, SAR:Done}

\subsubsection{{\sc swap-cnot-dcnot} face; $c_1^{(0)}=\pi/4$}
In this case seed unitaries lie on the {\sc swap-cnot-dcnot} face of the Weyl chamber with $c_1^{(0)}=\pi/4$. Under the action of the map $c_1^{(n)}=\pi/4$ for all $n$ and thus the corresponding map is 2-dimesional defined in terms of $c_2^{(n)}$ and $c_3^{(n)}$. An important property of the map observed in this case, which follows as $c_{\pm}=s_{\mp}$, is that 
the phase in Eq.~(\ref{eq:abcdmap}) $\chi_{n+1}=0$.
%it preserves the $\mathcal{SU}$ property of seed unitaries {\it i.e.,} if $U_0 \in \mathcal{SU}(4)$ then $\mathcal{M}_R[U_0]:=U_1\,\in \mathcal{SU}(4)$. 
This property is crucial for simplifying the map  as shown below.

Defining $y_n=1/\tan^2(2\,c_2^{(n)})$ and $z_n=1/\tan^2(2\,c_3^{(n)})$, the corresponding 2-dimensional map takes a purely algebraic form given by (for a derivation see App.~(\ref{app:CartanAlg}))
%\beq
%\begin{split}
%y_{n+1}' & =\frac{y_n'}{\sqrt{1+z_n'^2}},\\
%z_{n+1}' & =\frac{z_n'}{\sqrt{1+y_n'^2}}.
%\end{split}
%\eeq
%In terms of $y_n=y_n'^2$ and $z_n=z_n'^2$, the above further simplifies to 
\beq
\begin{split}
y_{n+1} & =\frac{y_n}{1+z_n},\\
z_{n+1} & =\frac{z_n}{1+y_n}.
\label{eq:xnyn_alg_form}
\end{split}
\eeq
Although the above map has a symmetric form, due to the specific choice of Cartan parameters Eq.~(\ref{eq:weyl}), the symmetry is broken and the fixed points are $y^*=0$ (or $c_2^{(\infty)}=\pi/4$) and $z^* \in [0,\infty]$ (or $c_3^{(\infty)}\in [0,\pi/4]$) corresponding to the set of dual unitaries. This 2-dimensional map can be solved analytically by noting that
\beq
\Omega=\frac{1+y_{n+1}}{1+z_{n+1}}=\frac{1+y_{n}}{1+z_{n}}
\label{eq:flow_invariant}
\eeq
is an invariant. It's value is determined by the initial conditions as:
\beq 
\Omega=\frac{1+y_0}{1+x_0}=\left(\frac{\sin(2\,c_3^{(0)})}{\sin(2\,c_2^{(0)})}\right)^2 < 1 
\label{eq:Omega}
\eeq 
for $c_3^{(0)}<c_2^{(0)}$.

 Using this to eliminate $z_n$, we have the 1-dimensional map:
\beq
\begin{split}
y_{n+1}= \Omega \frac{y_n}{1+y_n},
\end{split}
\eeq
which has the exact solution 
\beq
y_n= \frac{\Omega^n\,y_0}{1+\left(\dfrac{1-\Omega^n}{1-\Omega}\right)\,y_0},
\label{eq:exact_exp_sol}
\eeq
and implies that
\beq
z_n= \frac{z_0}{\Omega^n+\left(\dfrac{1-\Omega^n}{1-\Omega}\right)\,\Omega\, z_0}.
%\frac{1}{y_n}=\Omega^{-n}\left[\frac{1}{y_0}+\frac{1}{1-\Omega}\right]-\frac{1}{1-\Omega}\, ,\\
%\frac{1}{z_n}=\Omega^{n}\left[\frac{1}{z_0}-\frac{\Omega}{1-\Omega}\right]+\frac{\Omega}{1-\Omega}, \, \Omega \neq1.
\eeq
It follows from Eq.~(\ref{eq:exact_exp_sol}) that $y_{\infty}=0$ and $z_{\infty}=\frac{1}{\Omega}-1$ respectively. Also $c_3^{(\infty)}$ which parametrizes the dual unitary to which the map converges can be written explicity in terms of the initial pair $(c_2^{(0)},c_3^{(0)})$ as 
\beq
\begin{split}
c_3^{(\infty)}&=\frac{1}{2}\arctan\left[\sqrt{\frac{\Omega}{1-\Omega}}\right]\\
&=\frac{1}{2}\arctan\left[\frac{\sin(2\,c_3^{(0)})}{\sqrt{\sin^2(2\,c_2^{(0)})-\sin^2(2\,c_3^{(0)})}}\right].
\end{split}
\eeq

Defining $\Delta c_i^{(n)}=c_i^{(\infty)}-c_i^{(n)}$ where $c_1^{(\infty)}=c_2^{(\infty)}=\pi/4$. Note that 
$\Omega=\sin^2(2\,c_3^{(\infty)})$, governs the exponential 
approach to the duals. From the explicit and full solution in Eq.~(\ref{eq:exact_exp_sol}) it follows that
\beq
\begin{split}
\Delta c_2^{(n)} & \sim |\sin 2\,c_3^{(\infty)}|^n,\\
\Delta c_3^{(n)} & \sim |\sin 2\,c_3^{(\infty)}|^{2\,n}.
\label{eq:decay_eqns}
\end{split}
\eeq
We will see below that these continue to hold for the general case as well.
%{\red S. A. R: The following family can be merged in the previous case as it is just a special case of it?}

The marginal case $\Omega =1$ corresponds to seed unitaries on the {\sc swap-cnot} edge with $c_2^{(0)}=c_3^{(0)}$ and is dealt separately below.
\subsubsection{{\sc swap-cnot} edge}
For these gates $c_1^{(0)}=\pi/4,\, c_2^{(0)}=c_3^{(0)}$ and is a special case of the face just discussed. In this case, the 2-dimensional map Eq.~(\ref{eq:xnyn_alg_form})  degenerates to a 1-dimensional map given by
\beq
y_{n+1}=\frac{y_n}{1+y_n},
\label{eq:XYYunmap}
\eeq
with $\Omega=1$. This map also can be solved analytically and the solution is given by
\beq
y_n=\frac{y_0}{\sqrt{n\,y_0^2+1}}.
\label{eq:XYYsol}
\eeq
%Consider seed two-qubit gates with Cartan parameters: $c_1^{(0)} =\pi/4$, $c_2^{(0)} = c_3^{(0)} = c^{(0)}$. At any time step $n$, $c_1^{(n)}=\pi/4$, and $c_2^{(n)}=c_3^{(n)}=c^{(n)}$. This refers to the gates on the edge of the Weyl chamber that connects the {\sc cnot} gate to the {\sc swap}, 
%see Fig.~(\ref{fig:weylMRmap}). The map $\mathcal{M}_R$ takes gates on this edge onto itself. 
%
%%\begin{widetext} 
%%\beq
%%c^{(n+1)}=\frac{\pi}{8}+\frac{c^{(n)}}{2}-\frac{1}{2}\arctan\left\lbrace\frac{\sin(2\,c^{(n)})\left[\cos(c^{(n)})-\sin(c^{(n)})\right]}{\cos(c^{(n)})+\sin(c^{(n)})+\cos(2\,c^{(n)})\left[\cos(c^{(n)})-\sin(c^{(n)})\right]}\right\rbrace. 
%%\label{eq:XYYmodel}
%%\eeq
%%\end{widetext} 
%%{\red A.L: shift large eqs to appendix. S.A.R: Moved}
%
%%It can be checked that $c^{(n)}=\pi/4$ is the fixed point. 
%In terms of $x_n$, defined in the same way as in Eq.~(\ref{eq:xndef}), the map takes a remarkably simple form:
%\beq
%x_{n+1}=\frac{x_n}{\sqrt{x_n^2+1}}\,,
%\label{eq:XYYunmap}
%\eeq
%which yields the exact solution 
%\beq
%x_n=\frac{x_0}{\sqrt{n\,x_0^2+1}}.
%\label{eq:XYYsol}
%\eeq
The approach to the unique fixed point $y^*=0$ is algebraic in contrast to other gates on the {\sc swap-cnot-dcnot} face and goes as  $\sim 1/\sqrt{n}$. Thus the {\sc swap} gate is approached slowly along both the edges that connect it in the Weyl chamber from the locals or from the {\sc cnot} gates. The other edge is the dual-unitary edge that is already a line of fixed points. In fact the entire face of the Weyl chamber containing locals-{\sc swap}-{\sc cnot} is mapped into itself and all initial conditions on this approach the dual-unitary {\sc swap} gate algebraically. This face is characterized by two of the Cartan coefficients being equal, namely $c_2^{(n)}=c_3^{(n)} \equiv c^{(n)}$. In the limit of large $n$, $\Delta c_1^{(n)}=\pi/4-c_1^{(n)} \sim 1/n$ while $\Delta c_2^{(n)}=\Delta c_3^{(n)} =\pi/4-c^{(n)} \sim 1/\sqrt{n}$.

% as shown in Fig.~(\ref{fig:XYYMRmap}).
%\begin{figure}
%\includegraphics[scale=.45]{Figures_quant_design/XYY_M_R_map.jpg}
%\caption{XYY case: (Main) Convergence to the respective fixed point $x_n=0$ under the map initiated with $c^{(0)}=\frac{\pi}{16}$. (Inset) Power law behaviour for the same initial condition as seen from the convergence $x_n \sqrt{n}\rightarrow 1 $ for large $n$.}
%\label{fig:XYYMRmap}
%\end{figure}

\subsubsection{XXZ family: $c_1=c_2$}

Let us consider now a family of two-qubit gates for which $c_1^{(0)}=c_2^{(0)}=c^{(0)}\in (0,\frac{\pi}{4}]$, and $c_3^{(0)} \leq c^{(0)} \in [0,\frac{\pi}{4}]$. This restricts the seed unitaries to the face of the Weyl chamber that contains locals-{\sc swap}-{\sc dcnot}. Under the action of the map the unitaries remain on this face for even $n$ and up to a local unitary transformation for odd $n$. 
%As the face contains the dual-unitary edge {\sc swap}-{\sc dcnot}, the map now
%converges to other dual-unitaries than the {\sc swap}, and do so exponentially fast.
%In this case under the action of the map $c_1^{(n)}=c_2^{(n)}=c_{1,2}^{(n)} \in (0,\frac{\pi}{4}]$, and $c_3^{(n)} \leq c_{1,2}^{(n)} \in (0,\frac{\pi}{4}]$. 

The map is 2-dimensional defined on $c_1^{(n)}=c_2^{(n)}=c^{(n)} \leq \pi/4$ and $c_3^{(n)}\leq c^{(n)}$ is given by
 
\begin{align}
\label{eq:XXZcartanmap1}
c^{(n+1)} & = \frac{\pi}{4}-\frac{1}{4}\arctan\left\lbrace\frac{1}{2} \sin(2\,c_3^{(n)})\left[\frac{1}{\tan^2(c^{(n)})}-\tan^2(	c^{(n)})\right]\right\rbrace,\\
\label{eq:XXZcartanmap2}
c_3^{(n+1)} & = \frac{c_3^{(n)}}{2}+\frac{1}{4}\arctan\left\lbrace\frac{1}{2} \tan(2\,c_3^{(n)})\left[\frac{1}{\tan^2(c^{(n)})}+\tan^2(	c^{(n)})\right]\right\rbrace,
\end{align}

%In contrast to the cases considered above, there are now an
%infinity of fixed points.
The fixed points consist of $c^*=\pi/4$ and  $c_{3}^*$ can take any value in $[0,\frac{\pi}{4}]$ which is a line of fixed points corresponding to two-qubit dual unitaries.
It is not hard to see that these are the only fixed points of the map. 
%It is be noted that for odd $n$, $c_1^{(n)}\geq \pi/4$ and in order to satisfy Eq.~(\ref{eq:weyl}) one needs to multiply the unitary $U_n$ with $c_1^{(n)}>\pi/4$ by local unitary transformation $-i\,(\sigma_x\otimes \sigma_x)$ such that $-i\,(\sigma_x\otimes \sigma_x)U_n$ has $c_1^{(n)}<\pi/4$. 

The important information about the nature of the map can be obtained in the large $n$ limit, which is effectively a linear stability analysis. Defining $\Delta c_i^{(n)}=c_i^{(\infty)}-c_i^{(n)}$ where $c_1^{(\infty)}=c_2^{(\infty)}=\pi/4$. For small $\Delta c^{(n)}$, Eq.~(\ref{eq:XXZcartanmap1}) gives
\beq
\Delta c^{(n+1)} \approx \sin(2\,c_3^{(n)}) \Delta c^{(n)}.
\label{eq:expmapxxz}
\eeq
Whereas Eq.~(\ref{eq:XXZcartanmap2}) yields simply  $c_3^{(n+1)}\approx c_3^{(n)}$ to first order in $X_n$ indicating that 
it can take any value only determined by the initial condition. We denote this value as $c_3^*=c_3^{(\infty)}$. Thus the above equation is of the form,
$
\Delta c^{(n+1)}=r\Delta c^{(n)}
$, with $r=\sin(2\,c_3^{(\infty)})$ , and we get the solution:
\beq
\Delta c^{(n)}=e^{-n\xi}\Delta c^{(0)}, \; 
 \xi=|\ln r|=\left|\ln \sin(2\,c_3^{(\infty)})\right |
 \label{eq:rateX}
\eeq 
  Therefore, the convergence to the respective fixed points: $c^*=\frac{\pi}{4}$ and $c_3^*=c_3^{(\infty)} \in [0,\pi/4]$, is also exponential with the rate determined by the value $c_3^{(\infty)}$ as found in Eq.~(\ref{eq:decay_eqns}) for gates lying on the {\sc swap-cnot-dcnot} face. This is shown qualitatively in Fig.~(\ref{fig:xxzrates}). The rate $\xi=\infty$ when $c_3^{(\infty)}=0$ and the unitaries converge to the {\sc dcnot}. This is consistent with the discussion in the XY discussion above where it was shown that in this case just one step of the map is needed. The rate $\xi=0$ when $c_3^{(\infty)}=\pi/4$ when the gate attained aymptotically is the {\sc swap}. This is consistent with the discussion of the edge XXX above where an algebraic approach was obtained. 

For large $n$ the behaviour of $\Delta c_3^{(n+1)}=c_3^{(\infty)}-c_3^{(n)}$  is found by analyzing Eq.~(\ref{eq:XXZcartanmap2}) keeping the second order terms in $\Delta c^{(n)}$, and we get 
\beq
\Delta c_3^{(n+1)} =\frac{1}{2}\sin (4 c_3^{(\infty)}) (\Delta c^{(n)})^2=\frac{1}{2}\sin (4 c_3^{(\infty)})e^{-n \xi_3},
\label{eq:rateZ}
\eeq
with $\xi_3=2 \xi$, 
and hence the approach to $c_3^{(\infty)}$ is exponential at a rate that is {\em twice} that of the other Cartan parameters.

\begin{figure}[htbp]
\includegraphics[scale=0.45]{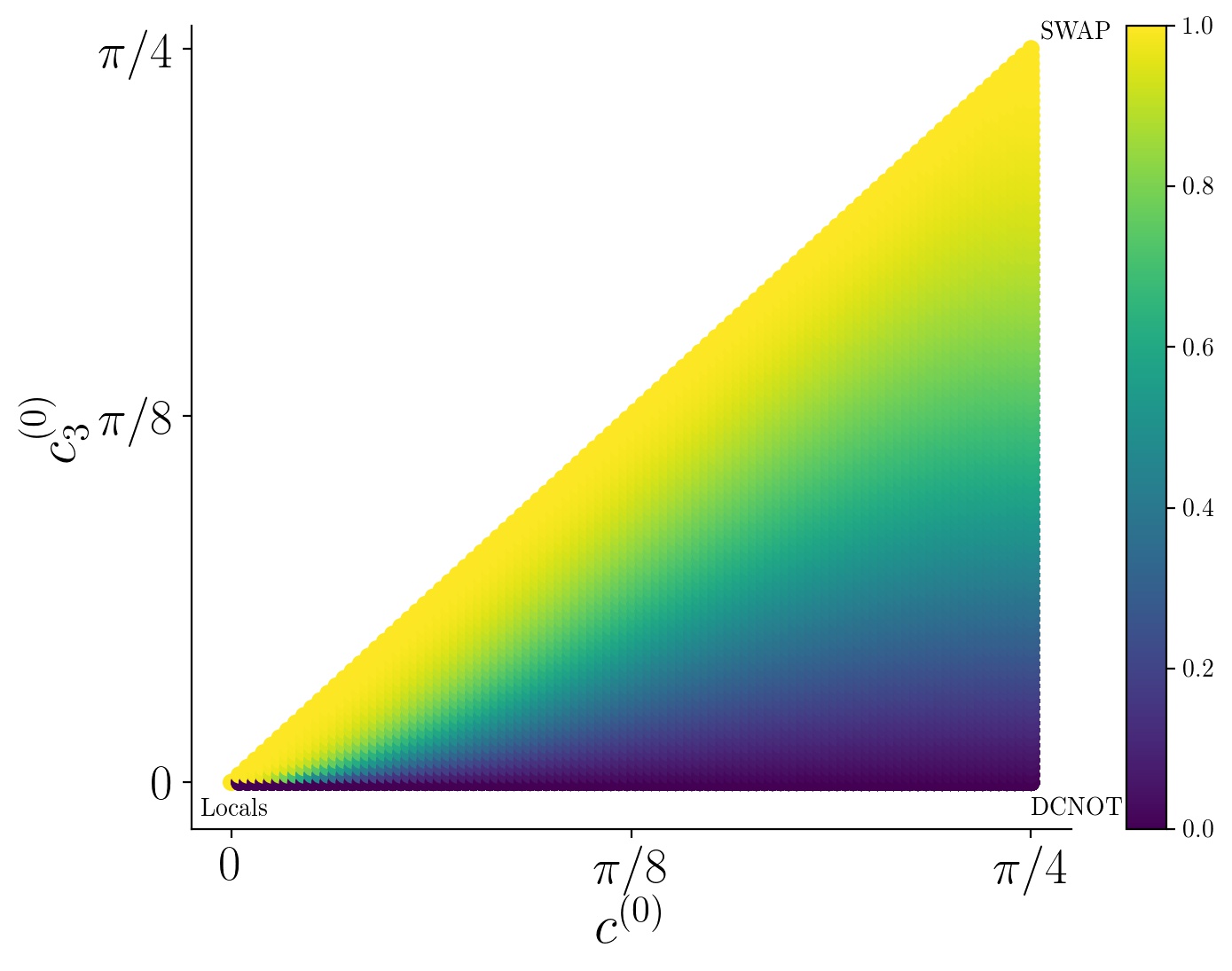}
\caption{Convergence in the XXZ case: initial condition in the local-{\sc swap}-{\sc dcnot} face. Plotted is $r=\sin{(2\,c_3^{(100)})}\approx \sin{(2\,c_3^{(\infty)})}$, which is related to the rate of convergence to a dual-unitary gate as $\xi=\ln r$. Around $10^5$ initial conditions $(c^{(0)},c_3^{(0)})$ are taken and evolved for $n=100$ times under the 2-d map Eqs.~( \ref{eq:XXZcartanmap1})--~(\ref{eq:XXZcartanmap2}). Initial conditions with $c_3^{(0)}=0$, in the base of the triangle above,  converge to the {\sc dcnot} gate at an infinite rate, while
initial conditions with $c^{(0)}=c_3^{(0)} \in (0,\pi/4]$ converge to the {\sc swap} gate at vanishing rate, namely algebraically.}
\label{fig:xxzrates}
\end{figure} 

\subsubsection{Generic initial conditions} Interestingly, numerical results indicate that the exponential approach to $(\pi/4,\pi/4,c_3^{(\infty)})$ given in Eq.~(\ref{eq:rateX}) and Eq.~(\ref{eq:rateZ}) continues to hold  for a generic initial condition inside the Weyl chamber. An illustration is  displayed in Fig.~(\ref{fig:gen_init_U}) for the initial condition $(c_1^{(0)},c_2^{(0)},c_3^{(0)})=(\pi/6,\pi/8,\pi/12)$. Under the map it converges to a dual-unitary gate with $c_3^{\infty} \approx 0.443$. The rates $\xi_1$ and $\xi_2$ at which $\Delta c_1^{(n)}= \pi/4-c_1^{(n)}$ and $\Delta c_2^{(n)}=\pi/4-c_2^{(n)}$  approach $0$ are almost the same given by  $ \xi = \left| \ln \sin(2\,c_3^{(\infty)})\right|$. The rate $\xi_3$ at which $\Delta c_3^{(n)}=c_3^{(\infty)}-c_3^{(n)} \rightarrow 0 $  continues to be a very good approximation $\xi_3 = 2 \xi$. Initial conditions which converge to dual-unitaries with large $c_3^{(\infty)}$ values {\em i.e.}, small entangling power, take longer times. This is reflected in Fig.~(\ref{fig:weylMRmap}) for random realizations where gates that are closer to the {\sc swap} gate take longer to reach the corresponding point on the dual-unitary edge.
\begin{figure}[htbp]
\includegraphics[scale=0.6]{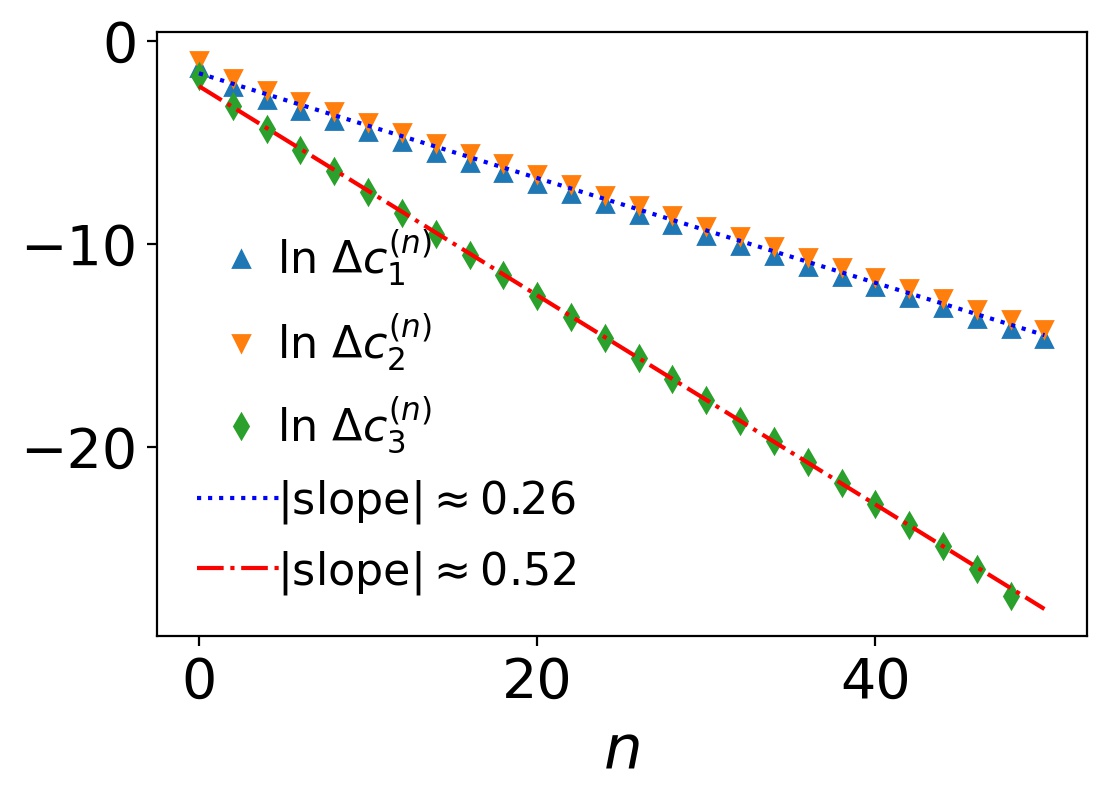}
\caption{Convergence for initial condition inside the Weyl chamber: seed unitary with $(c_1^{(0)},c_2^{(0)},c_3^{(0)})=(\pi/6,\pi/8,\pi/12)$ is evolved under the map for $n=50$ steps. The exponential rates at which $\Delta c_1^{(n)}= \pi/4-c_1^{(n)}$ and $\Delta c_2^{(n)}= \pi/4-c_2^{(n)}$ decay are almost the same $\xi =|\ln \sin(2 c_3^{(\infty)})|$  determined by $c_3^{(\infty)} \approx 0.443$. Rate at which $\Delta c_3^{(n)}= c_3^{(\infty)}-c_3^{(n)} \rightarrow 0$ is almost twice than that of $\Delta c_1^{(n)}$ or $\Delta c_2^{(n)}$. The numerically calculated slopes match with these values.}
\label{fig:gen_init_U}
\end{figure}

%{\red \em Qualitative features of the map for generic initial conditions:}

%\begin{figure}[htbp]
%\includegraphics[scale=0.3]{Figures_quant_design/weyl_chamber_grid_M_R_conv.png}
%\caption{Convergence to respective fixed points for initial conditions inside the Weyl chamber. The qualitatative features are not much different than that of the XXZ case.}
%\label{fig:xyzrates}
%\end{figure}
 We summarize the convergence of the map for different families in Table \ref{Tab:qubitmapconv}.
 The slow algebraic approach hold for all initial conditions that approach the {\sc swap} gate. In fact,
 we have numerically verified that even for higher dimensions, $d>2$, if the seed unitary is a fractional power of {\sc swap}, they approach the dual-unitary {\sc swap} gate algebraically as $\sim 1/\sqrt{n}$.
 
It maybe noted that a very different map on the Weyl chamber has been studied by looking at the the powers of two-qubit gates in Ref.~\cite{Mandarino2018}.  This map is ergodic on the Weyl chamber and is 
related to billiard dynamics in a tetrahedron, unlike the dissipative nature of the $\mathcal{M}_R$ map that we have studied.
 
%\begin{widetext}
\begin{table*}[htbp]
\centering
	\caption{Convergence to dual unitaries for different two-qubit seed unitaries, parametrized by the Cartan coefficients $c_i^{(0)}$. In all cases $c_1^{(\infty)}=c_2^{(\infty)}=\pi/4$, and $\Delta c_i= c_i^{(\infty)}-c_i^{(n)}$.}
	\begin{tabular}{||c|c|c||}
		\hline
		Cartan coefficients and Weyl chamber location of seeds & Dual-unitary approached  & Nature of convergence\\
		\hline
		\hline
		Base, $c_3^{(0)}=0$, $c_2^{(0)}>0$  & {\sc dcnot}, $c_3^{(1)}=\, c_3^{(\infty)}=0$& Instantaneous, rate $\infty$ \\
		\hline
		{\sc swap}-Local Edge: $c_1^{(0)}=c_2^{(0)}=c_3^{(0)}$ &  {\sc swap} &
		Algebraic: $\Delta c_i \sim 1/\sqrt{n}$\\
		\cline{1-1} \cline{3-3} 
		{\sc swap}-Local-{\sc cnot} face, $c_2^{(0)}=c_3^{(0)}\neq c_1^{(0)} $ & $c_3^{(\infty)}=\pi/4$ & \makecell{Algebraic: \\ $\Delta c_3, \, \Delta c_2 \sim 1/\sqrt{n},\, \Delta c_1 \sim 1/n$}\\
		\hline
		{\sc swap}-Local-{\sc dcnot} face, $c_1^{(0)}=c_2^{(0)}\neq c_3^{(0)}$ &    & Exponential: \\
		\cline{1-1}
		{\sc swap-cnot-dcnot} face, $c_1^{(0)}=\pi/4$, $c_2^{(0)} \neq c_3^{(0)} $& Generic  & $\Delta c_1, \Delta c_2 \sim \exp(- \xi\, n),\, \Delta c_3 \sim \exp(-2 \xi \, n)$  \\
		\cline{1-1}
		Interior, $c_1^{(0)} > c_2^{(0)} > c_3^{(0)}$ &$c_3^{(\infty)} \neq 0, \, \pi/4$ & $\xi =|\ln \sin(2 c_3^{(\infty)})|$\\
		\hline
	\end{tabular}
\label{Tab:qubitmapconv}
\end{table*}
%\end{widetext}

\section{Combinatorial designs corresponding to dual unitary operators \label{sec:combi}}
Tools developed in combinatorial mathematics have been very useful in constructing multipartite entangled states \cite{Goyeneche2015,GRMZ_2018}. In Ref.~\cite{Clarisse2005} it was shown that orthogonal Latin squares of order $d$ can be used to contruct 2-unitary permutation matrices of order $d^2$. Since 2-unitary operators belong to a subset of dual-unitary operators, we point to less restrictive combinatorial structures corresponding to general dual-unitary operators. In the case of dual-unitary permutations such designs were discussed earlier in \cite{ASL2020}, which we first summarize.

\subsection{Permutation matrices: classical design}
A permutation of $d^2$ symbols or elements from $\left[ d \right] \times \left[ d\right]$, $\left[d \right]=\left\lbrace{1,2,\cdots,d}\right\rbrace$, is specified by the operator on computational product basis states $\ket{ij}$ as,
\beq 
P\ket{ij}=\ket{k_{ij}l_{ij}}.
\label{eq:KLmat}
\eeq
Thus, this can be written in terms of a pair of $d \times d$ matrices $K=(k_{ij})$ and $L=(l_{ij})$. 
In \cite{ASA_2021}, it was shown that for $P$ to be dual-unitary (T-dual), 
%which is  restrictive conditions on $P$ than to be 2-unitary, 
the conditions on $K$ and $L$ matrices are:
\begin{enumerate}
\item[(i)]
{\em Condition on $K$:} No element repeats along any {\em row} ({\em column}). 
%Equivalently, vectors in each {\bf row} forms an orthonomal basis in $\mathcal{H}^d$.
\item[(ii)]
{\em Condition on $L$:} No element repeats along any {\em column} ({\em row}). 
%Equivalently, vectors in each {\bf column} forms an orthonomal basis in $\mathcal{H}^d$.
\end{enumerate}
As an example, for $d=2$, the dual-unitary {\sc swap} gate permutes the basis states as
\beq
\begin{array}{|cc|}
\hline
11 & 12 \\
21 & 22 \\
\hline
\end{array}
\; \longrightarrow \;
\begin{array}{|cc|}
\hline
11 & 21 \\
12 & 22 \\
\hline
\end{array},
\eeq
with $K$ and $L$ given by
\beq
K=
\begin{array}{|cc|}
\hline
1 & 2 \\
\hline
1 & 2 \\
\hline
\end{array}\; , \;
L=
\begin{array}{|c|c|}
\hline
1 & 1 \\
2 & 2 \\
\hline
\end{array}.
\eeq
%and the operator $P = \sum_{i,j=1}^d \op{k_{ij} l_{ij}}{ij}$.

Orthogonal Latin squares, denoted OLS$(d)$ \cite{keedwell2015latin}, are examples of designs used to construct the 2-unitary operators \cite{Clarisse2005}. A Latin square is a $d \times d$ array with $d$ distinct elements such that every element appears exactly once in each column and in each row. Two Latin squares with elements $s_{ij}$ and $t_{ij}$  are orthogonal if the ordered pairs $(s_{ij},t_{ij})$ are all distinct. 
% OLS$(d)$ exist for all $d$ except for $d=2$ and $d=6$ \cite{}. 

%In general, for a two-qudit SWAP gate
%\beq
%K=
%\begin{pmatrix}
%{\blu 1} & {\blu 2} & \cdots & {\blu d} \\
%1 & 2 & \cdots & d \\
%\vdots & \vdots & \vdots & \vdots \\
%1 & 2 & \cdots & d \\
%\end{pmatrix}, \;
%L=
%\begin{pmatrix}
%{\blu 1} &  1 & \cdots &  1 \\
%{\blu 2} & 2 & \cdots &  2 \\
%\vdots & \vdots & \vdots & \vdots \\
%{\blu d} & d & \cdots & d \\
%\end{pmatrix}=K^T
%\eeq
%Note that each symbol along any row in $K$ is different and each symbol along any column in $L$ is different. $K$ and $L$ are combinatorial structures less constrained than Latin squares in which no symbol repeats along any row and column.
%ermutation operators as each row of $K$ and column of $L$ are permutations of the computational basis set, there are  only $d$ distinct basis vectors. We will make above statements more clear in next sections.
%\subsubsection{Combinatorial designs of 2-unitary permutations}

If $K$ and $L$, defined above, are Latin squares, then the corresponding permutation matrix $P$ is both dual-unitary and T-dual, hence it is 2-unitary.
OLS($d$) exist for all $d$ except $d=2$ and $6$ \cite{bose1960further}. Thus 2-unitary permutations exist for all $d$ except $d=2$ and $6$. An example of an OLS($3$) is:
\beq
\begin{array}{|c|c|c|}
\hline
1 & 2 & 3 \\
\hline
3 & 1 & 2 \\
\hline
2 & 3 & 1 \\
\hline
\end{array} \; \cup \; 
\begin{array}{|c|c|c|}
\hline
1 & 3 & 2 \\
\hline
3 & 2 & 1 \\
\hline
2 & 1 & 3 \\
\hline
\end{array} \;= \;
\begin{array}{|c|c|c|}
\hline
11 & 23 & 32 \\
\hline
33 & 12 & 21 \\
\hline
22 & 31 & 13 \\
\hline
\end{array}
\label{eq:OLS3}.
\eeq 
Note that all nine pairs from the set $\left\lbrace 1,2,3 \right\rbrace \times \left\lbrace 1,2,3 \right\rbrace=\left\lbrace  11,12,\dots,32,33 \right\rbrace$ are present.

For dual-unitary or, T-dual permutations $K$ and $L$ are not Latin squares in general. We define {\em r-Latin square (c-Latin square)} as an arrangement of $d$ symbols in a $d \times d$ array if it satisfies conditions of a Latin square only along rows (columns). Note that the usual Latin square is both r-Latin square as well as c-Latin square. Two such less constrained Latin squares are orthogonal if by superposing them all $d^2$ ordered pairs obtained are distinct. For dual-unitary permutations $K$ is r-Latin square and $L$ is c-Latin square while for T-dual permutations $K$ is c-Latin square and $L$ is r-Latin square, which are restatements of the conditions above for duality (T-duality). 

\subsection{General dual-unitary operators: Quantum design}
Here we discuss the underlying combinatorial structure of general dual-unitary operators. Consider a unitary operator $U \in \mathcal{B}(\mathcal{H}_d \otimes \mathcal{H}_d) $. Define 
\beq
\ket{\psi_{ij}}=U\ket{i j},
\label{eq:Uijdef}
\eeq
where $\{\ket{ij}\}_{i,j=1}^d$ is the computational basis in $\mathcal{H}_d \otimes \mathcal{H}_d$. The unitarity of $U$ implies that the set of vectors $\{\ket{\psi_{ij}}\}_{i,j=1}^d$ also forms an orthonormnal basis in $\mathcal{H}_d \otimes \mathcal{H}_d$.

Consider $\ket{\psi_{ij}}$'s which are of product form
\beq
\ket{\psi_{ij}}=\ket{\alpha_{ij}} \otimes \ket{\beta_{ij}}.
\eeq
%which most of the known families of dual unitary operators satisfy.
%We consider the most general case in which $\ket{\psi_{ij}}$'s are entangled in Appendix ~(\ref{app:general_comb_design}). 
Analagous to $K$ and $L$ defined in the previous section for permutation operators, we arrange $d^2$ single qudit states $\ket{\alpha_{ij}}$ and  $\ket{\beta_{ij}}$ as follows: 
 \beq 
 \begin{split}
\mathcal{K} &=
\begin{array}{|cccc|}
\hline
\ket{\alpha_{11}} & \ket{\alpha_{12}} & \cdots & \ket{\alpha_{1d}} \\
\ket{\alpha_{21}} & \ket{\alpha_{22}} & \cdots & \ket{\alpha_{2d}} \\
\vdots & \vdots & \vdots & \vdots \\
\ket{\alpha_{d1}} & \ket{\alpha_{d2}} & \cdots & \ket{\alpha_{dd}} \\
\hline
\end{array} \\
\mathcal{L} &=
\begin{array}{|cccc|}
\hline
\ket{\beta_{11}} & \ket{\beta_{12}} & \cdots & \ket{\beta_{1d}} \\
\ket{\beta_{21}} & \ket{\beta_{22}} & \cdots & \ket{\beta_{2d}} \\
\vdots & \vdots & \vdots & \vdots \\
\ket{\beta_{d1}} & \ket{\beta_{d2}} & \cdots & \ket{\beta_{dd}} \\
\hline
\end{array}
\end{split} 
\label{eq:KLmat}
\eeq 
The conditions for $U$ to be dual-unitary in terms of $\mathcal{K}$ and $\mathcal{L}$ is presented below.
\begin{thm} 
If every  row of $\mathcal{K}$ and every  column of $\mathcal{L}$  forms an orthonormal basis in $\mathcal{H}_d$, then the unitary operator $U = \sum_{i,j=1}^d \op{\psi_{ij}}{ij} = \sum_{i,j=1}^d \op{\alpha_{ij}\beta_{ij}}{ij}$ is dual-unitary. 
\end{thm} 

\begin{proof} 
The orthonormality condition on the vectors in every row of $\mathcal{K}$ and every column of $\mathcal{L}$ imply that  $
\langle \alpha_{ij}|\alpha_{ij'}\rangle=\delta_{jj'},\;\sum_{j=1}^d\ket{\alpha_{ij}}\bra{\alpha_{ij}}=I_d, \; \forall\:i $ and $\langle \beta_{ij}|\beta_{i'j}\rangle=\delta_{ii'},\;\sum_{i=1}^d\ket{\beta_{ij}}\bra{\beta_{ij}}=I_d, \; \forall\:j $. 
%Realignment of $U$is given by $U^{R_2} = \sum_{i,j=1}^d \ket{\alpha_{ij} i} \bra{\beta_{ij}j}.$ It is straight forward to see that $U^{R_2}$ is unitary if it satisfies the orthogonality condition defined above. 
Using these conditions, it follows,
\begin{align*}
U^RU^{R \dagger} & = \left(\sum_{i,j=1}^d \ket{\alpha_{ij} i} \bra{\beta_{ij}j}\right)\left(\sum_{i',j'=1}^d \ket{\beta_{i'j'} j'} \bra{\alpha_{i'j'}i'}\right),\\
& = \left(\sum_{j=1}^d \ket{\alpha_{ij}} \bra{\alpha_{ij}}\right) \otimes \left(\sum_{i=1}^d\ket{i}\bra{i}\right),\\
& =I_d \otimes I_d= I_{d^2}.
\end{align*}
It is similarly shown that $U^{R \dagger} U^R=I_{d^2}$, and hence unitary $U$ is dual-unitary.
\end{proof}

The conditions on $\mathcal{K}$, $\mathcal{L}$ for $U$ to be dual-unitary are generalizations of $K$ and $L$ corresponding to dual-unitary permutations. In $K$, $L$ the notion of symbols being {\em different} in row or column is replaced by it's quantum analog, the {\em orthogonality} of vectors (quantum states). In fact such a generalization is known for Latin square and OLS called as quantum Latin square (QLS) \cite{MV16} and orthogonal quantum Latin square (OQLS)  \cite{GRMZ_2018,MV19} respectively. A quantum Latin  square is a $d \times d$ array of $d$-dimensional vectors such that each row and each column forms an orthonormal basis in $\mathcal{H}_d$. Two quantum Latin squares are orthogonal if together they form an orthonormal basis in $\mathcal{H}_d \otimes \mathcal{H}_d$. If $\mathcal{K}$, $\mathcal{L}$ defined above are quantum Latin squares then $U$ is a 2-unitary operator \cite{GRMZ_2018}. 

The fact that there are no repetitions of symbols in a Latin square in any row or column translates into orthogonality of vectors in each row and column in the corresponding QLS. The ``quantumness" and equivalence between quantum Latin squares was defined in Ref.~\cite{Paczos_2021} in terms of the number of distinct basis vectors (up to phases), known as the {\em cardinality}. For QLS constructed from classical Latin squares, simply by replacing the symbol $k$ by a basis vector $\ket{k}$ in a $d$ dimensional space, the cardinality is $d$ and is said to be {\em classical}.
 
Quantum Latin square with cardinality more than $d$ cannot be obtained from classical Latin square using unitary transformations of the basis vectors and is referred to as {\em genuinely quantum} \cite{Paczos_2021}. Quantum Latin squares with cardinality equal to $d^2$, the maximum possible value, for general $d$ and their relation to quantum sudoku is discussed in Refs.~\cite{Paczos_2021,Nechita_qsudoku}.

For dual-unitary or, T-dual unitary operators $\mathcal{K}$ and $\mathcal{L}$ are not quantum Latin squares in general. We define r-quantum Latin square (c-quantum Latin square) denoted by r-QLS (c-QLS) as a $d \times d$ array of $d$-dimensional vectors if it satisfies conditions of a quantum Latin square only along rows (columns). Note that the quantum Latin square is both r-QLS as well as c-QLS. For dual-unitary operators $\mathcal{K}$ is r-QLS and $\mathcal{L}$ is c-QLS while for T-dual operators $\mathcal{K}$ is c-QLS  and $\mathcal{L}$ is r-QLS.

Two such less constrained QLS are said to be orthogonal if together they form an orthonormal basis in $\mathcal{H}_d \otimes \mathcal{H}_d $. In analogy with cardinality of a quantum Latin square, we define cardinality of $\mathcal{K}$ or $\mathcal{L}$ as the number of distinct basis vectors (up to phases) they contain. An r-QLS or c-QLS of size $d$ is \emph{classical} if it contains $d$ distinct basis vectors and \emph{genuinely quantum} if it contains more than $d$ distinct basis vectors. For dual-unitary permutations cardinality of $\mathcal{K}$, $\mathcal{L}$ is always equal to $d$ and are thus classical. An example of a pair of genuine r-QLS and c-QLS of size $3$ are respectively,
\beq
\mathcal{K}:
\begin{array}{|c c c|}
\hline
\ket{1} & \ket{2} & \ket{3} \\
\hline
\ket{1} & \ket{2} & \ket{3} \\
\hline
\frac{1}{\sqrt{2}}\left(\ket{1}+\ket{2}\right) & \frac{1}{\sqrt{2}}\left(\ket{1}-\ket{2}\right) & \ket{3}\\
\hline
\end{array}\;,
\eeq
\beq
\mathcal{L}:
\begin{array}{|c|c|c|}
\hline
\ket{1} & -\ket{1} & \frac{1}{\sqrt{2}}\left(\ket{1}+\ket{2}\right) \\
\ket{2} & \ket{2} & \ket{3} \\
\ket{3} & \ket{3} & \frac{1}{\sqrt{2}}\left(\ket{1}-\ket{2}\right)\\
\hline
\end{array}\;.
\eeq
Note that both $\mathcal{K}$ (r-QLS) and $\mathcal{L}$ (c-QLS) contain five distinct basis vectors (quantum states), across two different orthonormal bases, and are thus genuinely quantum. Together $\mathcal{K}$ and $\mathcal{L}$ form an orthonormal basis in $\mathcal{H}_3 \otimes \mathcal{H}_3$ arranged in $d \times d $ array as,
\beq
\small{
\begin{array}{|c | c | c|}
\hline
\ket{1} \otimes \ket{1} & -\ket{2} \otimes \ket{1} & \ket{3} \otimes \frac{1}{\sqrt{2}}\left(\ket{1}+\ket{2}\right)  \\
\hline
\ket{1} \otimes \ket{2} & \ket{2} \otimes \ket{2} & \ket{3} \otimes \ket{3}  \\
\hline
\frac{1}{\sqrt{2}}\left(\ket{1}+\ket{2}\right) \otimes \ket{3} & \frac{1}{\sqrt{2}}\left(\ket{1}-\ket{2}\right) \otimes \ket{3}  & \ket{3} \otimes \frac{1}{\sqrt{2}}\left(\ket{1}-\ket{2}\right)\\
\hline
\end{array}\;
}.
\eeq
The dual-unitary gate corresponding to the above arrangement of size 9 is,
\beq
U_9=\left(\begin{array}{ccccccccc}
1 & 0 & 0 & 0 & 0 & 0 & 0 & 0 & 0 \\
0 & 0 & 0 & 1 & 0 & 0 & 0 & 0 & 0 \\
0 & 0 & 0 & 0 & 0 & 0 & \frac{1}{\sqrt{2}} & \frac{1}{\sqrt{2}} & 0 \\
0 & -1 & 0 & 0 & 0 & 0 & 0 & 0 & 0 \\
0 & 0 & 0 & 0 & 1 & 0 & 0 & 0 & 0 \\
0 & 0 & 0 & 0 & 0 & 0 & \frac{1}{\sqrt{2}} & -\frac{1}{\sqrt{2}} & 0 \\
0 & 0 &  \frac{1}{\sqrt{2}} & 0 & 0 & 0 & 0 & 0 &  \frac{1}{\sqrt{2}} \\
0 & 0 &  \frac{1}{\sqrt{2}} & 0 & 0 & 0 & 0 & 0 &  -\frac{1}{\sqrt{2}} \\
0 & 0 & 0 & 0 & 0 & 1 & 0 & 0 & 0 
\end{array}\right),
\eeq
with $(e_p(U_9),g_t(U_9))=(3/4,5/8)$. This dual unitary is not locally equivalent to any dual unitary permutation matrix (corresponding $\mathcal{K}$ and $\mathcal{L}$ contain only three distinct vectors) with the same entangling power and gate-typicality. We obtained dual-unitary $U_9$ using the $\mathcal{M}_R$ map (see Sec.~(\ref{subsec:MRmap}). This is one of the nice properties of the map that it yields structured dual-unitaries by choosing appropriate seed unitaries like permutations. For $e_p(U)<1$, it is relatively easier to construct dual unitaries which are LU inequivalent to dual unitary permutations with the same entangling power. However for $e_p(U)=1$ {\em i.e.}, 2-unitaries this is not the case as they satisfy additional constraints which we discuss in the next section.
%{\blu An important point to note is that although there are only 3 distinct basis vectors in $\mathcal{K}$ but there are 7 distinct basis vectors in $\mathcal{L}$; $3$ different bases.}
%\subsubsection{Combinatorial design for 2-unitaries: orthogonal quantum Latin squares (OQLS)}
%For a 2-unitary operator each row and column in $\mathcal{K}$ and $\mathcal{L}$ form an orthonormal basis in $\mathcal{H}_d$ and are called quantum Latin squares \cite{MV16}. Together they form an orthogonal basis in $\mathcal{H}_d \otimes \mathcal{H}_d$ and the resulting combinatorial is called orthogonal quantum Latin square (OQLS) \cite{GRMZ_2018,MV19}. An example of OQLS of size $4$ is:
%{\red S.A.R: To be put in}

\subsection{Combinatorial structures of known families of dual unitaries}
\subsubsection{Diagonal ensemble}
Dual unitaries have one-to-one correspondence with T-dual unitary operators which are easier to construct. Simplest ensemble of T-dual unitaries one can think of is that of diagonal unitaries with arbitrary phases, denoted $D_1$. A $d^2$ parameter subset of dual unitaries can be obtained by (pre- or post-) multiplying diagonal unitaries with the {\sc swap} gate $S$ \cite{claeys2020ergodic,ASA_2021}. It is easy to see that for dual unitaries of the form $U=D_1S$ obtained from the diagonal ensemble,
\beq
U(\ket{k}\otimes \ket{l})=D_1S(\ket{k}\otimes \ket{l})=\exp(i\,\theta_{lk})(\ket{l}\otimes \ket{k}).
\eeq
Thus, the corresponding $\mathcal{K}$ and $\mathcal{L}$ are same as that of the {\sc swap} gate (up to phases) and hence are classical. 
\subsubsection{Block-diagonal ensemble}
A more general $d^3$ parameter family of dual unitary gates, $U=D_dS$, can be obtained from block-diagonal unitaries \cite{ASA_2021,prosen2021many,borsi2022remarks}, given by   
\beq
D_d=\sum_{i=1}^d \ket{i}\bra{i} \otimes u_i,\;u_i \in \mathcal{U}(d).
\label{eq:controlunitary}
\eeq
This is a controlled unitary from first subsystem to the second. For this family of dual unitaries the combinatorial structures are given by
\beq
\mathcal{K}:
\begin{array}{|c c c c|}
\hline
\ket{1} & \ket{2} & \cdots & \ket{d}\\
\hline
\ket{1} & \ket{2} & \cdots  & \ket{d}\\
\hline
\vdots & \vdots & \vdots  & \vdots\\
\hline
\ket{1} & \ket{2} & \cdots & \ket{d}\\
\hline
\end{array}\;\;,\;
\mathcal{L}:
\begin{array}{|c|c|c|c|}
\hline
u_1\ket{1} & u_2\ket{1} & \cdots & u_d\ket{1}\\

u_1\ket{2} & u_2\ket{2} & \cdots  & u_d\ket{2}\\

\vdots & \vdots & \vdots  & \vdots\\

u_1\ket{d} & u_2\ket{d} & \cdots & u_d\ket{d}\\
\hline
\end{array}\;,
\eeq
where $u_i$'s are related to the dual unitary $U=D_dS$ by Eq.~(\ref{eq:controlunitary}). Note the orthonormality along the columns in $\mathcal{L}$ is ensured by the identical unitary transformation of each basis vector. Although $\mathcal{K}$ contains only $d$ distinct vectors and is classical, $\mathcal{L}$ contains in general (the maximum possible) $d^2$ number of distinct vectors and hence is genuinely quantum.

The quantum designs considered so far are mostly unentangled, such as the ones above. Generalizations to entangled designs are needed to describe for example the recently found 2-unitary operator behind the AME$(4,6)$ state \cite{SRatherAME46}. 
Although one can write necessary and sufficient conditions for $U$ to be 2-unitary; see Appendix (\ref{app:general_comb_design}), in terms of reduced density matrices of bipartite states defined in Eq.~(\ref{eq:Uijdef}) but the orthogonality relations in the corresponding OQLS are harder to interpret than in OLS. 

An unitary gate $U$ on $\mathcal{H}_d \otimes \mathcal{H}_d$
 is an {\em universal entangler} if $U(\ket{\alpha_i} \otimes \ket{\beta_i})$ is {\em always} entangled for any chioce of the product state $\ket{\alpha_i} \otimes \ket{\beta_i}$. It is known that universal entanglers do not exist for $d=2$ and $3$ i.e., there is no two-qubit or two-qutrit unitary gate which maps every product state to an entangled state \cite{ChenDuan2007}. It is easy to see that all columns of a universal entangler must be entangled, however this condition is necessary but not sufficient \cite{MENDES2015}. Those dual-unitary and 2-unitary gates which are universal entanglers will have genuinely entangled quantum designs. Unfortunately there are no known constructions of universal entangler and conditions under which they are obtained are not known.

\section{Local unitary equivalence of 2-unitary operators \label{Sec:spldual}}

\subsection{A necessary criterion \label{subsec:LUequiv}}

Given any two bipartite unitary operators $U$ and $U'$, as far as we know, there is no procedure to determine if they are LU equivalent, $U\overset{\text{LU}}{\sim}U'$, or not denoted by $U\overset{\text{LU}}{\nsim}U'$. Namely if Eq.~(\ref{eq:U1U2equi}) is satisfied for some local operators $u_i$ and $v_i$. The problem is exacerbated for the case of 2-unitary operators as the singular values of $U^{R}$ and $U^{\Gamma}$, which are LUI, are all equal, and hence maximize the standard invariants such as $E(U)$ and $E(US)$. 

Here we propose a necessary criterion to investigate the LU equivalence between unitary operators based on the distributions of the entanglement they produce when applied on an ensemble of uniformly generated product states. Action of a bipartite unitary operator $U$ on product states generically results in entangled states,
\beq
\ket{\psi_{AB}}=U(\ket{\phi_A} \otimes \ket{\phi_B}).
\label{eq:stateUAB}
\eeq
Let $\mathcal{E}(\ket{\psi}_{AB})$ be any measure of entanglement, and let $\phi_A$ and $\phi_B$ be sampled from the Haar measure on the subspaces. Then the resulting distribution $p(x;U)$ of the entanglement is 
\beq
p(x;U)=\int \delta\left( x-\mathcal{E}[U(\ket{\phi_A} \otimes \ket{\phi_B})] \right)d\mu(\phi_A) d\mu(\phi_B). 
\eeq
It is clear that if $U$ is left multiplied by local unitaries $p(x; U)$ is unchanged as entangled measures are invariant under such operations. 
If $U'=U (u_A \otimes u_B)$, then 
$p(x; U')=p(x; U)$, as $d \mu( u_A^{\dagger} \ket{\phi_A})=d \mu(\ket{\phi_A})$, which is a property of the Haar measure.
Thus if $U\overset{\text{LU}}{\sim}U'$ then $p(x; U')=p(x; U)$.
Conversely if $p(x; U') \neq p(x; U)$ this implies that $U\overset{\text{LU}}{\nsim}U'$.

However, if the distributions are {\em indistinguishable}, {\em i.e.}, $p(x; U') =p(x; U)$, then $U$ and $U'$ {\em may} or, {\em may not} be LU equivalent. To see that the criterion is necessary but not sufficient, consider two LU inequivalent operators $U$ and $U'=US$, where $S$ is the {\sc swap} gate. Although $U$ and $U'$ are LU inequivalent they generate identical entanglement distributions, $p(x; U)=p(x; U')$. Note that $U$ and $U'$ have the same entangling power; $e_p(U)=e_p(U')$, but have different gate typicalities; $g_t(U)\neq g_t(U')$, and are thus LU inequivalent. 

We enlarge local equivalence between $U$ and $U'$ to include multiplication by {\sc swap} gates on either or both sides,  denoted by $U' \stackrel{\text{LUS}}{\sim}U$ as
\beq
U'=(u_1 \otimes v_1)\, S^a\, U\, S^b \,(u_2 \otimes v_2),
\label{eq:LUSequiv}
\eeq
where $u_i$'s and $v_i$'s are single qudit gates, and $a,b$ takes values $0$ or $1$. Any operator in the LUS equivalence class of $U$ will produce the same entanglement distribution, $p(x; U)$.

\subsection{2-unitaries in $d=3$}
\subsubsection{Permutations}

2-unitary permutations of order $d^2$ maximize the entangling power and are in one-to-one correspondence with orthogonal Latin squares of size $d$; OLS(d) \cite{Clarisse2005}. In general, 2-unitary operators are in one-to-one correspondence with AME states of four qudits \cite{Goyeneche2015}. Under this mapping 2-unitary permutation matrices correspond to AME states with minimal support \cite{Goyeneche2015} i.e., these contain minimal possible terms equal to $d^2$ when written in the computational basis. A complete enumeration of all possible 2-unitary permutations  of size $d^2$ boils down to the possible number of OLS($d$) which is known for $d \leq 9$ (see A072377 , Ref. \cite{OEIS}). For $d=3$ there are 72 possible 2-unitary permutations of size 9. We find by a direct numerical exhaustive search over local permutation matrices of size 3 that \emph{all} 72 possible 2-unitary permutations are LU equivalent. This observation leads to the following proposition.

\begin{prop}
\label{prop:P9LUclass}
There is only one LU class of 2-unitary permutations of order 9.
\end{prop}
We choose the following 2-unitary permutation as a representative of the LU equivalent class of 2-unitary permutations of order 9,
\beq
P_9=\left(
\def\arraystretch{0.7}
\begin{array}{ccc|ccc|ccc}
1 & . & . & . & . & . & . & . & .\\
. & . & . & . & 1 & . & . & . & .\\
. & . & . & . & . & . & . & . & 1\\
\hline
. & . & . & . & . & 1 & . & . & .\\
. & . & . & . & . & . & 1 & . & .\\
. & 1 & . & . & . & . & . & . & .\\
\hline
. & . & . & . & . & . & . & 1 & .\\
. & . & 1 & . & . & . & . & . & .\\
. & . & . & 1 & . & . & . & . & .
\end{array}
\right)
\label{eq:P9}
\eeq
An easy way to obtain all 72 possible 2-unitary permutations is by searching over $(3!)^4=1296$ local permutations $p_i$ of size 3 in 
\beq
\label{eq:plocP}
P'=(p_1 \otimes p_2) P_9 (p_3 \otimes p_4).
\eeq
Although this is not an efficient way as each 2-unitary permutation is repeated $18$ times but all $1296/18=72$ possible permutations can be obtained. 

An equivalent statement in terms of LU equivalence of AME($4,3$) states with minimal support is known, see Ref.~\cite{Adam_SLOCC_2020}. An AME($4,3$) state with minimal support considered in Ref.~\cite{Adam_SLOCC_2020} contains arbitrary phases and is equivalent to an enphased 2-unitary permutation {\em i.e.}, 2-unitary permutation multiplied by a diagonal unitary. This is a special property of 2-unitary permutations that these remain 2-unitary upon multiplication by diagonal unitaries with arbitrary phases owing to their special combinatorial structure. Indeed one can show that in $d=3$ that all enphased permutations are LU equivalent to $P_9$. Local dimension $d=3$ is special in the sense that number of phases, $d^2-1=8$, exactly matches the number of phases one can absorb using four enphased local permutations each containing $d-1$ phases; $4(d-1)=8$. Note that $d^2-1=4(d-1)$ has a solution only for $d=3$ and thus such results about LU equivalence about enphased 2-unitary permutations in $d=3$ do not hold for $d>3$.

\subsubsection{LU equivalence of 2-unitaries in $d=3$}
Dynamical maps are very efficient in yielding 2-unitaries for local Hilbert space dimension $d=3,4$ from random seed unitaries. The 2-unitaries so obtained do not have an evident simple structure as 2-unitary permutations have. It is natural to ask if these are LU equivalent to each other. For the purposes of LU equivalence we compare the entanglement distributions of 2-unitaries obtained from the map, $p(x;U)$ with that of the 2-unitary permutation matrix $p(x; P_9)$. We find that the von Neumann entropy is a good measure to highlight the differences in the distributions, especially it performs better than the linear entropy and hence we use this. Von Neumann entropy of the single qudit reduced density matrix of $\ket{\psi}_{AB}$ (see Eq.~(\ref{eq:stateUAB})) is defined as
$$\mathcal{E}_{v}\left(\rho_A\right)=-\text{tr}(\rho_A \log \rho_A).$$

The distribution $p(x;P_9)$ and $p(x;U_9)$ from a 2-unitary $U_9$ are shown in Fig.~(\ref{fig:ent_dist_d_3}). The matrix $U_9$ has been obtained using a random seed in the dynamical map $\mathcal{M}_{\Gamma R}$.  
\begin{figure}
\centering
\includegraphics[scale=0.65]{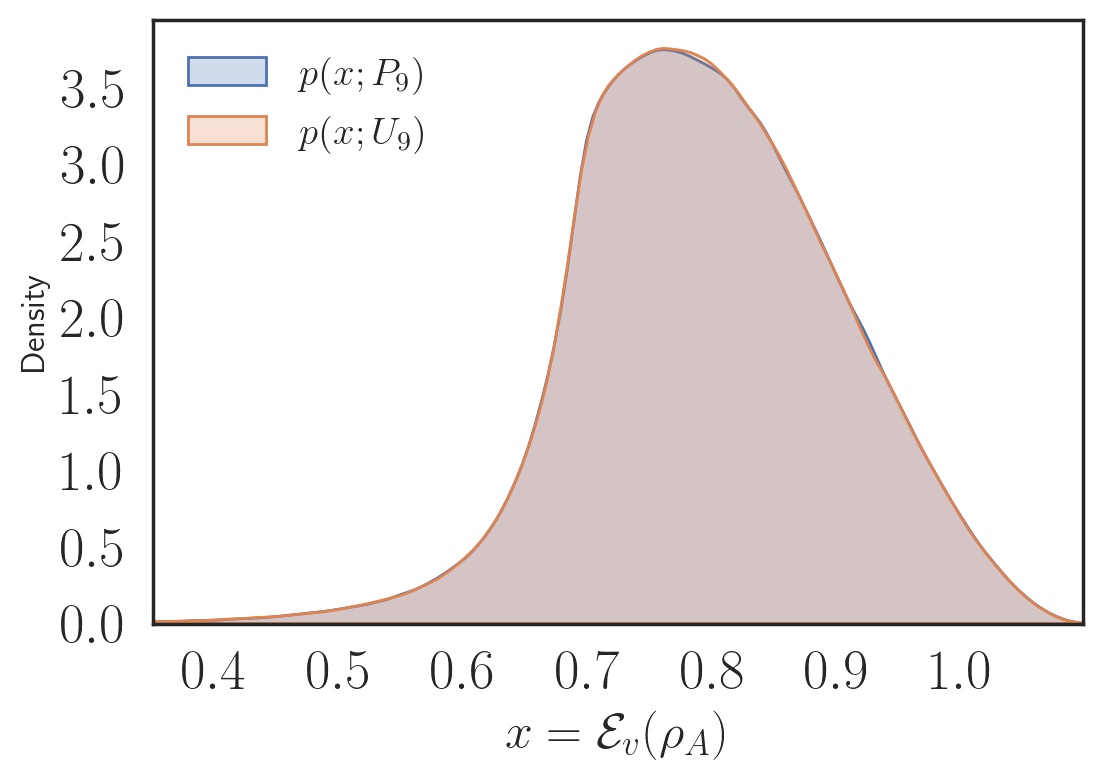}
\caption{Two-qutrit case, $d^2=9$. Distributions $p(x;U)$ of entanglement obtained from the action of  $9 \times 9$ unitaries $U$ on  Haar distributed product states. 
Shown are two distributions, one corresponding to the 2-unitary permutation matrix $P_{9}$, Eq.~(\ref{eq:P9}) and other a 2-unitary $U_9$ obtained from the map with a random seed. These two are numerically indistinguishable, and we could not find a different distribution than the one shown here for $O(10^3)$ number of 2-unitaries obtained from the map.}
\label{fig:ent_dist_d_3}
\end{figure}
Surprisingly, both the  distributions are indistinguishable for these 2-unitaries, although their origins and forms are very different. We have checked entanglement distributions for $O(10^3)$ 2-unitaries obtained from the map but could not find a different distribution from that of the 2-unitary permutation matrix. In fact in most cases we could transform 2-unitaries obtained from the map, using random seeds, into 2-unitary permutation matrices using appropriate local transformations in $\mathcal{U}(3) \otimes \mathcal{U}(3)$. Based on overwhelming numerical evidence we propose the following conjecture:
\begin{conj}
All 2-unitaries of order 9 are LU equivalent to $P_9$.
\end{conj}

\subsection{2-unitaries in $d=4$}
\subsubsection{Permutations}
The total number of 2-unitary permutations of size $16$ is $2\times 3456=6912$. By performing a direct exhaustive search over local permutations, quite remarkably even in this case it turns out that {\em all} $6912$ 2-unitary permutations are LU equivalent and thus lead to the following proposition.

\begin{prop}
\label{prop:P16LUclass}
There is only one LU class of 2-unitary permutations of order 16.
\end{prop}
We choose the following 2-unitary permutation matrix as a representative of the LU equivalent class of 2-unitary permutations of order 16,
\beq
P_{16}=
\begin{pmatrix}
\def\arraystretch{0.7}
\begin{array}{cccc|cccc|cccc|cccc}
1 & . & . & . & . & . & . & . & . & . & . & . & . & . & . & .\\
. & . & . & . & . & 1 & . & . & . & . & . & . & . & . & . & .\\
. & . & . & . & . & . & . & . & . & . & 1 & . & . & . & . & .\\ 
. & . & . & . & . & . & . & . & . & . & . & . & . & . & . & 1\\
\hline
. & . & . & . & . & . & . & 1 & . & . & . & . & . & . & . & .\\
. & . & 1 & . & . & . & . & . & . & . & . & . & . & . & . & .\\
. & . & . & . & . & . & . & . & . & . & . & . & . & 1 & . & .\\
. & . & . & . & . & . & . & . & 1 & . & . & . & . & . & . & .\\
\hline
. & . & . & . & . & . & . & . & . & 1 & . & . & . & . & . & .\\ 
. & . & . & . & . & . & . & . & . & . & . & . & 1 & . & . & .\\
. & . & . & 1 & . & . & . & . & . & . & . & . & . & . & . & .\\
. & . & . & . & . & . & 1 & . & . & . & . & . & . & . & . & .\\
\hline
. & . & . & . & . & . & . & . & . & . & . & . & . & . & 1 & .\\
. & . & . & . & . & . & . & . & . & . & . & 1 & . & . & . & .\\
. & . & . & . & 1 & . & . & . & . & . & . & . & . & . & . & .\\
. & 1 & . & . & . & . & . & . & . & . & . & . & . & . & . & .
\end{array}
\end{pmatrix}
\label{eq:P16}
\eeq

Using $(4!)^4=3,31,776$ possible local permutations, each 2-unitary permutation is obtained $48$ times and therefore all $3,31,776/48=6912$ are taken into account. 

\subsubsection{ Entangled OLS of size $4$: A new example of AME($4,4$)}
Although there is only one LU class of 2-unitary permutations of order 16, we give an explict example of a 2-unitary orthogonal matrix which is not LU equivalent to any 2-unitary permutation. This is obtained via the nonlinear map $\mathcal{M}_{\Gamma R}$ given in Eq.~(\ref{eq:MTRmap}) with a permutation seed, and is given by

\begin{widetext}
\beq
O_{16}=\frac{1}{2}
\begin{pmatrix}
\begin{array}{cccc|cccc|cccc|cccc}
1 & . & . & . & . & 1 & . & . & . & . & -1 & . & . & . & . & -1 \\
. & 1 & . & . & -1 & . & . & . & . & . & . & -1 & . & . &-1 & . \\
. & . & -1 & . & . & . & . & 1 & -1 & . & . & . & . & -1 & . & . \\
. & . & . & -1 & . & . & 1 & . & . & 1 & . & . & 1 & . & . & .  \\
\hline
. & 1 & . & . & -1 & . & . & . & . & . & . & 1 & . & . & 1 & . \\
-1 & . & . & . & . & 1 & . & . & . & . & 1 & . & . & . & . & -1 \\
. & . & . & -1 & . & . & 1 & . & . & -1 & . & . & -1 & . & . & . \\
. & . & -1 & . & . & . & . & -1 & 1 & . & . & . & . & -1 & . & . \\
\hline
. & . & -1 & . & . & . & . & -1 & -1 & . & . & . & . & 1 & . & . \\
. & . & . & 1 & . & . & 1 & . & . & -1 & . & . & 1 & . & . & . \\
-1 & . & . & . & . & -1 & . & . & . & . & -1 & . & . & . & . & -1 \\
. & -1 & . & . & -1 & . & . & . & . & . & . & -1 & . & . & 1 & . \\
\hline
. & . & . & -1 & . & . & -1 & . & . & -1 & . & . & 1 & . & . & . \\
. & . & 1 & . & . & . & . & -1 & -1 & . & . & . & . & -1 & . & . \\
. & 1 & . & . & 1 & . & . & . & . & . & . & -1 & . & . & 1 & . \\
1 & . & . & . & . & -1 & . & . & . & . & 1 & . & . & . & . & -1
\end{array}
\end{pmatrix}
\label{eq:O16mat}
\eeq 
\end{widetext}

This matrix can be written in a compact form as
\beq
O_{16}=P_{16}^T D_4 P_{16},
\label{eq:O16}
\eeq
where $D_4$ is a block-diagonal matrix consisting of four $4 \times 4$ Hadamard matrices. 
%The choice of $P_{16}$ in Eq.~(\ref{eq:O16}) is not unique and for the sake of convenience we choose it to be a 2-unitary permutation matrix.
Each row or column of $O_{16}$ contains four non-zero entries being equal to either $1/2$ or, $-1/2$ and is such that its eighth power is equal to the Identity, $O_{16}^8=\mathbb{I}$. 
To show that $O_{16}$ is indeed not LU equivalent to $P_{16}$, we compare the entanglement distributions $p(x;O_{16})$ and $p(x;P_{16})$. The distributions, shown in Fig.~(\ref{fig:ent_dist_d_4})
 are clearly distinguishable, $p(x;O_{16}) \neq p(x;P_{16})$ and thus $O_{16} \stackrel{LU}{\nsim} P_{16}$. Entanglement distributions for a large number of generic 2-unitaries obtained by applying the dynamical map $\mathcal{M}_{\Gamma R}$ on random seed unitaries did not result in any other distinguishable distributions other than the ones shown in Fig.~(\ref{fig:ent_dist_d_4}). This suggests that there are at least three LU classes of 2-unitaries in $d^2=16$. The representatives of these 3 LU classes are: (i) $P_{16}$ given by Eq.~(\ref{eq:P16}), (ii) enphased $P_{16}$; $P_{16}^{'}=D_1\,P_{16}\,D_2$, where $D_1$ and $D_2$ are diagonal unitaries with arbitrary phases, and (iii) $O_{16}$ given by Eq.~(\ref{eq:O16mat}). Note that $P_{16}$ and $P_{16}^{'}$ are not LU equivalent in general and thus the corresponding AME states of minimal support are not LU equivalent.
%\begin{figure}
%\centering
%\includegraphics[scale=0.5]{Figures_quant_design/O_16_2_uni.jpg}
%\caption{(Color online) Non-zero entries of an orthogonal 2-unitary matrix of order 16 are shown. Each row and column has four non-zero entries being equal to either $1/2$ (blue square) or, $-1/2$ (red square).}
%\label{fig:O162uni}
%\end{figure}

Each row or, column of $O_{16}$ treated as a pure state in $\mathcal{H}_4 \otimes \mathcal{H}_4$ is maximally entangled and thus the underlying combinatorial design corresponding to $O_{16}$ does not factor into separable structures $\mathcal{K}$ and $\mathcal{L}$ defined in Eq.~(\ref{eq:KLmat}). Also note that each $4 \times 4$ block in Eq.~(\ref{eq:O16mat}) is unitary up to a scale factor and thus rows or columns of $O_{16}^R$ are also maximally entangled states in $\mathcal{H}_4 \otimes \mathcal{H}_4$. Similar entangled combinatorial structure corresponding to 2-unitary of size $36$ has been referred to as entangled orthogonal quantum Latin squares (OQLS) in Ref.~\cite{SRatherAME46} in which entangled OQLS of size six was found. Based on our discussion in previous section on 2-unitaries of size $9$ and the known fact that there are no 2-unitaries of size $4$,  $d=4$ seems to be the smallest possible dimension in which entangled OLS exists.

%As a consequence of this we have given explicit example of AME state of four ququads which is not LU equivalent to any AME state of minimal support. 
This allows us to construct a new kind of AME state of four ququads; AME($4,4$), which is not LU equivalent to AME state of minimal support constructed from $P_{16}$. The corresponding AME state  written in computational basis is given by,
\beq
\ket{\text{AME}(4,4)}=\sum_{i,j,k,l=1}^4 (O_{16})_{ij,kl}\ket{ijkl}.
\label{eq:ketO16}
\eeq
The tensor $T_{ijkl}=(O_{16})_{ij,kl}$ is a  perfect tensor \cite{Pastawski2015} whose non-zero entries are given by Eq.~(\ref{eq:O16mat}). To our knowledge this is the simplest AME state that is not derived from a classical design or is equivalent to one, for equivalence among AME states, see for example \cite{Adam_SLOCC_2020,Acin_AME_2020}. Thus it qualifies as a younger cousin of the AME(4,6) which is a genuine quantum orthogonal Latin square \cite{SRatherAME46}. However, unlike the golden state AME(4,6), this is purely real.   Earlier constructions of ququad AME states have much larger number of particles \cite{Acin_AME_2020}.
%The resulting AME states constructed form $P_{16}$ and $O_{16}$ are not LU equivalent and belong to two different LU equivalence classes. 
%Similar distribution studies (not shown here) for several 2-unitaries obtained from the map seem to suggest that there are many LU equivalence classes of 2-unitary operators of size 16 and a complete classification of these is left for future studies. 

We performed several local unitary transformations on $O_{16}$ and reduced the numer of it's non-zero entries or, equivalently the support of AME state given by Eq.~(\ref{eq:ketO16}), from $64$ to $42$, although the transformed matrix has entries other than $\pm 1/2$. The transformed matrix has two unentangled columns and theferore $O_{16}$ is not a universal entangler.

%We have checked entanglement distributions of different 2-unitary permutations of size $16$ and found that these produce the same entanglement distribution as that of $P_{16}$ shown in Fig.~(\ref{fig:ent_dist_d_4}).  %Similar 2-unitary matrices can be generated using the map for $d>4$ which are not locally equivalent to 2-unitary permutation matrices.
\begin{figure}
\centering
\includegraphics[scale=0.65]{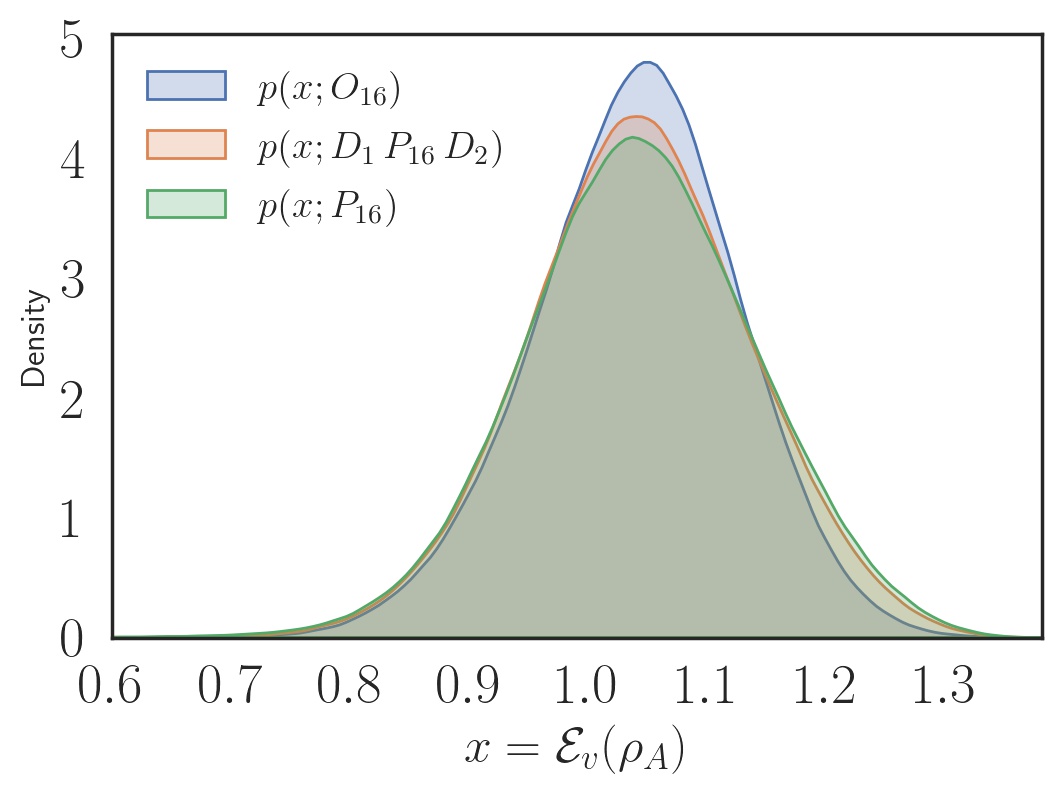}
\caption{Two-ququad case, $d^2=16$. Distributions $p(x;U)$ of entanglement corresponding to 2-unitary permutation matrix $U=P_{16}$ (Eq.~(\ref{eq:P16})) and the 2-unitary orthogonal matrix $U=O_{16}$ (Eq.~(\ref{eq:O16mat})). The distributions are clearly distinguishable and shows that $P_{16}$ and $O_{16}$ are not LU equivalent. Shown also is the entanglement distribution corresponding to enphased $P_{16}$; $D_1 P_{16}D_2$ where $D_1$ and $D_2$ are diagonal unitaries, which is not LU equivalent to either $P_{16}$ or $O_{16}$.}
\label{fig:ent_dist_d_4}
\end{figure}

%For 2-unitary operators every row and column in $\mathcal{K}$ (and in $\mathcal{L}$) forms an orthonormal basis in $\mathcal{H}^d$ and is called quantum Latin square (QLS) \cite{MV16}. Together $\mathcal{K}$ and $\mathcal{L}$ matrices form an orthonormal basis in $\mathcal{H}^d \otimes \mathcal{H}^d$ and a pair of {\em orthogonal quantum Latin square} (OQLS). \cite{MV19,GRMZ_2018}. For 2-unitary permutations of order $d^2$, $\mathcal{K}$ and $\mathcal{L}$ contain only $d$ distinct basis vectors i.e, only one basis set the computational basis. One of our main goal is to construct 2-unitary operators with more than $d$ distinct basis vectors in $\mathcal{K},\mathcal{L}$ matrices and which are not locally connected to 2-unitary permutations. Constructing 2-unitary operators with more than $d$ distinct basis vectors in $\mathcal{K},\mathcal{L}$ matrices does not seem to be hard but whether it can be tranformed to a 2-unitary permutation is a hard task. 

\section{Entangling properties of dual and T-dual permutation matrices \label{Sec:permu}}
Permutation matrices form an important class of entangling unitary operators \cite{Clarisse2005}. In this section we study the entangling properties of dual and T-dual permutation matrices on $\mathcal{H}_d \otimes \mathcal{H}_d$ which are special subsets of the permutation group $\mathrm{P}(d^2)$. Dual unitary permutation matrices have been recently explored in \cite{borsi2022remarks} and used as building blocks of quantum circuits with interesting dynamical behaviour \cite{ASA_2021}. Two permutation matrices $P_1$ and $P_2$ have been defined in \cite{Clarisse2005} to belong to the same \emph{entangling class} if they have the same entangling power, $e_p(P_1)=e_p(P_2)$. Two LU equivalent permutation matrices always belong to the same entangling class but permutation matrices belonging to the same entangling power need not be LU equivalent. For the sake of convenience we write the permutation matrix in terms of the column number of the only non-zero entry in each row. For example in this notation $P_9$ given by Eq.~(\ref{eq:P9}) is written as  $P_9=\left\lbrace 1,5,9,6,7,2,8,3,4 \right\rbrace$ corresponding to the permutation, $\pi:\left\lbrace 1,2,3,4,5,6,7,8,9\right\rbrace \rightarrow \left\lbrace \pi(1)=1,\pi(2)=5,\cdots,\pi(8)=3,\pi(9)=4 \right\rbrace.$

\subsection{Dual-unitary and T-dual permutations in $d=2$ and $d=3$}

We list all possible entangling classes for dual-unitary permutations in $\mathrm{P}(d^2)$ for $d=2,3$.

\emph{Two-qubit case}: In this case the corresponding permutation group is $\mathrm{P}(4)$. The projection of $\mathrm{P}(4)$ on $e_p$-$g_t$ plane is shown in Fig.~(\ref{fig:epgtP4}). There are $4$ distinct points corresponding to $24$ possible permutations of order $4$, treated as two-qubit gates, on the $e_p$-$g_t$ plane and every permutation matrix is either dual or T-dual. Number of entangling classes is only $2$ which are listed in Table \ref{Tab:epgtP4}.
\begin{table}[htbp]
\centering
	\caption{Entangling and LU classes of dual (equivalently T-dual) permutations in $\mathrm{P}(4)$.}
	\begin{tabular}{||c|c|c||}
		\hline
		S.No. & $e_p(P)$ & \# LU classes\\
		\hline
		\hline
		1 & 0 & 1 \\
		\hline
		2 & $\frac{2}{3}$ & 1 \\
		\hline
	\end{tabular}
\label{Tab:epgtP4}
\end{table}

\begin{figure}[h]
\includegraphics[scale=0.6]{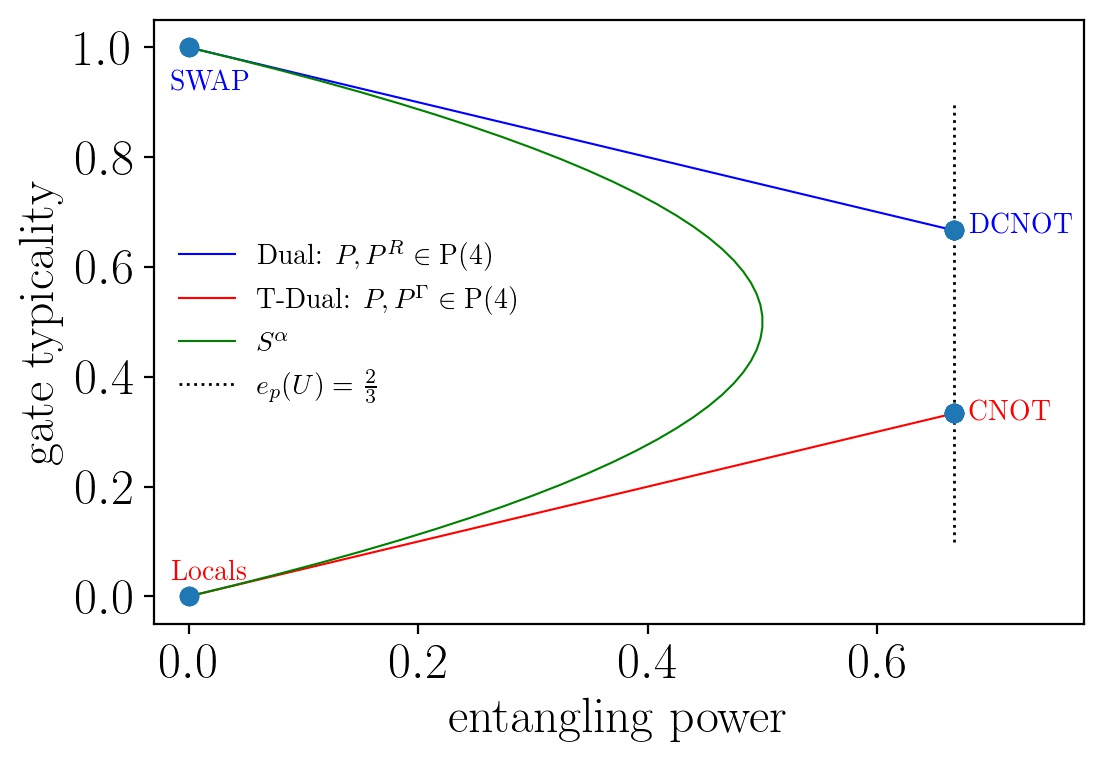}
\caption{Two-qubit case, $d^2=4$. The entangling power {\it vs} gate-typicality, of all permutations $\mathrm{P}(4)$, treated as two-qubit gates. Number of entangling classes, those with different entangling powers, is $2$.}
\label{fig:epgtP4}
\end{figure}

\emph{Two-qutrit case}: In this case the corresponding permutation group is $\mathrm{P}(9)$. The projection of $\mathrm{P}(9)$, treated as two-qutrit gates, on the $e_p$-$g_t$ plane is shown in Fig.~(\ref{fig:epgtP9}). There are only $60$ distinct points on the $e_p$-$g_t$ plane from the $9!=3,62,000$ possible permutation matrices of order $9$.
\begin{figure}[h]
\includegraphics[scale=0.6]{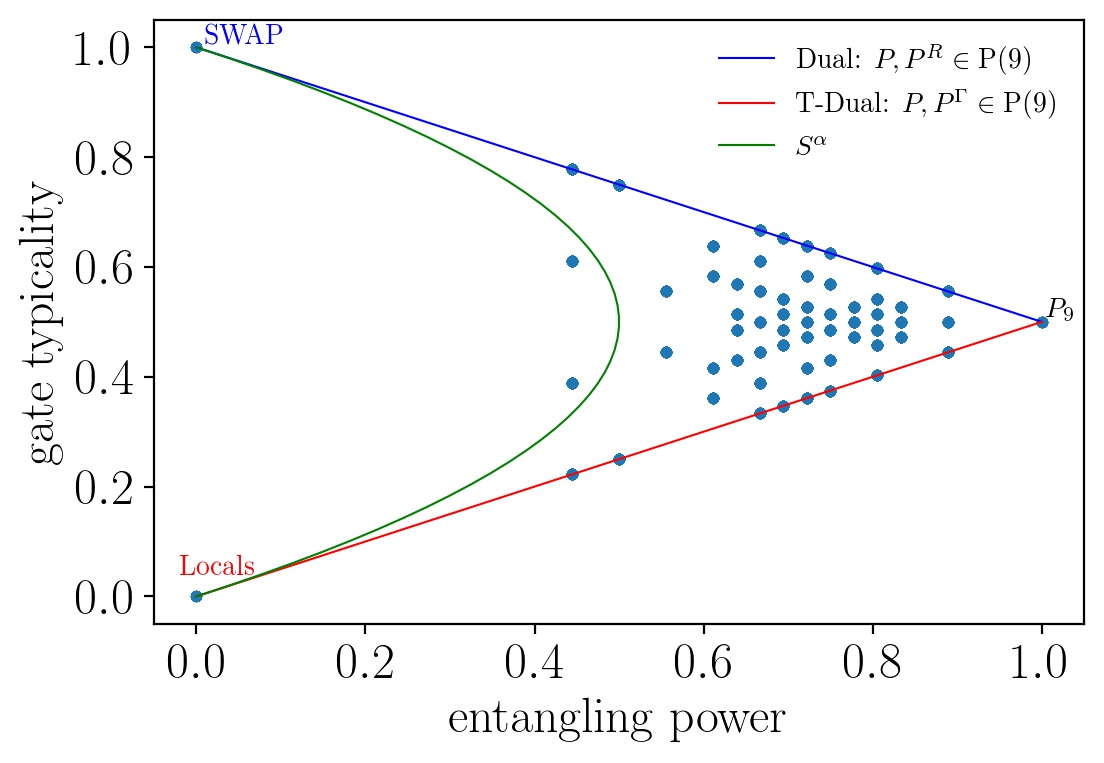}
\caption{Two-qutrit case, $d^2=9$. The entangling power {\it vs}  gate-typicality, of all $9!$ permutations $\mathrm{P}(9)$, treated as two-qutrit gates.  Number of entangling classes corresponding to dual (and equivalently T-dual) permutation matrices is $10$.}
\label{fig:epgtP9}
\end{figure}

It was shown in Ref. \cite{krishnan1996biunitary} that there are 18 LU classes of T-dual (or equivalently dual-unitary) permutation matrices. A representative permutation from each LU class are also listed therein. These LU classes are listed in the Table \ref{Tab:epgtP9}
along with their entangling powers. Therefore the number of entangling classes corresponding to dual unitary (and equivalently T-dual) permutation matrices is $10$.

Except for 3 entangling classes (corresponding to $e_p(P)=0,\,8/9,$ and $1$), there are more than one LU classes. Taking permutations from two different LU classes with the same entangling power (say $e_p(P)=1/2$), we observed that these produce the same entanglement distributions $p(x;U)$. This suggests that these LU inequivalent permutations with the same entangling power and gate-typicality might be connected by the {\sc swap} gate. Indeed we found that according to the LUS classification defined in Eq.~(\ref{eq:LUSequiv}) with $a=b=1$ that there are only $11$ LUS classes {\em i.e.}, two permutations $P$ and $P'$  belonging to the same entangling class but different LU classes are related (up to local permutations) as $P'=S P S$, where $S$ is the {\sc swap} gate.  This is the case for all entangling classes in the Table \ref{Tab:epgtP9} with more than one LU classes except for the entangling class $e_p(U)=2/3$.

\begin{table}
\setlength{\tabcolsep}{10pt}
\caption{Entangling, LU and LUS classes of dual (equivalently T-dual) permutations in $\mathrm{P}(9)$.}

\begin{tabular}{||c|c|c|c||}
\hline
S.No. & $e_p(P)$ & \# LU classes & \# LUS classes \\
\hline
\hline
1 & 0 & 1 & 1 \\
\hline
2 & $\frac{4}{9}$ & 2 & 1 \\
\hline
3 & $\frac{1}{2}$ & 2 & 1 \\
\hline
4 & $\frac{2}{3}$ & 3 & {\bf 2} \\
\hline
5 & $\frac{25}{36}$ & 2  & 1 \\
\hline
6 & $\frac{13}{18}$ & 2 & 1  \\
\hline
7 & $\frac{3}{4}$ & 2 & 1  \\
\hline
8 & $\frac{29}{36}$ & 2 & 1  \\
\hline
9 & $\frac{8}{9}$ & 1 & 1  \\
\hline
10 & 1 & 1 & 1 \\
\hline
\hline
Total & {} &  18 & 11 \\
\hline
\end{tabular}
\label{Tab:epgtP9}
\end{table}

The entangling class $e_p(U)=2/3$ is special and has two LUS classes. Representative permutations written in compact form as $P=\left\lbrace 1,4,8,2,5,7,6,3,9  \right\rbrace$ and $P'=\left\lbrace 1,4,9,2,5,8,6,3,7  \right\rbrace$ from both LUS classes produce distinguishable entanglement distributions shown in Fig.~(\ref{fig:LUSP9}). The LU inequivalence between $P$ and $P'$ can also be seen via the singular values of $P^{\Gamma}$ and $P^{' \Gamma}$ which are LUI's (see Eq.~(\ref{eq:EUS})); singular values of $P^{\Gamma}$ and $P^{' \Gamma}$ are $\left\lbrace 2,2,1 \right\rbrace$ and $\left\lbrace\sqrt{5},\sqrt{2},\sqrt{2}\right\rbrace$ respectively. This leads us to the strong suspicion that the equality of entanglement distributions may be a sufficient condition for LUS equivalence, {\em i.e.}, $p(x;U)=p(x;U')$ iff $U \stackrel{LUS}{\sim} U'$. 

An interesting fact we observed is that with von Neumann entropy as a measure of entanglement the average of distributions obtained for these permutations differ slightly; for $P$, $\overline{\mathcal{E}_v (\rho_A)} \approx 0.57$ while for $P'$ $\overline{\mathcal{E}_v (\rho_A)} \approx 0.55$, taking into acccount $10^6$ realisations of product states in both cases. Note that if the linear entropy is taken as a measure then the averages must be equal according to the definition of the entangling power \cite{Zanardi2001}. This suggests the role of other unknown LU invariants besides $E(U)$ and $E(US)$ which determine the average of the entanglement distribution when von Neumann entropy is taken as a measure of entanglement.

\begin{figure}
\centering
\includegraphics[scale=0.6]{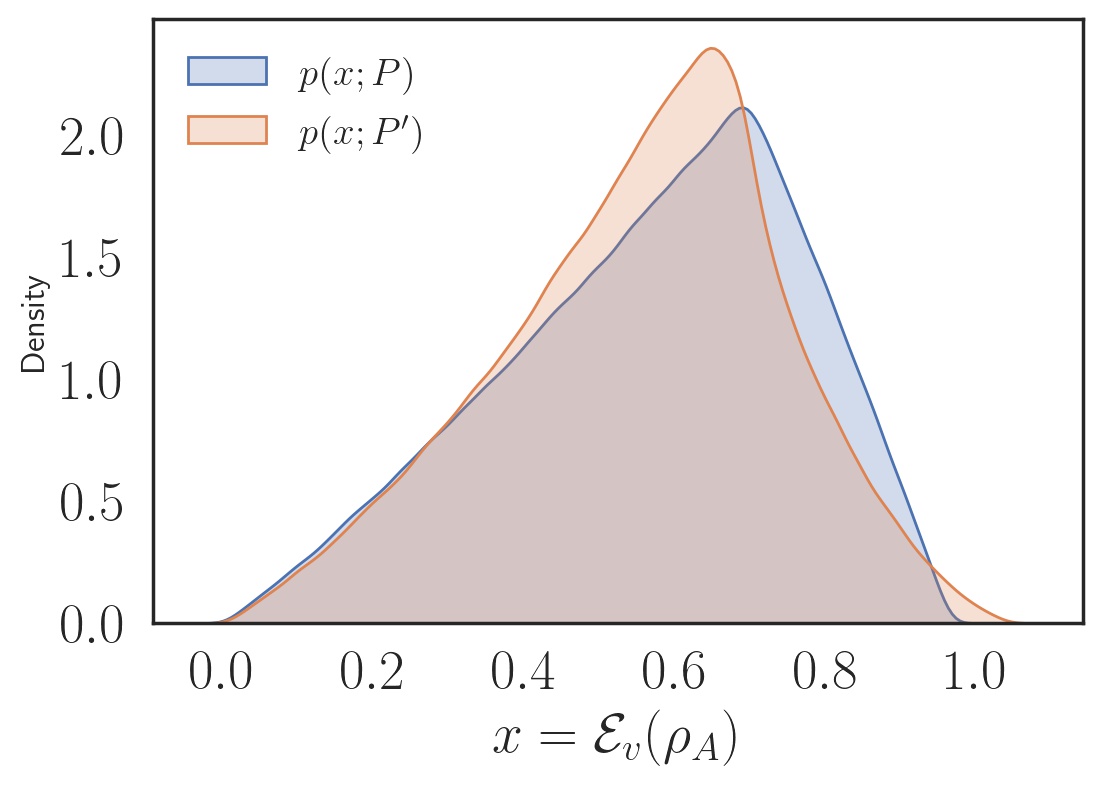}
\caption{Entanglement distributions obtained for $P$ and $P'$ permutation matrices of size $9$ belonging to two LUS classes, see text, corresponding to the entangling class with $e_p(U)=2/3$.  Distinguishabiliity of the distributions imply that $P$ and $P'$ are not LU equivalent.}
\label{fig:LUSP9}
\end{figure}

It is to be noted that out of 18 possible LU classes {\em only} one corresponds to the entangling class $e_p(P)=1$ of 2-unitary permutations. As a consequence of this all 72 possible 2-unitary permutations of order $9$ are locally equivalent consistent with Proposition \ref{prop:P9LUclass}.

\subsection{Numerical results for $d>3$}

Total possible number of LU classes of dual or, T-dual permutations in $\mathrm{P}(d^2)$ for $d>3$ is not known. The number of entangling classes is also not known, as an exhaustive enumeration of such permuatations is prohibitively large. To get a lower bound on the possible number of entangling classes corresponding to dual-unitary permutations of size 16 we numerically search over permutations in the vicinity of different permutations like {\sc swap} and 2-unitary gates. Results obtained from such a search over around $1.2 \times 10^7$ permutations of size $16$ (out of a possible $16!\sim 10^{13}$) is shown in Fig.~(\ref{fig:epgtP16}). We obtain $56$ entangling classes corresponding to dual or equivalently T-dual permutations. This provides a weak lower bound on the number of LU inequivalent classes for dual-unitary permutations of size $16$. Note that one of the entangling classes is $e_p(U)=1$ corresponding to 2-unitary permutations for which there is only one LU equivalence class (see Proposition \ref{prop:P16LUclass}).

\begin{figure}
\centering
\includegraphics[scale=0.6]{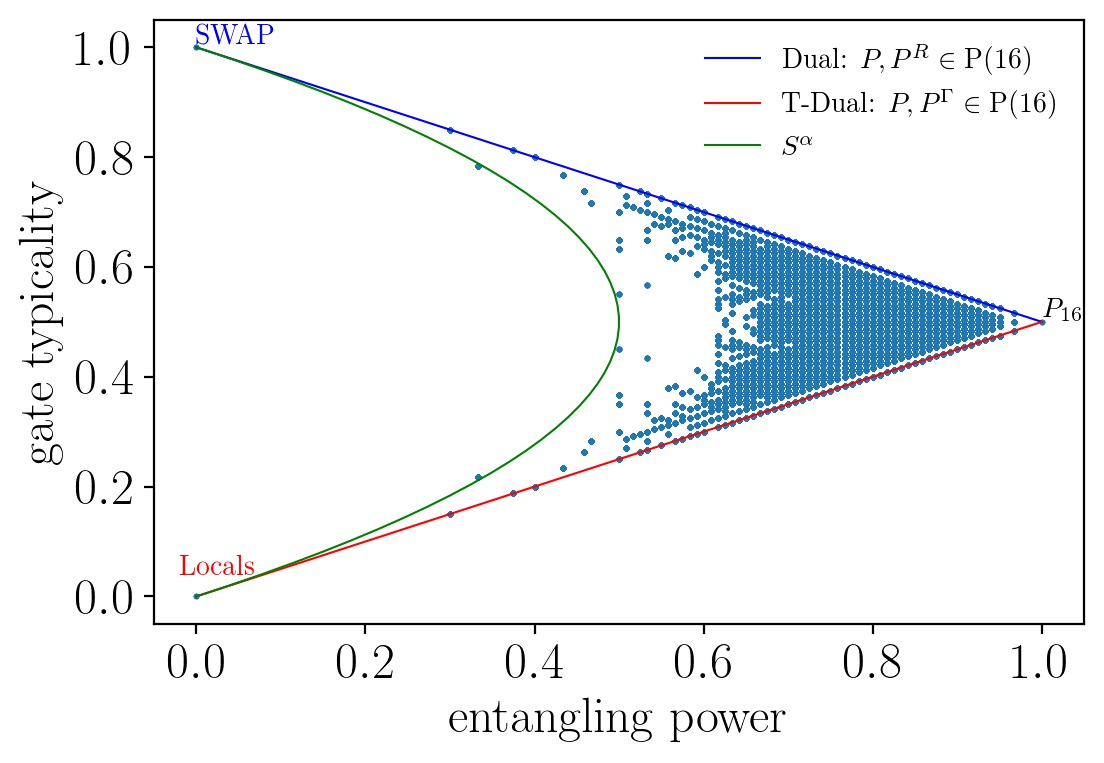}
\caption{Two-ququad case, $d^2=16$. The entangling power {\it vs}  gate-typicality of $1.2 \times 10^7$ permutations (out of total $16! \approx 10^{13}$ permutations), treated as two-ququad gates. Number of entangling classes corresponding to dual (and equivalently T-dual) permutation matrices obtained from our numerical search is found to be $56$.}
\label{fig:epgtP16}
\end{figure}

We end this section by showing that there exist more than one LU classes for 2-unitary permutations in $d>4$. An easy way to see this is by comparing entanglement distributions of 2-unitary permutation $P \in \mathrm{P}(d^2)$ and its realignment $P^R \in \mathrm{P}(d^2)$. This is shown in Fig.~(\ref{fig:entdistP25}) for a 2-unitary permutation in $d=5$ given by
 
\begin{widetext}
\beq
\begin{split}
P_{25}&=\left\lbrace  1,  7, 13, 19, 25, 22,  3,  9, 15, 16, 18, 24,  5,  6, 12, 14, 20, 21,  2,  8, 10, 11, 17, 23,  4 \right\rbrace,\\
P_{25}^R&=\left\lbrace  1,  7, 13, 19, 25,  8, 14, 20, 21,  2, 15, 16, 22,  3,  9, 17, 23,
        4, 10, 11, 24,  5,  6, 12, 18 \right\rbrace.     
\end{split}
\label{eq:P25}
\eeq
\end{widetext}
Entanglement distributions are different; $p(x;P_{25}) \neq p(x;P_{25}^R)$ and thus establishes that they are not LU equivalent. Recall that one can cannot jusify LU inequivalence between 2-unitaries based on the singular values of their reshuffled and partially transposed rearrangements as they are all maximized being equal to $1$. We checked entanglement distributions for $100$ 2-unitary permutations of size $25$ together with their different rearrangements but found only two different distributions shown in Fig.~(\ref{fig:entdistP25}). Thus total number of LU and LUS classes for 2-unitary permutations in $d=5$ remains unknown but it is certainly greater than $1$. Similarly in $d=7,8,$ and $9$ we observed only two different distributions corresponding to 2-unitary permutations and their rearrangements. 

A consequence of having more than one LU classes of 2-unitary permutations in $d>3$ results in minimal support AME states which are not LU equivalent. Our results thus contradict the {\em Conjecture 2} in Ref.~\cite{Adam_SLOCC_2020} which in particular for four party states implies that there is only one LU class of AME states of minimal support. As we have illustrated in Fig.~(\ref{fig:ent_dist_d_4}) and Fig.~(\ref{fig:entdistP25}) there exist more than one LU classes of AME states of minimal support for $d=4$ and $5$.
%\begin{prop}
%There exist more than one LU equivalent classes of 2-unitary permutations in $d>4$
%\end{prop}

Assuming that 2-unitary permutations belonging to the same LU class are always related by local permutations one can argue that there are more than one LU classes of 2-unitary permutations for $d=7,8,$ and $9$ as the number of possible  orthogonal Latin squares (see A072377 , Ref. \cite{OEIS}) exceeds $(d!)^4$ number of orthogonal Latin squares obtained from 2-unitary permutation $P$ using local permutations $p_i$'s of size $d$ on both sides; $(p_1 \otimes p_2) P (p_3 \otimes p_4)$. Note that all $(d!)^4$ number of orthogonal Latin squares so obtained are not all different but even with redundancy this is less than the number of possible  orthogonal Latin squares in $d>5$. Interestingly in $d=5$ we found that although $(d!)^4$ exceeds the number of possible orthogonal Latin squares but it is not a factor of it which is the case for $d=3$ and $4$.

\begin{figure}[h]
\includegraphics[scale=0.65]{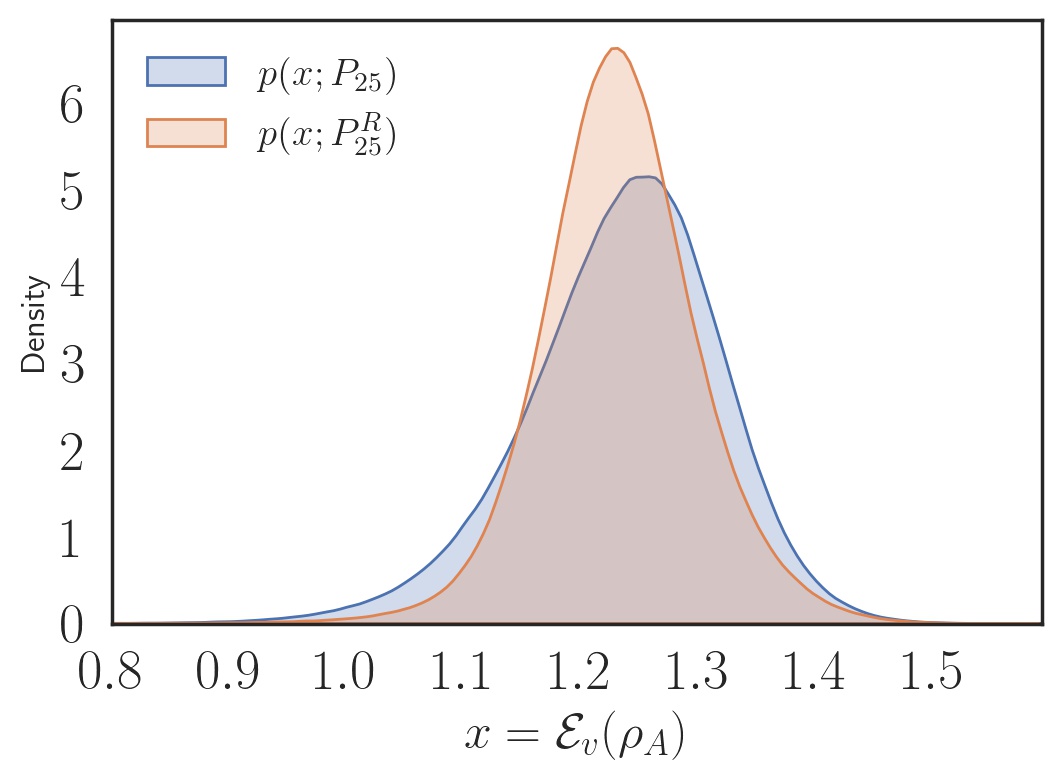}
\caption{$d^2=25$. Entanglement distributions of 2-unitary permutation $P_{25}$, Eq.~(\ref{eq:P25}), and it's realignment $P_{25}^R$ which is also a 2-unitary. The distributions are clearly distinguishable; $p(x;P_{25}) \neq p(x;P_{25}^R)$, and thus shows that $P_{25} \stackrel{LU} {\nsim} P_{25}^R$.}
\label{fig:entdistP25}
\end{figure}

\section{Summary and discussions\label{sec:conc}}

Despite several constructions of the set of dual-unitary operators, the complete characterization for arbitrary local Hilbert space dimension remains an open problem. In Ref.~\cite{SAA2020}, we proposed a non-linear iterative map that produces dual-unitary operators from arbitrary seed unitaries. The map acts as a dynamical system in the space of bipartite unitary operators. In this work, we have studied the period-two fixed points of the map which are dual-unitary operators and provided a stochastic generalization of the map which produces structured fixed points which are dual unitaries. Complete characterization of the fixed points of all orders remains to be understood and makes the map a novel dynamical system in its own right. For two-qubit gates, using the canonical or Cartan decomposition we analytically study the convergence rates for various initial conditions. However, convergence of the map in local Hilbert space dimension $d>2$ remains an unsolved problem. 

The subset of dual-unitary operators having maximum entangling power is that of 2-unitary operators. The 2-unitary permutation operators can be constructed from combinatorial designs called orthogonal Latin squares (OLS). The non-existence of OLS of size 6 motivates to look for general quantum combinatorial designs corresponding to 2-unitary operators as was recently found in Ref.~\cite{SRatherAME46} for local dimension $d=6$. The problem of finding such quantum combinatorial designs reduces to finding the 2-unitary operators which are not LU equivalent to any 2-unitary permutation matrix. From our extensive numerical searches using the dynamical map and known constructions of 2-unitaries, we could not find any such quantum design for local dimension $d=3$. All 2-unitary permutation operators of size $9$ are local unitarily (LU) equivalent to each other. Based on these results we conjecture that {\em all} 2-unitary operators of size $9$, not just permutations, are LU equivalent to each other. If true,  this implies that there is just {\em one} 2-unitary two-qutrit gate up to LU equivalence. 

Methods to ascertain local unitary (LU) equivalence between bipartite unitary operators is not known in general. For unitary operators with identical values of known local unitary invariants (LUI's) like the entangling power and the gate-typicality, the problem of LU equivalence becomes harder. In this paper we have proposed a necessary criterion for distinguishing LU inequivalent 2-unitary operators based on the entanglement distribution these produce. Using the iterative map we found a 2-unitary operator for local dimension $d=4$ which is LU inequivalent to any 2-unitary permutation of the same size. Thus this qualifies as a genuine 2-unitary quantum design in the lowest possible dimension, as they do not exist for $d=2$ and as far as we know, for $d=3$. This also implies that we have displayed an explicit example of an AME($4,4$) state that is not LU equivalent to AME($4,4$) of minimal support. We have shown that for $d=5$ there are at least two LU classes of 2-unitary permutations and thus there are two LU inequivalent AME states of minimal support. Consequences of these new examples of AME states for quantum error-correction is an interesting direction and is left for future studies. Stochastic local operations and classical communication (SLOCC) equivalence of LU inequivalent four party AME states found in this work for $d>3$ is an interesting problem and is left for future studies.

 {\em Note:} After this paper was written, we came to know of Ref. \cite{VijayK} wherein a criterion for determining local unitary equivalence of operators is presented that involves an exponential (in local dimension) set of invariants.

\begin{acknowledgments}
We are grateful to Bal\'azs Pozsgay for discussions on dual unitarity, and its connections to biunitarity. We thank Karol \.Zyczkowski and Adam Burchardt for comments on a preliminary version, and for illuminating remarks on the issue of LU equivalence. We are thankful to Vijay Kodiyalam for pointing out \cite{VijayK} and 
discussions around it. 
Rohan Narayan and Shrigyan Brahmachari's inputs and questions were much appreciated. This work was partially funded
by the Center for Quantum Information Theory in Matter and Spacetime, IIT Madras, 
and the Department of Science and Technology, Govt. of India, under Grant No.
DST/ICPS/QuST/Theme-3/2019/Q69. SA acknowledges the Institute postdoctoral fellowship
program of IIT Madras for funding during the initial stages of this work.
\end{acknowledgments}

%\newpage
\bibliography{ent_local}

%merlin.mbs apsrev4-1.bst 2010-07-25 4.21a (PWD, AO, DPC) hacked
%Control: key (0)
%Control: author (0) dotless jnrlst
%Control: editor formatted (1) identically to author
%Control: production of article title (0) allowed
%Control: page (1) range
%Control: year (0) verbatim
%Control: production of eprint (0) enabled
\begin{thebibliography}{96}%
\makeatletter
\providecommand \@ifxundefined [1]{%
 \@ifx{#1\undefined}
}%
\providecommand \@ifnum [1]{%
 \ifnum #1\expandafter \@firstoftwo
 \else \expandafter \@secondoftwo
 \fi
}%
\providecommand \@ifx [1]{%
 \ifx #1\expandafter \@firstoftwo
 \else \expandafter \@secondoftwo
 \fi
}%
\providecommand \natexlab [1]{#1}%
\providecommand \enquote  [1]{``#1''}%
\providecommand \bibnamefont  [1]{#1}%
\providecommand \bibfnamefont [1]{#1}%
\providecommand \citenamefont [1]{#1}%
\providecommand \href@noop [0]{\@secondoftwo}%
\providecommand \href [0]{\begingroup \@sanitize@url \@href}%
\providecommand \@href[1]{\@@startlink{#1}\@@href}%
\providecommand \@@href[1]{\endgroup#1\@@endlink}%
\providecommand \@sanitize@url [0]{\catcode `\\12\catcode `\$12\catcode
  `\&12\catcode `\#12\catcode `\^12\catcode `\_12\catcode `\%12\relax}%
\providecommand \@@startlink[1]{}%
\providecommand \@@endlink[0]{}%
\providecommand \url  [0]{\begingroup\@sanitize@url \@url }%
\providecommand \@url [1]{\endgroup\@href {#1}{\urlprefix }}%
\providecommand \urlprefix  [0]{URL }%
\providecommand \Eprint [0]{\href }%
\providecommand \doibase [0]{http://dx.doi.org/}%
\providecommand \selectlanguage [0]{\@gobble}%
\providecommand \bibinfo  [0]{\@secondoftwo}%
\providecommand \bibfield  [0]{\@secondoftwo}%
\providecommand \translation [1]{[#1]}%
\providecommand \BibitemOpen [0]{}%
\providecommand \bibitemStop [0]{}%
\providecommand \bibitemNoStop [0]{.\EOS\space}%
\providecommand \EOS [0]{\spacefactor3000\relax}%
\providecommand \BibitemShut  [1]{\csname bibitem#1\endcsname}%
\let\auto@bib@innerbib\@empty
%</preamble>
\bibitem [{\citenamefont {Amico}\ \emph {et~al.}(2008)\citenamefont {Amico},
  \citenamefont {Fazio}, \citenamefont {Osterloh},\ and\ \citenamefont
  {Vedral}}]{amico2008entanglement}%
  \BibitemOpen
  \bibfield  {author} {\bibinfo {author} {\bibfnamefont {Luigi}\ \bibnamefont
  {Amico}}, \bibinfo {author} {\bibfnamefont {Rosario}\ \bibnamefont {Fazio}},
  \bibinfo {author} {\bibfnamefont {Andreas}\ \bibnamefont {Osterloh}}, \ and\
  \bibinfo {author} {\bibfnamefont {Vlatko}\ \bibnamefont {Vedral}},\
  }\bibfield  {title} {\enquote {\bibinfo {title} {Entanglement in many-body
  systems},}\ }\href {\doibase 10.1103/RevModPhys.80.517} {\bibfield  {journal}
  {\bibinfo  {journal} {Rev. Mod. Phys.}\ }\textbf {\bibinfo {volume} {80}},\
  \bibinfo {pages} {517--576} (\bibinfo {year} {2008})}\BibitemShut {NoStop}%
\bibitem [{\citenamefont {Zeng}\ \emph {et~al.}(2015)\citenamefont {Zeng},
  \citenamefont {Chen}, \citenamefont {Zhou},\ and\ \citenamefont
  {Wen}}]{zeng2015quantum}%
  \BibitemOpen
  \bibfield  {author} {\bibinfo {author} {\bibfnamefont {Bei}\ \bibnamefont
  {Zeng}}, \bibinfo {author} {\bibfnamefont {Xie}\ \bibnamefont {Chen}},
  \bibinfo {author} {\bibfnamefont {Duan-Lu}\ \bibnamefont {Zhou}}, \ and\
  \bibinfo {author} {\bibfnamefont {Xiao-Gang}\ \bibnamefont {Wen}},\
  }\bibfield  {title} {\enquote {\bibinfo {title} {Quantum information meets
  quantum matter -- from quantum entanglement to topological phase in many-body
  systems},}\ }\href {\doibase 10.48550/arXiv.1508.02595} {\bibfield  {journal}
  {\bibinfo  {journal} {arXiv:1508.02595}\ } (\bibinfo {year} {2015}),\
  10.48550/arXiv.1508.02595}\BibitemShut {NoStop}%
\bibitem [{\citenamefont {Jahn}\ and\ \citenamefont
  {Eisert}(2021)}]{jahn2021holographic}%
  \BibitemOpen
  \bibfield  {author} {\bibinfo {author} {\bibfnamefont {Alexander}\
  \bibnamefont {Jahn}}\ and\ \bibinfo {author} {\bibfnamefont {Jens}\
  \bibnamefont {Eisert}},\ }\bibfield  {title} {\enquote {\bibinfo {title}
  {Holographic tensor network models and quantum error correction: a topical
  review},}\ }\href {\doibase 10.1088/2058-9565/ac0293} {\bibfield  {journal}
  {\bibinfo  {journal} {Quantum Science and Technology}\ }\textbf {\bibinfo
  {volume} {6}},\ \bibinfo {pages} {033002} (\bibinfo {year}
  {2021})}\BibitemShut {NoStop}%
\bibitem [{\citenamefont {Kibe}\ \emph {et~al.}(2022)\citenamefont {Kibe},
  \citenamefont {Mandayam},\ and\ \citenamefont
  {Mukhopadhyay}}]{kibe2021holographic}%
  \BibitemOpen
  \bibfield  {author} {\bibinfo {author} {\bibfnamefont {Tanay}\ \bibnamefont
  {Kibe}}, \bibinfo {author} {\bibfnamefont {Prabha}\ \bibnamefont {Mandayam}},
  \ and\ \bibinfo {author} {\bibfnamefont {Ayan}\ \bibnamefont
  {Mukhopadhyay}},\ }\bibfield  {title} {\enquote {\bibinfo {title}
  {Holographic spacetime, black holes and quantum error correcting codes: a
  review},}\ }\href {\doibase 10.1140/epjc/s10052-022-10382-1} {\bibfield
  {journal} {\bibinfo  {journal} {The European Physical Journal C}\ }\textbf
  {\bibinfo {volume} {82}} (\bibinfo {year} {2022}),\
  10.1140/epjc/s10052-022-10382-1}\BibitemShut {NoStop}%
\bibitem [{\citenamefont {Feynman}(1982)}]{feynman2018simulating}%
  \BibitemOpen
  \bibfield  {author} {\bibinfo {author} {\bibfnamefont {Richard~P}\
  \bibnamefont {Feynman}},\ }\bibfield  {title} {\enquote {\bibinfo {title}
  {Simulating physics with computers},}\ }\href {\doibase 10.1007/BF02650179}
  {\bibfield  {journal} {\bibinfo  {journal} {International Journal of
  Theoretical Physics}\ }\textbf {\bibinfo {volume} {21}},\ \bibinfo {pages}
  {467–488} (\bibinfo {year} {1982})}\BibitemShut {NoStop}%
\bibitem [{\citenamefont {Georgescu}\ \emph {et~al.}(2014)\citenamefont
  {Georgescu}, \citenamefont {Ashhab},\ and\ \citenamefont
  {Nori}}]{georgescu2014quantum}%
  \BibitemOpen
  \bibfield  {author} {\bibinfo {author} {\bibfnamefont {I.~M.}\ \bibnamefont
  {Georgescu}}, \bibinfo {author} {\bibfnamefont {S.}~\bibnamefont {Ashhab}}, \
  and\ \bibinfo {author} {\bibfnamefont {Franco}\ \bibnamefont {Nori}},\
  }\bibfield  {title} {\enquote {\bibinfo {title} {Quantum simulation},}\
  }\href {\doibase 10.1103/RevModPhys.86.153} {\bibfield  {journal} {\bibinfo
  {journal} {Rev. Mod. Phys.}\ }\textbf {\bibinfo {volume} {86}},\ \bibinfo
  {pages} {153--185} (\bibinfo {year} {2014})}\BibitemShut {NoStop}%
\bibitem [{\citenamefont {Preskill}(2018)}]{preskill2018quantum}%
  \BibitemOpen
  \bibfield  {author} {\bibinfo {author} {\bibfnamefont {John}\ \bibnamefont
  {Preskill}},\ }\bibfield  {title} {\enquote {\bibinfo {title} {Quantum
  {C}omputing in the {NISQ} era and beyond},}\ }\href {\doibase
  10.22331/q-2018-08-06-79} {\bibfield  {journal} {\bibinfo  {journal}
  {{Quantum}}\ }\textbf {\bibinfo {volume} {2}},\ \bibinfo {pages} {79}
  (\bibinfo {year} {2018})}\BibitemShut {NoStop}%
\bibitem [{\citenamefont {Ippoliti}\ \emph {et~al.}(2021)\citenamefont
  {Ippoliti}, \citenamefont {Kechedzhi}, \citenamefont {Moessner},
  \citenamefont {Sondhi},\ and\ \citenamefont {Khemani}}]{ippoliti2021many}%
  \BibitemOpen
  \bibfield  {author} {\bibinfo {author} {\bibfnamefont {Matteo}\ \bibnamefont
  {Ippoliti}}, \bibinfo {author} {\bibfnamefont {Kostyantyn}\ \bibnamefont
  {Kechedzhi}}, \bibinfo {author} {\bibfnamefont {Roderich}\ \bibnamefont
  {Moessner}}, \bibinfo {author} {\bibfnamefont {S.L.}\ \bibnamefont {Sondhi}},
  \ and\ \bibinfo {author} {\bibfnamefont {Vedika}\ \bibnamefont {Khemani}},\
  }\bibfield  {title} {\enquote {\bibinfo {title} {Many-body physics in the
  nisq era: Quantum programming a discrete time crystal},}\ }\href {\doibase
  10.1103/PRXQuantum.2.030346} {\bibfield  {journal} {\bibinfo  {journal} {PRX
  Quantum}\ }\textbf {\bibinfo {volume} {2}},\ \bibinfo {pages} {030346}
  (\bibinfo {year} {2021})}\BibitemShut {NoStop}%
\bibitem [{\citenamefont {Arute}\ \emph {et~al.}(2019)\citenamefont {Arute},
  \citenamefont {Arya}, \citenamefont {Babbush}, \citenamefont {Bacon} \emph
  {et~al.}}]{QSupreme}%
  \BibitemOpen
  \bibfield  {author} {\bibinfo {author} {\bibfnamefont {Frank}\ \bibnamefont
  {Arute}}, \bibinfo {author} {\bibfnamefont {Kunal}\ \bibnamefont {Arya}},
  \bibinfo {author} {\bibfnamefont {Ryan}\ \bibnamefont {Babbush}}, \bibinfo
  {author} {\bibfnamefont {Dave}\ \bibnamefont {Bacon}},  \emph {et~al.},\
  }\bibfield  {title} {\enquote {\bibinfo {title} {Quantum supremacy using a
  programmable superconducting processor},}\ }\href {\doibase
  10.1038/s41586-019-1666-5} {\bibfield  {journal} {\bibinfo  {journal}
  {Nature}\ }\textbf {\bibinfo {volume} {574}},\ \bibinfo {pages} {505--510}
  (\bibinfo {year} {2019})}\BibitemShut {NoStop}%
\bibitem [{\citenamefont {Zhu}\ \emph {et~al.}(2022)\citenamefont {Zhu},
  \citenamefont {Cao}, \citenamefont {Chen}, \citenamefont {Chen} \emph
  {et~al.}}]{zhu2021quantum}%
  \BibitemOpen
  \bibfield  {author} {\bibinfo {author} {\bibfnamefont {Qingling}\
  \bibnamefont {Zhu}}, \bibinfo {author} {\bibfnamefont {Sirui}\ \bibnamefont
  {Cao}}, \bibinfo {author} {\bibfnamefont {Fusheng}\ \bibnamefont {Chen}},
  \bibinfo {author} {\bibfnamefont {Ming-Cheng}\ \bibnamefont {Chen}},  \emph
  {et~al.},\ }\bibfield  {title} {\enquote {\bibinfo {title} {Quantum
  computational advantage via 60-qubit 24-cycle random circuit sampling},}\
  }\href {\doibase https://doi.org/10.1016/j.scib.2021.10.017} {\bibfield
  {journal} {\bibinfo  {journal} {Science Bulletin}\ }\textbf {\bibinfo
  {volume} {67}},\ \bibinfo {pages} {240--245} (\bibinfo {year}
  {2022})}\BibitemShut {NoStop}%
\bibitem [{\citenamefont {Nahum}\ \emph {et~al.}(2018)\citenamefont {Nahum},
  \citenamefont {Vijay},\ and\ \citenamefont {Haah}}]{Nahum2018}%
  \BibitemOpen
  \bibfield  {author} {\bibinfo {author} {\bibfnamefont {Adam}\ \bibnamefont
  {Nahum}}, \bibinfo {author} {\bibfnamefont {Sagar}\ \bibnamefont {Vijay}}, \
  and\ \bibinfo {author} {\bibfnamefont {Jeongwan}\ \bibnamefont {Haah}},\
  }\bibfield  {title} {\enquote {\bibinfo {title} {Operator spreading in random
  unitary circuits},}\ }\href {\doibase 10.1103/PhysRevX.8.021014} {\bibfield
  {journal} {\bibinfo  {journal} {Phys. Rev. X}\ }\textbf {\bibinfo {volume}
  {8}},\ \bibinfo {pages} {021014} (\bibinfo {year} {2018})}\BibitemShut
  {NoStop}%
\bibitem [{\citenamefont {Nahum}\ \emph {et~al.}(2017)\citenamefont {Nahum},
  \citenamefont {Ruhman}, \citenamefont {Vijay},\ and\ \citenamefont
  {Haah}}]{nahum2017quantum}%
  \BibitemOpen
  \bibfield  {author} {\bibinfo {author} {\bibfnamefont {Adam}\ \bibnamefont
  {Nahum}}, \bibinfo {author} {\bibfnamefont {Jonathan}\ \bibnamefont
  {Ruhman}}, \bibinfo {author} {\bibfnamefont {Sagar}\ \bibnamefont {Vijay}}, \
  and\ \bibinfo {author} {\bibfnamefont {Jeongwan}\ \bibnamefont {Haah}},\
  }\bibfield  {title} {\enquote {\bibinfo {title} {Quantum entanglement growth
  under random unitary dynamics},}\ }\href {\doibase 10.1103/PhysRevX.7.031016}
  {\bibfield  {journal} {\bibinfo  {journal} {Phys. Rev. X}\ }\textbf {\bibinfo
  {volume} {7}},\ \bibinfo {pages} {031016} (\bibinfo {year}
  {2017})}\BibitemShut {NoStop}%
\bibitem [{\citenamefont {Khemani}\ \emph {et~al.}(2018)\citenamefont
  {Khemani}, \citenamefont {Vishwanath},\ and\ \citenamefont
  {Huse}}]{khemani2018operator}%
  \BibitemOpen
  \bibfield  {author} {\bibinfo {author} {\bibfnamefont {Vedika}\ \bibnamefont
  {Khemani}}, \bibinfo {author} {\bibfnamefont {Ashvin}\ \bibnamefont
  {Vishwanath}}, \ and\ \bibinfo {author} {\bibfnamefont {David~A.}\
  \bibnamefont {Huse}},\ }\bibfield  {title} {\enquote {\bibinfo {title}
  {Operator spreading and the emergence of dissipative hydrodynamics under
  unitary evolution with conservation laws},}\ }\href {\doibase
  10.1103/PhysRevX.8.031057} {\bibfield  {journal} {\bibinfo  {journal} {Phys.
  Rev. X}\ }\textbf {\bibinfo {volume} {8}},\ \bibinfo {pages} {031057}
  (\bibinfo {year} {2018})}\BibitemShut {NoStop}%
\bibitem [{\citenamefont {von Keyserlingk}\ \emph {et~al.}(2018)\citenamefont
  {von Keyserlingk}, \citenamefont {Rakovszky}, \citenamefont {Pollmann},\ and\
  \citenamefont {Sondhi}}]{von2018operator}%
  \BibitemOpen
  \bibfield  {author} {\bibinfo {author} {\bibfnamefont {C.~W.}\ \bibnamefont
  {von Keyserlingk}}, \bibinfo {author} {\bibfnamefont {Tibor}\ \bibnamefont
  {Rakovszky}}, \bibinfo {author} {\bibfnamefont {Frank}\ \bibnamefont
  {Pollmann}}, \ and\ \bibinfo {author} {\bibfnamefont {S.~L.}\ \bibnamefont
  {Sondhi}},\ }\bibfield  {title} {\enquote {\bibinfo {title} {Operator
  hydrodynamics, otocs, and entanglement growth in systems without conservation
  laws},}\ }\href {\doibase 10.1103/PhysRevX.8.021013} {\bibfield  {journal}
  {\bibinfo  {journal} {Phys. Rev. X}\ }\textbf {\bibinfo {volume} {8}},\
  \bibinfo {pages} {021013} (\bibinfo {year} {2018})}\BibitemShut {NoStop}%
\bibitem [{\citenamefont {Chan}\ \emph {et~al.}(2018)\citenamefont {Chan},
  \citenamefont {De~Luca},\ and\ \citenamefont {Chalker}}]{chan2018solution}%
  \BibitemOpen
  \bibfield  {author} {\bibinfo {author} {\bibfnamefont {Amos}\ \bibnamefont
  {Chan}}, \bibinfo {author} {\bibfnamefont {Andrea}\ \bibnamefont {De~Luca}},
  \ and\ \bibinfo {author} {\bibfnamefont {J.~T.}\ \bibnamefont {Chalker}},\
  }\bibfield  {title} {\enquote {\bibinfo {title} {Solution of a minimal model
  for many-body quantum chaos},}\ }\href {\doibase 10.1103/PhysRevX.8.041019}
  {\bibfield  {journal} {\bibinfo  {journal} {Phys. Rev. X}\ }\textbf {\bibinfo
  {volume} {8}},\ \bibinfo {pages} {041019} (\bibinfo {year}
  {2018})}\BibitemShut {NoStop}%
\bibitem [{\citenamefont {Akila}\ \emph {et~al.}(2016)\citenamefont {Akila},
  \citenamefont {Waltner}, \citenamefont {Gutkin},\ and\ \citenamefont
  {Guhr}}]{Akila2016}%
  \BibitemOpen
  \bibfield  {author} {\bibinfo {author} {\bibfnamefont {M}~\bibnamefont
  {Akila}}, \bibinfo {author} {\bibfnamefont {D}~\bibnamefont {Waltner}},
  \bibinfo {author} {\bibfnamefont {B}~\bibnamefont {Gutkin}}, \ and\ \bibinfo
  {author} {\bibfnamefont {T}~\bibnamefont {Guhr}},\ }\bibfield  {title}
  {\enquote {\bibinfo {title} {Particle-time duality in the kicked {I}sing spin
  chain},}\ }\href {\doibase 10.1088/1751-8113/49/37/375101} {\bibfield
  {journal} {\bibinfo  {journal} {Journal of Physics A: Mathematical and
  Theoretical}\ }\textbf {\bibinfo {volume} {49}},\ \bibinfo {pages} {375101}
  (\bibinfo {year} {2016})}\BibitemShut {NoStop}%
\bibitem [{\citenamefont {Bertini}\ \emph
  {et~al.}(2019{\natexlab{a}})\citenamefont {Bertini}, \citenamefont {Kos},\
  and\ \citenamefont {Prosen}}]{Bertini2019}%
  \BibitemOpen
  \bibfield  {author} {\bibinfo {author} {\bibfnamefont {Bruno}\ \bibnamefont
  {Bertini}}, \bibinfo {author} {\bibfnamefont {Pavel}\ \bibnamefont {Kos}}, \
  and\ \bibinfo {author} {\bibfnamefont {Toma{\v{z}}}\ \bibnamefont {Prosen}},\
  }\bibfield  {title} {\enquote {\bibinfo {title} {Exact correlation functions
  for dual-unitary lattice models in $1+1$ dimensions},}\ }\href {\doibase
  10.1103/PhysRevLett.123.210601} {\bibfield  {journal} {\bibinfo  {journal}
  {Phys. Rev. Lett.}\ }\textbf {\bibinfo {volume} {123}},\ \bibinfo {pages}
  {210601} (\bibinfo {year} {2019}{\natexlab{a}})}\BibitemShut {NoStop}%
\bibitem [{\citenamefont {Gutkin}\ \emph {et~al.}(2020)\citenamefont {Gutkin},
  \citenamefont {Braun}, \citenamefont {Akila}, \citenamefont {Waltner},\ and\
  \citenamefont {Guhr}}]{gutkin2020exact}%
  \BibitemOpen
  \bibfield  {author} {\bibinfo {author} {\bibfnamefont {Boris}\ \bibnamefont
  {Gutkin}}, \bibinfo {author} {\bibfnamefont {Petr}\ \bibnamefont {Braun}},
  \bibinfo {author} {\bibfnamefont {Maram}\ \bibnamefont {Akila}}, \bibinfo
  {author} {\bibfnamefont {Daniel}\ \bibnamefont {Waltner}}, \ and\ \bibinfo
  {author} {\bibfnamefont {Thomas}\ \bibnamefont {Guhr}},\ }\bibfield  {title}
  {\enquote {\bibinfo {title} {Exact local correlations in kicked chains},}\
  }\href {\doibase 10.1103/PhysRevB.102.174307} {\bibfield  {journal} {\bibinfo
   {journal} {Phys. Rev. B}\ }\textbf {\bibinfo {volume} {102}},\ \bibinfo
  {pages} {174307} (\bibinfo {year} {2020})}\BibitemShut {NoStop}%
\bibitem [{\citenamefont {Braun}\ \emph {et~al.}(2020)\citenamefont {Braun},
  \citenamefont {Waltner}, \citenamefont {Akila}, \citenamefont {Gutkin},\ and\
  \citenamefont {Guhr}}]{BraunPRE2020}%
  \BibitemOpen
  \bibfield  {author} {\bibinfo {author} {\bibfnamefont {Petr}\ \bibnamefont
  {Braun}}, \bibinfo {author} {\bibfnamefont {Daniel}\ \bibnamefont {Waltner}},
  \bibinfo {author} {\bibfnamefont {Maram}\ \bibnamefont {Akila}}, \bibinfo
  {author} {\bibfnamefont {Boris}\ \bibnamefont {Gutkin}}, \ and\ \bibinfo
  {author} {\bibfnamefont {Thomas}\ \bibnamefont {Guhr}},\ }\bibfield  {title}
  {\enquote {\bibinfo {title} {Transition from quantum chaos to localization in
  spin chains},}\ }\href {\doibase 10.1103/PhysRevE.101.052201} {\bibfield
  {journal} {\bibinfo  {journal} {Phys. Rev. E}\ }\textbf {\bibinfo {volume}
  {101}},\ \bibinfo {pages} {052201} (\bibinfo {year} {2020})}\BibitemShut
  {NoStop}%
\bibitem [{\citenamefont {Bertini}\ \emph {et~al.}(2018)\citenamefont
  {Bertini}, \citenamefont {Kos},\ and\ \citenamefont {Prosen}}]{Bertini2018}%
  \BibitemOpen
  \bibfield  {author} {\bibinfo {author} {\bibfnamefont {Bruno}\ \bibnamefont
  {Bertini}}, \bibinfo {author} {\bibfnamefont {Pavel}\ \bibnamefont {Kos}}, \
  and\ \bibinfo {author} {\bibfnamefont {T.}~\bibnamefont {Prosen}},\
  }\bibfield  {title} {\enquote {\bibinfo {title} {Exact spectral form factor
  in a minimal model of many-body quantum chaos},}\ }\href {\doibase
  10.1103/PhysRevLett.121.264101} {\bibfield  {journal} {\bibinfo  {journal}
  {Phys. Rev. Lett.}\ }\textbf {\bibinfo {volume} {121}},\ \bibinfo {pages}
  {264101} (\bibinfo {year} {2018})}\BibitemShut {NoStop}%
\bibitem [{\citenamefont {Bertini}\ \emph
  {et~al.}(2019{\natexlab{b}})\citenamefont {Bertini}, \citenamefont {Kos},\
  and\ \citenamefont {Prosen}}]{bertini2019entanglement}%
  \BibitemOpen
  \bibfield  {author} {\bibinfo {author} {\bibfnamefont {Bruno}\ \bibnamefont
  {Bertini}}, \bibinfo {author} {\bibfnamefont {Pavel}\ \bibnamefont {Kos}}, \
  and\ \bibinfo {author} {\bibfnamefont {T.}~\bibnamefont {Prosen}},\
  }\bibfield  {title} {\enquote {\bibinfo {title} {Entanglement spreading in a
  minimal model of maximal many-body quantum chaos},}\ }\href {\doibase
  10.1103/PhysRevX.9.021033} {\bibfield  {journal} {\bibinfo  {journal} {Phys.
  Rev. X}\ }\textbf {\bibinfo {volume} {9}},\ \bibinfo {pages} {021033}
  (\bibinfo {year} {2019}{\natexlab{b}})}\BibitemShut {NoStop}%
\bibitem [{\citenamefont {Bertini}\ \emph
  {et~al.}(2020{\natexlab{a}})\citenamefont {Bertini}, \citenamefont {Kos},\
  and\ \citenamefont {Prosen}}]{bertini2019operator}%
  \BibitemOpen
  \bibfield  {author} {\bibinfo {author} {\bibfnamefont {Bruno}\ \bibnamefont
  {Bertini}}, \bibinfo {author} {\bibfnamefont {Pavel}\ \bibnamefont {Kos}}, \
  and\ \bibinfo {author} {\bibfnamefont {Toma{\v z}}\ \bibnamefont {Prosen}},\
  }\bibfield  {title} {\enquote {\bibinfo {title} {{Operator Entanglement in
  Local Quantum Circuits I: Chaotic Dual-Unitary Circuits}},}\ }\href {\doibase
  10.21468/SciPostPhys.8.4.067} {\bibfield  {journal} {\bibinfo  {journal}
  {SciPost Phys.}\ }\textbf {\bibinfo {volume} {8}},\ \bibinfo {pages} {067}
  (\bibinfo {year} {2020}{\natexlab{a}})}\BibitemShut {NoStop}%
\bibitem [{\citenamefont {Kos}\ \emph {et~al.}(2021)\citenamefont {Kos},
  \citenamefont {Bertini},\ and\ \citenamefont {Prosen}}]{kos2020chaos}%
  \BibitemOpen
  \bibfield  {author} {\bibinfo {author} {\bibfnamefont {Pavel}\ \bibnamefont
  {Kos}}, \bibinfo {author} {\bibfnamefont {Bruno}\ \bibnamefont {Bertini}}, \
  and\ \bibinfo {author} {\bibfnamefont {Toma{\v z}}\ \bibnamefont {Prosen}},\
  }\bibfield  {title} {\enquote {\bibinfo {title} {Chaos and ergodicity in
  extended quantum systems with noisy driving},}\ }\href {\doibase
  10.1103/PhysRevLett.126.190601} {\bibfield  {journal} {\bibinfo  {journal}
  {Phys. Rev. Lett.}\ }\textbf {\bibinfo {volume} {126}},\ \bibinfo {pages}
  {190601} (\bibinfo {year} {2021})}\BibitemShut {NoStop}%
\bibitem [{\citenamefont {Lerose}\ \emph {et~al.}(2021)\citenamefont {Lerose},
  \citenamefont {Sonner},\ and\ \citenamefont {Abanin}}]{lerose2020influence}%
  \BibitemOpen
  \bibfield  {author} {\bibinfo {author} {\bibfnamefont {Alessio}\ \bibnamefont
  {Lerose}}, \bibinfo {author} {\bibfnamefont {Michael}\ \bibnamefont
  {Sonner}}, \ and\ \bibinfo {author} {\bibfnamefont {Dmitry~A.}\ \bibnamefont
  {Abanin}},\ }\bibfield  {title} {\enquote {\bibinfo {title} {Influence matrix
  approach to many-body floquet dynamics},}\ }\href {\doibase
  10.1103/PhysRevX.11.021040} {\bibfield  {journal} {\bibinfo  {journal} {Phys.
  Rev. X}\ }\textbf {\bibinfo {volume} {11}},\ \bibinfo {pages} {021040}
  (\bibinfo {year} {2021})}\BibitemShut {NoStop}%
\bibitem [{\citenamefont {Garratt}\ and\ \citenamefont
  {Chalker}(2021)}]{garratt2020many}%
  \BibitemOpen
  \bibfield  {author} {\bibinfo {author} {\bibfnamefont {S.~J.}\ \bibnamefont
  {Garratt}}\ and\ \bibinfo {author} {\bibfnamefont {J.~T.}\ \bibnamefont
  {Chalker}},\ }\bibfield  {title} {\enquote {\bibinfo {title} {Local pairing
  of {F}eynman histories in many-body floquet models},}\ }\href {\doibase
  10.1103/PhysRevX.11.021051} {\bibfield  {journal} {\bibinfo  {journal} {Phys.
  Rev. X}\ }\textbf {\bibinfo {volume} {11}},\ \bibinfo {pages} {021051}
  (\bibinfo {year} {2021})}\BibitemShut {NoStop}%
\bibitem [{\citenamefont {Flack}\ \emph {et~al.}(2020)\citenamefont {Flack},
  \citenamefont {Bertini},\ and\ \citenamefont {Prosen}}]{flack2020statistics}%
  \BibitemOpen
  \bibfield  {author} {\bibinfo {author} {\bibfnamefont {Ana}\ \bibnamefont
  {Flack}}, \bibinfo {author} {\bibfnamefont {Bruno}\ \bibnamefont {Bertini}},
  \ and\ \bibinfo {author} {\bibfnamefont {Toma{\v z}}\ \bibnamefont
  {Prosen}},\ }\bibfield  {title} {\enquote {\bibinfo {title} {Statistics of
  the spectral form factor in the self-dual kicked ising model},}\ }\href
  {\doibase 10.1103/PhysRevResearch.2.043403} {\bibfield  {journal} {\bibinfo
  {journal} {Phys. Rev. Research}\ }\textbf {\bibinfo {volume} {2}},\ \bibinfo
  {pages} {043403} (\bibinfo {year} {2020})}\BibitemShut {NoStop}%
\bibitem [{\citenamefont {Bertini}\ \emph
  {et~al.}(2020{\natexlab{b}})\citenamefont {Bertini}, \citenamefont {Kos},\
  and\ \citenamefont {Prosen}}]{bertini2019operatorII}%
  \BibitemOpen
  \bibfield  {author} {\bibinfo {author} {\bibfnamefont {Bruno}\ \bibnamefont
  {Bertini}}, \bibinfo {author} {\bibfnamefont {Pavel}\ \bibnamefont {Kos}}, \
  and\ \bibinfo {author} {\bibfnamefont {Tomaz}\ \bibnamefont {Prosen}},\
  }\bibfield  {title} {\enquote {\bibinfo {title} {{Operator Entanglement in
  Local Quantum Circuits II: Solitons in Chains of Qubits}},}\ }\href {\doibase
  10.21468/SciPostPhys.8.4.068} {\bibfield  {journal} {\bibinfo  {journal}
  {SciPost Phys.}\ }\textbf {\bibinfo {volume} {8}},\ \bibinfo {pages} {068}
  (\bibinfo {year} {2020}{\natexlab{b}})}\BibitemShut {NoStop}%
\bibitem [{\citenamefont {Piroli}\ \emph {et~al.}(2020)\citenamefont {Piroli},
  \citenamefont {Bertini}, \citenamefont {Cirac},\ and\ \citenamefont
  {Prosen}}]{piroli2020exact}%
  \BibitemOpen
  \bibfield  {author} {\bibinfo {author} {\bibfnamefont {Lorenzo}\ \bibnamefont
  {Piroli}}, \bibinfo {author} {\bibfnamefont {Bruno}\ \bibnamefont {Bertini}},
  \bibinfo {author} {\bibfnamefont {J.~Ignacio}\ \bibnamefont {Cirac}}, \ and\
  \bibinfo {author} {\bibfnamefont {Toma{\v z}}\ \bibnamefont {Prosen}},\
  }\bibfield  {title} {\enquote {\bibinfo {title} {Exact dynamics in
  dual-unitary quantum circuits},}\ }\href {\doibase
  10.1103/PhysRevB.101.094304} {\bibfield  {journal} {\bibinfo  {journal}
  {Phys. Rev. B}\ }\textbf {\bibinfo {volume} {101}},\ \bibinfo {pages}
  {094304} (\bibinfo {year} {2020})}\BibitemShut {NoStop}%
\bibitem [{\citenamefont {Claeys}\ and\ \citenamefont
  {Lamacraft}(2020)}]{claeys2020maximum}%
  \BibitemOpen
  \bibfield  {author} {\bibinfo {author} {\bibfnamefont {Pieter~W.}\
  \bibnamefont {Claeys}}\ and\ \bibinfo {author} {\bibfnamefont {Austen}\
  \bibnamefont {Lamacraft}},\ }\bibfield  {title} {\enquote {\bibinfo {title}
  {Maximum velocity quantum circuits},}\ }\href {\doibase
  10.1103/PhysRevResearch.2.033032} {\bibfield  {journal} {\bibinfo  {journal}
  {Phys. Rev. Research}\ }\textbf {\bibinfo {volume} {2}},\ \bibinfo {pages}
  {033032} (\bibinfo {year} {2020})}\BibitemShut {NoStop}%
\bibitem [{\citenamefont {Klobas}\ \emph {et~al.}(2021)\citenamefont {Klobas},
  \citenamefont {Bertini},\ and\ \citenamefont {Piroli}}]{klobas2021exact}%
  \BibitemOpen
  \bibfield  {author} {\bibinfo {author} {\bibfnamefont {Katja}\ \bibnamefont
  {Klobas}}, \bibinfo {author} {\bibfnamefont {Bruno}\ \bibnamefont {Bertini}},
  \ and\ \bibinfo {author} {\bibfnamefont {Lorenzo}\ \bibnamefont {Piroli}},\
  }\bibfield  {title} {\enquote {\bibinfo {title} {Exact thermalization
  dynamics in the ``rule 54'' quantum cellular automaton},}\ }\href {\doibase
  10.1103/PhysRevLett.126.160602} {\bibfield  {journal} {\bibinfo  {journal}
  {Phys. Rev. Lett.}\ }\textbf {\bibinfo {volume} {126}},\ \bibinfo {pages}
  {160602} (\bibinfo {year} {2021})}\BibitemShut {NoStop}%
\bibitem [{\citenamefont {Ippoliti}\ \emph {et~al.}(2022)\citenamefont
  {Ippoliti}, \citenamefont {Rakovszky},\ and\ \citenamefont
  {Khemani}}]{ippoliti2022fractal}%
  \BibitemOpen
  \bibfield  {author} {\bibinfo {author} {\bibfnamefont {Matteo}\ \bibnamefont
  {Ippoliti}}, \bibinfo {author} {\bibfnamefont {Tibor}\ \bibnamefont
  {Rakovszky}}, \ and\ \bibinfo {author} {\bibfnamefont {Vedika}\ \bibnamefont
  {Khemani}},\ }\bibfield  {title} {\enquote {\bibinfo {title} {Fractal,
  logarithmic, and volume-law entangled nonthermal steady states via spacetime
  duality},}\ }\href {\doibase 10.1103/PhysRevX.12.011045} {\bibfield
  {journal} {\bibinfo  {journal} {Phys. Rev. X}\ }\textbf {\bibinfo {volume}
  {12}},\ \bibinfo {pages} {011045} (\bibinfo {year} {2022})}\BibitemShut
  {NoStop}%
\bibitem [{\citenamefont {Ippoliti}\ and\ \citenamefont
  {Khemani}(2021)}]{ippoliti2021postselection}%
  \BibitemOpen
  \bibfield  {author} {\bibinfo {author} {\bibfnamefont {Matteo}\ \bibnamefont
  {Ippoliti}}\ and\ \bibinfo {author} {\bibfnamefont {Vedika}\ \bibnamefont
  {Khemani}},\ }\bibfield  {title} {\enquote {\bibinfo {title}
  {Postselection-free entanglement dynamics via spacetime duality},}\ }\href
  {\doibase 10.1103/PhysRevLett.126.060501} {\bibfield  {journal} {\bibinfo
  {journal} {Phys. Rev. Lett.}\ }\textbf {\bibinfo {volume} {126}},\ \bibinfo
  {pages} {060501} (\bibinfo {year} {2021})}\BibitemShut {NoStop}%
\bibitem [{\citenamefont {Klobas}\ and\ \citenamefont
  {Bertini}(2021)}]{klobas2021entanglement}%
  \BibitemOpen
  \bibfield  {author} {\bibinfo {author} {\bibfnamefont {Katja}\ \bibnamefont
  {Klobas}}\ and\ \bibinfo {author} {\bibfnamefont {Bruno}\ \bibnamefont
  {Bertini}},\ }\bibfield  {title} {\enquote {\bibinfo {title} {{Entanglement
  dynamics in Rule 54: Exact results and quasiparticle picture}},}\ }\href
  {\doibase 10.21468/SciPostPhys.11.6.107} {\bibfield  {journal} {\bibinfo
  {journal} {SciPost Phys.}\ }\textbf {\bibinfo {volume} {11}},\ \bibinfo
  {pages} {107} (\bibinfo {year} {2021})}\BibitemShut {NoStop}%
\bibitem [{\citenamefont {Claeys}\ and\ \citenamefont
  {Lamacraft}(2022)}]{claeys2022emergent}%
  \BibitemOpen
  \bibfield  {author} {\bibinfo {author} {\bibfnamefont {Pieter~W.}\
  \bibnamefont {Claeys}}\ and\ \bibinfo {author} {\bibfnamefont {Austen}\
  \bibnamefont {Lamacraft}},\ }\bibfield  {title} {\enquote {\bibinfo {title}
  {Emergent quantum state designs and biunitarity in dual-unitary circuit
  dynamics},}\ }\href {\doibase 10.22331/q-2022-06-15-738} {\bibfield
  {journal} {\bibinfo  {journal} {{Quantum}}\ }\textbf {\bibinfo {volume}
  {6}},\ \bibinfo {pages} {738} (\bibinfo {year} {2022})}\BibitemShut {NoStop}%
\bibitem [{\citenamefont {Zhou}\ and\ \citenamefont
  {Harrow}(2022)}]{zhou2022maximal}%
  \BibitemOpen
  \bibfield  {author} {\bibinfo {author} {\bibfnamefont {Tianci}\ \bibnamefont
  {Zhou}}\ and\ \bibinfo {author} {\bibfnamefont {Aram~W.}\ \bibnamefont
  {Harrow}},\ }\bibfield  {title} {\enquote {\bibinfo {title} {Maximal
  entanglement velocity implies dual unitarity},}\ }\href {\doibase
  10.1103/PhysRevB.106.L201104} {\bibfield  {journal} {\bibinfo  {journal}
  {Phys. Rev. B}\ }\textbf {\bibinfo {volume} {106}},\ \bibinfo {pages}
  {L201104} (\bibinfo {year} {2022})}\BibitemShut {NoStop}%
\bibitem [{\citenamefont {Zanardi}\ \emph {et~al.}(2000)\citenamefont
  {Zanardi}, \citenamefont {Zalka},\ and\ \citenamefont {Faoro}}]{Zanardi2000}%
  \BibitemOpen
  \bibfield  {author} {\bibinfo {author} {\bibfnamefont {Paolo}\ \bibnamefont
  {Zanardi}}, \bibinfo {author} {\bibfnamefont {Christof}\ \bibnamefont
  {Zalka}}, \ and\ \bibinfo {author} {\bibfnamefont {Lara}\ \bibnamefont
  {Faoro}},\ }\bibfield  {title} {\enquote {\bibinfo {title} {Entangling power
  of quantum evolutions},}\ }\href {\doibase 10.1103/PhysRevA.62.030301}
  {\bibfield  {journal} {\bibinfo  {journal} {Phys. Rev. A}\ }\textbf {\bibinfo
  {volume} {62}},\ \bibinfo {pages} {030301} (\bibinfo {year}
  {2000})}\BibitemShut {NoStop}%
\bibitem [{\citenamefont {Zanardi}(2001)}]{Zanardi2001}%
  \BibitemOpen
  \bibfield  {author} {\bibinfo {author} {\bibfnamefont {Paolo}\ \bibnamefont
  {Zanardi}},\ }\bibfield  {title} {\enquote {\bibinfo {title} {Entanglement of
  quantum evolutions},}\ }\href {\doibase 10.1103/PhysRevA.63.040304}
  {\bibfield  {journal} {\bibinfo  {journal} {Phys. Rev. A}\ }\textbf {\bibinfo
  {volume} {63}},\ \bibinfo {pages} {040304} (\bibinfo {year}
  {2001})}\BibitemShut {NoStop}%
\bibitem [{\citenamefont {Wang}\ and\ \citenamefont
  {Zanardi}(2002)}]{Wang2002}%
  \BibitemOpen
  \bibfield  {author} {\bibinfo {author} {\bibfnamefont {Xiaoguang}\
  \bibnamefont {Wang}}\ and\ \bibinfo {author} {\bibfnamefont {Paolo}\
  \bibnamefont {Zanardi}},\ }\bibfield  {title} {\enquote {\bibinfo {title}
  {Quantum entanglement of unitary operators on bipartite systems},}\ }\href
  {\doibase 10.1103/PhysRevA.66.044303} {\bibfield  {journal} {\bibinfo
  {journal} {Phys. Rev. A}\ }\textbf {\bibinfo {volume} {66}},\ \bibinfo
  {pages} {044303} (\bibinfo {year} {2002})}\BibitemShut {NoStop}%
\bibitem [{\citenamefont {Wang}\ \emph {et~al.}(2003)\citenamefont {Wang},
  \citenamefont {Sanders},\ and\ \citenamefont {Berry}}]{Wang2003}%
  \BibitemOpen
  \bibfield  {author} {\bibinfo {author} {\bibfnamefont {Xiaoguang}\
  \bibnamefont {Wang}}, \bibinfo {author} {\bibfnamefont {Barry~C.}\
  \bibnamefont {Sanders}}, \ and\ \bibinfo {author} {\bibfnamefont
  {Dominic~W.}\ \bibnamefont {Berry}},\ }\bibfield  {title} {\enquote {\bibinfo
  {title} {Entangling power and operator entanglement in qudit systems},}\
  }\href {\doibase 10.1103/PhysRevA.67.042323} {\bibfield  {journal} {\bibinfo
  {journal} {Phys. Rev. A}\ }\textbf {\bibinfo {volume} {67}},\ \bibinfo
  {pages} {042323} (\bibinfo {year} {2003})}\BibitemShut {NoStop}%
\bibitem [{\citenamefont {Nielsen}\ \emph {et~al.}(2003)\citenamefont
  {Nielsen}, \citenamefont {Dawson}, \citenamefont {Dodd}, \citenamefont
  {Gilchrist}, \citenamefont {Mortimer}, \citenamefont {Osborne}, \citenamefont
  {Bremner}, \citenamefont {Harrow},\ and\ \citenamefont
  {Hines}}]{Nielsen2003}%
  \BibitemOpen
  \bibfield  {author} {\bibinfo {author} {\bibfnamefont {Michael~A.}\
  \bibnamefont {Nielsen}}, \bibinfo {author} {\bibfnamefont {Christopher~M.}\
  \bibnamefont {Dawson}}, \bibinfo {author} {\bibfnamefont {Jennifer~L.}\
  \bibnamefont {Dodd}}, \bibinfo {author} {\bibfnamefont {Alexei}\ \bibnamefont
  {Gilchrist}}, \bibinfo {author} {\bibfnamefont {Duncan}\ \bibnamefont
  {Mortimer}}, \bibinfo {author} {\bibfnamefont {Tobias~J.}\ \bibnamefont
  {Osborne}}, \bibinfo {author} {\bibfnamefont {Michael~J.}\ \bibnamefont
  {Bremner}}, \bibinfo {author} {\bibfnamefont {Aram~W.}\ \bibnamefont
  {Harrow}}, \ and\ \bibinfo {author} {\bibfnamefont {Andrew}\ \bibnamefont
  {Hines}},\ }\bibfield  {title} {\enquote {\bibinfo {title} {Quantum dynamics
  as a physical resource},}\ }\href {\doibase 10.1103/PhysRevA.67.052301}
  {\bibfield  {journal} {\bibinfo  {journal} {Phys. Rev. A}\ }\textbf {\bibinfo
  {volume} {67}},\ \bibinfo {pages} {052301} (\bibinfo {year}
  {2003})}\BibitemShut {NoStop}%
\bibitem [{\citenamefont {Vidal}\ and\ \citenamefont
  {Cirac}(2002)}]{Vidal2002}%
  \BibitemOpen
  \bibfield  {author} {\bibinfo {author} {\bibfnamefont {G.}~\bibnamefont
  {Vidal}}\ and\ \bibinfo {author} {\bibfnamefont {J.~I.}\ \bibnamefont
  {Cirac}},\ }\bibfield  {title} {\enquote {\bibinfo {title} {Catalysis in
  nonlocal quantum operations},}\ }\href {\doibase
  10.1103/PhysRevLett.88.167903} {\bibfield  {journal} {\bibinfo  {journal}
  {Phys. Rev. Lett.}\ }\textbf {\bibinfo {volume} {88}},\ \bibinfo {pages}
  {167903} (\bibinfo {year} {2002})}\BibitemShut {NoStop}%
\bibitem [{\citenamefont {Hammerer}\ \emph {et~al.}(2002)\citenamefont
  {Hammerer}, \citenamefont {Vidal},\ and\ \citenamefont
  {Cirac}}]{Hammerer2002}%
  \BibitemOpen
  \bibfield  {author} {\bibinfo {author} {\bibfnamefont {K.}~\bibnamefont
  {Hammerer}}, \bibinfo {author} {\bibfnamefont {G.}~\bibnamefont {Vidal}}, \
  and\ \bibinfo {author} {\bibfnamefont {J.~I.}\ \bibnamefont {Cirac}},\
  }\bibfield  {title} {\enquote {\bibinfo {title} {Characterization of nonlocal
  gates},}\ }\href {\doibase 10.1103/PhysRevA.66.062321} {\bibfield  {journal}
  {\bibinfo  {journal} {Phys. Rev. A}\ }\textbf {\bibinfo {volume} {66}},\
  \bibinfo {pages} {062321} (\bibinfo {year} {2002})}\BibitemShut {NoStop}%
\bibitem [{\citenamefont {Collins}\ \emph {et~al.}(2001)\citenamefont
  {Collins}, \citenamefont {Linden},\ and\ \citenamefont
  {Popescu}}]{Collins2001}%
  \BibitemOpen
  \bibfield  {author} {\bibinfo {author} {\bibfnamefont {Daniel}\ \bibnamefont
  {Collins}}, \bibinfo {author} {\bibfnamefont {Noah}\ \bibnamefont {Linden}},
  \ and\ \bibinfo {author} {\bibfnamefont {Sandu}\ \bibnamefont {Popescu}},\
  }\bibfield  {title} {\enquote {\bibinfo {title} {Nonlocal content of quantum
  operations},}\ }\href {\doibase 10.1103/PhysRevA.64.032302} {\bibfield
  {journal} {\bibinfo  {journal} {Phys. Rev. A}\ }\textbf {\bibinfo {volume}
  {64}},\ \bibinfo {pages} {032302} (\bibinfo {year} {2001})}\BibitemShut
  {NoStop}%
\bibitem [{\citenamefont {Eisert}\ \emph {et~al.}(2000)\citenamefont {Eisert},
  \citenamefont {Jacobs}, \citenamefont {Papadopoulos},\ and\ \citenamefont
  {Plenio}}]{Eisert2000}%
  \BibitemOpen
  \bibfield  {author} {\bibinfo {author} {\bibfnamefont {J.}~\bibnamefont
  {Eisert}}, \bibinfo {author} {\bibfnamefont {K.}~\bibnamefont {Jacobs}},
  \bibinfo {author} {\bibfnamefont {P.}~\bibnamefont {Papadopoulos}}, \ and\
  \bibinfo {author} {\bibfnamefont {M.~B.}\ \bibnamefont {Plenio}},\ }\bibfield
   {title} {\enquote {\bibinfo {title} {Optimal local implementation of
  nonlocal quantum gates},}\ }\href {\doibase 10.1103/PhysRevA.62.052317}
  {\bibfield  {journal} {\bibinfo  {journal} {Phys. Rev. A}\ }\textbf {\bibinfo
  {volume} {62}},\ \bibinfo {pages} {052317} (\bibinfo {year}
  {2000})}\BibitemShut {NoStop}%
\bibitem [{\citenamefont {Cirac}\ \emph {et~al.}(2001)\citenamefont {Cirac},
  \citenamefont {D\"ur}, \citenamefont {Kraus},\ and\ \citenamefont
  {Lewenstein}}]{cirac2001}%
  \BibitemOpen
  \bibfield  {author} {\bibinfo {author} {\bibfnamefont {J.~I.}\ \bibnamefont
  {Cirac}}, \bibinfo {author} {\bibfnamefont {W.}~\bibnamefont {D\"ur}},
  \bibinfo {author} {\bibfnamefont {B.}~\bibnamefont {Kraus}}, \ and\ \bibinfo
  {author} {\bibfnamefont {M.}~\bibnamefont {Lewenstein}},\ }\bibfield  {title}
  {\enquote {\bibinfo {title} {Entangling operations and their implementation
  using a small amount of entanglement},}\ }\href {\doibase
  10.1103/PhysRevLett.86.544} {\bibfield  {journal} {\bibinfo  {journal} {Phys.
  Rev. Lett.}\ }\textbf {\bibinfo {volume} {86}},\ \bibinfo {pages} {544--547}
  (\bibinfo {year} {2001})}\BibitemShut {NoStop}%
\bibitem [{\citenamefont {Nielsen}\ \emph {et~al.}(2006)\citenamefont
  {Nielsen}, \citenamefont {Dowling}, \citenamefont {Gu},\ and\ \citenamefont
  {Doherty}}]{dowling2008geometry}%
  \BibitemOpen
  \bibfield  {author} {\bibinfo {author} {\bibfnamefont {Michael~A.}\
  \bibnamefont {Nielsen}}, \bibinfo {author} {\bibfnamefont {Mark~R.}\
  \bibnamefont {Dowling}}, \bibinfo {author} {\bibfnamefont {Mile}\
  \bibnamefont {Gu}}, \ and\ \bibinfo {author} {\bibfnamefont {Andrew~C.}\
  \bibnamefont {Doherty}},\ }\bibfield  {title} {\enquote {\bibinfo {title}
  {Quantum computation as geometry},}\ }\href {\doibase
  10.1126/science.1121541} {\bibfield  {journal} {\bibinfo  {journal}
  {Science}\ }\textbf {\bibinfo {volume} {311}},\ \bibinfo {pages} {1133--1135}
  (\bibinfo {year} {2006})}\BibitemShut {NoStop}%
\bibitem [{\citenamefont {Rather}\ \emph {et~al.}(2020)\citenamefont {Rather},
  \citenamefont {Aravinda},\ and\ \citenamefont {Lakshminarayan}}]{SAA2020}%
  \BibitemOpen
  \bibfield  {author} {\bibinfo {author} {\bibfnamefont {Suhail~Ahmad}\
  \bibnamefont {Rather}}, \bibinfo {author} {\bibfnamefont {S.}~\bibnamefont
  {Aravinda}}, \ and\ \bibinfo {author} {\bibfnamefont {Arul}\ \bibnamefont
  {Lakshminarayan}},\ }\bibfield  {title} {\enquote {\bibinfo {title} {Creating
  ensembles of dual unitary and maximally entangling quantum evolutions},}\
  }\href {\doibase 10.1103/PhysRevLett.125.070501} {\bibfield  {journal}
  {\bibinfo  {journal} {Phys. Rev. Lett.}\ }\textbf {\bibinfo {volume} {125}},\
  \bibinfo {pages} {070501} (\bibinfo {year} {2020})}\BibitemShut {NoStop}%
\bibitem [{\citenamefont {Goyeneche}\ \emph {et~al.}(2015)\citenamefont
  {Goyeneche}, \citenamefont {Alsina}, \citenamefont {Latorre}, \citenamefont
  {Riera},\ and\ \citenamefont {\ifmmode~\dot{Z}\else
  \.{Z}\fi{}yczkowski}}]{Goyeneche2015}%
  \BibitemOpen
  \bibfield  {author} {\bibinfo {author} {\bibfnamefont {Dardo}\ \bibnamefont
  {Goyeneche}}, \bibinfo {author} {\bibfnamefont {Daniel}\ \bibnamefont
  {Alsina}}, \bibinfo {author} {\bibfnamefont {Jos\'e~I.}\ \bibnamefont
  {Latorre}}, \bibinfo {author} {\bibfnamefont {Arnau}\ \bibnamefont {Riera}},
  \ and\ \bibinfo {author} {\bibfnamefont {Karol}\ \bibnamefont
  {\ifmmode~\dot{Z}\else \.{Z}\fi{}yczkowski}},\ }\bibfield  {title} {\enquote
  {\bibinfo {title} {Absolutely maximally entangled states, combinatorial
  designs, and multiunitary matrices},}\ }\href {\doibase
  10.1103/PhysRevA.92.032316} {\bibfield  {journal} {\bibinfo  {journal} {Phys.
  Rev. A}\ }\textbf {\bibinfo {volume} {92}},\ \bibinfo {pages} {032316}
  (\bibinfo {year} {2015})}\BibitemShut {NoStop}%
\bibitem [{\citenamefont {Pastawski}\ \emph {et~al.}(2015)\citenamefont
  {Pastawski}, \citenamefont {Yoshida}, \citenamefont {Harlow},\ and\
  \citenamefont {Preskill}}]{Pastawski2015}%
  \BibitemOpen
  \bibfield  {author} {\bibinfo {author} {\bibfnamefont {Fernando}\
  \bibnamefont {Pastawski}}, \bibinfo {author} {\bibfnamefont {Beni}\
  \bibnamefont {Yoshida}}, \bibinfo {author} {\bibfnamefont {Daniel}\
  \bibnamefont {Harlow}}, \ and\ \bibinfo {author} {\bibfnamefont {John}\
  \bibnamefont {Preskill}},\ }\bibfield  {title} {\enquote {\bibinfo {title}
  {Holographic quantum error-correcting codes: toy models for the bulk/boundary
  correspondence},}\ }\href {\doibase 10.1007/JHEP06(2015)149} {\bibfield
  {journal} {\bibinfo  {journal} {Journal of High Energy Physics}\ }\textbf
  {\bibinfo {volume} {2015}},\ \bibinfo {pages} {149} (\bibinfo {year}
  {2015})}\BibitemShut {NoStop}%
\bibitem [{\citenamefont {Aravinda}\ \emph {et~al.}(2021)\citenamefont
  {Aravinda}, \citenamefont {Rather},\ and\ \citenamefont
  {Lakshminarayan}}]{ASA_2021}%
  \BibitemOpen
  \bibfield  {author} {\bibinfo {author} {\bibfnamefont {S.}~\bibnamefont
  {Aravinda}}, \bibinfo {author} {\bibfnamefont {Suhail~Ahmad}\ \bibnamefont
  {Rather}}, \ and\ \bibinfo {author} {\bibfnamefont {Arul}\ \bibnamefont
  {Lakshminarayan}},\ }\bibfield  {title} {\enquote {\bibinfo {title} {From
  dual-unitary to quantum {B}ernoulli circuits: Role of the entangling power in
  constructing a quantum ergodic hierarchy},}\ }\href {\doibase
  10.1103/PhysRevResearch.3.043034} {\bibfield  {journal} {\bibinfo  {journal}
  {Phys. Rev. Research}\ }\textbf {\bibinfo {volume} {3}},\ \bibinfo {pages}
  {043034} (\bibinfo {year} {2021})}\BibitemShut {NoStop}%
\bibitem [{\citenamefont {Tyson}(2003{\natexlab{a}})}]{Tyson_2003}%
  \BibitemOpen
  \bibfield  {author} {\bibinfo {author} {\bibfnamefont {Jon~E}\ \bibnamefont
  {Tyson}},\ }\bibfield  {title} {\enquote {\bibinfo {title} {Operator-schmidt
  decompositions and the fourier transform, with applications to the
  operator-schmidt numbers of unitaries},}\ }\href {\doibase
  10.1088/0305-4470/36/39/309} {\bibfield  {journal} {\bibinfo  {journal}
  {Journal of Physics A: Mathematical and General}\ }\textbf {\bibinfo {volume}
  {36}},\ \bibinfo {pages} {10101--10114} (\bibinfo {year}
  {2003}{\natexlab{a}})}\BibitemShut {NoStop}%
\bibitem [{\citenamefont {Tyson}(2003{\natexlab{b}})}]{Tyson2003}%
  \BibitemOpen
  \bibfield  {author} {\bibinfo {author} {\bibfnamefont {Jon}\ \bibnamefont
  {Tyson}},\ }\bibfield  {title} {\enquote {\bibinfo {title} {Operator-schmidt
  decomposition of the quantum fourier transform on $\mathbb{C}^{N_1} \otimes
  \mathbb{C}^{N_2}$},}\ }\href {\doibase 10.1088/0305-4470/36/24/317}
  {\bibfield  {journal} {\bibinfo  {journal} {Journal of Physics A:
  Mathematical and General}\ }\textbf {\bibinfo {volume} {36}},\ \bibinfo
  {pages} {6813--6819} (\bibinfo {year} {2003}{\natexlab{b}})}\BibitemShut
  {NoStop}%
\bibitem [{\citenamefont {Jonnadula}\ \emph {et~al.}(2020)\citenamefont
  {Jonnadula}, \citenamefont {Mandayam}, \citenamefont {{\.Z}yczkowski},\ and\
  \citenamefont {Lakshminarayan}}]{Bhargavi2019}%
  \BibitemOpen
  \bibfield  {author} {\bibinfo {author} {\bibfnamefont {Bhargavi}\
  \bibnamefont {Jonnadula}}, \bibinfo {author} {\bibfnamefont {Prabha}\
  \bibnamefont {Mandayam}}, \bibinfo {author} {\bibfnamefont {Karol}\
  \bibnamefont {{\.Z}yczkowski}}, \ and\ \bibinfo {author} {\bibfnamefont
  {Arul}\ \bibnamefont {Lakshminarayan}},\ }\bibfield  {title} {\enquote
  {\bibinfo {title} {Entanglement measures of bipartite quantum gates and their
  thermalization under arbitrary interaction strength},}\ }\href {\doibase
  10.1103/physrevresearch.2.043126} {\bibfield  {journal} {\bibinfo  {journal}
  {Physical Review Research}\ }\textbf {\bibinfo {volume} {2}} (\bibinfo {year}
  {2020}),\ 10.1103/physrevresearch.2.043126}\BibitemShut {NoStop}%
\bibitem [{\citenamefont {Claeys}\ and\ \citenamefont
  {Lamacraft}(2021)}]{claeys2020ergodic}%
  \BibitemOpen
  \bibfield  {author} {\bibinfo {author} {\bibfnamefont {Pieter~W.}\
  \bibnamefont {Claeys}}\ and\ \bibinfo {author} {\bibfnamefont {Austen}\
  \bibnamefont {Lamacraft}},\ }\bibfield  {title} {\enquote {\bibinfo {title}
  {Ergodic and nonergodic dual-unitary quantum circuits with arbitrary local
  hilbert space dimension},}\ }\href {\doibase 10.1103/PhysRevLett.126.100603}
  {\bibfield  {journal} {\bibinfo  {journal} {Phys. Rev. Lett.}\ }\textbf
  {\bibinfo {volume} {126}},\ \bibinfo {pages} {100603} (\bibinfo {year}
  {2021})}\BibitemShut {NoStop}%
\bibitem [{\citenamefont {Singh}\ and\ \citenamefont
  {Nechita}(2022)}]{singh2021diagonal}%
  \BibitemOpen
  \bibfield  {author} {\bibinfo {author} {\bibfnamefont {Satvik}\ \bibnamefont
  {Singh}}\ and\ \bibinfo {author} {\bibfnamefont {Ion}\ \bibnamefont
  {Nechita}},\ }\bibfield  {title} {\enquote {\bibinfo {title} {Diagonal
  unitary and orthogonal symmetries in quantum theory: {II}. evolution
  operators},}\ }\href {\doibase 10.1088/1751-8121/ac7017} {\bibfield
  {journal} {\bibinfo  {journal} {Journal of Physics A: Mathematical and
  Theoretical}\ }\textbf {\bibinfo {volume} {55}},\ \bibinfo {pages} {255302}
  (\bibinfo {year} {2022})}\BibitemShut {NoStop}%
\bibitem [{\citenamefont {Rather}\ \emph {et~al.}(2022)\citenamefont {Rather},
  \citenamefont {Burchardt}, \citenamefont {Bruzda}, \citenamefont
  {Rajchel-Mieldzio{\'c}}, \citenamefont {Lakshminarayan},\ and\ \citenamefont
  {{\.Z}yczkowski}}]{SRatherAME46}%
  \BibitemOpen
  \bibfield  {author} {\bibinfo {author} {\bibfnamefont {Suhail~Ahmad}\
  \bibnamefont {Rather}}, \bibinfo {author} {\bibfnamefont {Adam}\ \bibnamefont
  {Burchardt}}, \bibinfo {author} {\bibfnamefont {Wojciech}\ \bibnamefont
  {Bruzda}}, \bibinfo {author} {\bibfnamefont {Grzegorz}\ \bibnamefont
  {Rajchel-Mieldzio{\'c}}}, \bibinfo {author} {\bibfnamefont {Arul}\
  \bibnamefont {Lakshminarayan}}, \ and\ \bibinfo {author} {\bibfnamefont
  {Karol}\ \bibnamefont {{\.Z}yczkowski}},\ }\bibfield  {title} {\enquote
  {\bibinfo {title} {Thirty-six entangled officers of {E}uler: Quantum solution
  to a classically impossible problem},}\ }\href {\doibase
  10.1103/PhysRevLett.128.080507} {\bibfield  {journal} {\bibinfo  {journal}
  {Phys. Rev. Lett.}\ }\textbf {\bibinfo {volume} {128}},\ \bibinfo {pages}
  {080507} (\bibinfo {year} {2022})}\BibitemShut {NoStop}%
\bibitem [{\citenamefont {{\.Z}yczkowski}\ \emph {et~al.}(2022)\citenamefont
  {{\.Z}yczkowski}, \citenamefont {Bruzda}, \citenamefont
  {Rajchel-Mieldzio\'c}, \citenamefont {Burchardt}, \citenamefont {Rather},\
  and\ \citenamefont {Lakshminarayan}}]{AME46_conf}%
  \BibitemOpen
  \bibfield  {author} {\bibinfo {author} {\bibfnamefont {Karol}\ \bibnamefont
  {{\.Z}yczkowski}}, \bibinfo {author} {\bibfnamefont {Wojciech}\ \bibnamefont
  {Bruzda}}, \bibinfo {author} {\bibfnamefont {Grzegorz}\ \bibnamefont
  {Rajchel-Mieldzio\'c}}, \bibinfo {author} {\bibfnamefont {Adam}\ \bibnamefont
  {Burchardt}}, \bibinfo {author} {\bibfnamefont {Suhail~Ahmad}\ \bibnamefont
  {Rather}}, \ and\ \bibinfo {author} {\bibfnamefont {Arul}\ \bibnamefont
  {Lakshminarayan}},\ }\href {\doibase 10.48550/ARXIV.2204.06800} {\enquote
  {\bibinfo {title} {9 $\times$ 4 = 6 $\times$ 6: Understanding the quantum
  solution to the {E}uler's problem of 36 officers},}\ } (\bibinfo {year}
  {2022})\BibitemShut {NoStop}%
\bibitem [{\citenamefont {Helwig}\ \emph {et~al.}(2012)\citenamefont {Helwig},
  \citenamefont {Cui}, \citenamefont {Latorre}, \citenamefont {Riera},\ and\
  \citenamefont {Lo}}]{Helwig_2012}%
  \BibitemOpen
  \bibfield  {author} {\bibinfo {author} {\bibfnamefont {Wolfram}\ \bibnamefont
  {Helwig}}, \bibinfo {author} {\bibfnamefont {Wei}\ \bibnamefont {Cui}},
  \bibinfo {author} {\bibfnamefont {Jos\'e~Ignacio}\ \bibnamefont {Latorre}},
  \bibinfo {author} {\bibfnamefont {Arnau}\ \bibnamefont {Riera}}, \ and\
  \bibinfo {author} {\bibfnamefont {Hoi-Kwong}\ \bibnamefont {Lo}},\ }\bibfield
   {title} {\enquote {\bibinfo {title} {Absolute maximal entanglement and
  quantum secret sharing},}\ }\href {\doibase 10.1103/PhysRevA.86.052335}
  {\bibfield  {journal} {\bibinfo  {journal} {Phys. Rev. A}\ }\textbf {\bibinfo
  {volume} {86}},\ \bibinfo {pages} {052335} (\bibinfo {year}
  {2012})}\BibitemShut {NoStop}%
\bibitem [{\citenamefont {Goyeneche}\ \emph {et~al.}(2018)\citenamefont
  {Goyeneche}, \citenamefont {Raissi}, \citenamefont {Di~Martino},\ and\
  \citenamefont {\ifmmode~\dot{Z}\else \.{Z}\fi{}yczkowski}}]{GRMZ_2018}%
  \BibitemOpen
  \bibfield  {author} {\bibinfo {author} {\bibfnamefont {Dardo}\ \bibnamefont
  {Goyeneche}}, \bibinfo {author} {\bibfnamefont {Zahra}\ \bibnamefont
  {Raissi}}, \bibinfo {author} {\bibfnamefont {Sara}\ \bibnamefont
  {Di~Martino}}, \ and\ \bibinfo {author} {\bibfnamefont {Karol}\ \bibnamefont
  {\ifmmode~\dot{Z}\else \.{Z}\fi{}yczkowski}},\ }\bibfield  {title} {\enquote
  {\bibinfo {title} {Entanglement and quantum combinatorial designs},}\ }\href
  {\doibase 10.1103/PhysRevA.97.062326} {\bibfield  {journal} {\bibinfo
  {journal} {Phys. Rev. A}\ }\textbf {\bibinfo {volume} {97}},\ \bibinfo
  {pages} {062326} (\bibinfo {year} {2018})}\BibitemShut {NoStop}%
\bibitem [{\citenamefont {Clarisse}\ \emph {et~al.}(2005)\citenamefont
  {Clarisse}, \citenamefont {Ghosh}, \citenamefont {Severini},\ and\
  \citenamefont {Sudbery}}]{Clarisse2005}%
  \BibitemOpen
  \bibfield  {author} {\bibinfo {author} {\bibfnamefont {Lieven}\ \bibnamefont
  {Clarisse}}, \bibinfo {author} {\bibfnamefont {Sibasish}\ \bibnamefont
  {Ghosh}}, \bibinfo {author} {\bibfnamefont {Simone}\ \bibnamefont
  {Severini}}, \ and\ \bibinfo {author} {\bibfnamefont {Anthony}\ \bibnamefont
  {Sudbery}},\ }\bibfield  {title} {\enquote {\bibinfo {title} {Entangling
  power of permutations},}\ }\href {\doibase 10.1103/PhysRevA.72.012314}
  {\bibfield  {journal} {\bibinfo  {journal} {Phys. Rev. A}\ }\textbf {\bibinfo
  {volume} {72}},\ \bibinfo {pages} {012314} (\bibinfo {year}
  {2005})}\BibitemShut {NoStop}%
\bibitem [{\citenamefont {Burchardt}\ and\ \citenamefont
  {Raissi}(2020)}]{Adam_SLOCC_2020}%
  \BibitemOpen
  \bibfield  {author} {\bibinfo {author} {\bibfnamefont {Adam}\ \bibnamefont
  {Burchardt}}\ and\ \bibinfo {author} {\bibfnamefont {Zahra}\ \bibnamefont
  {Raissi}},\ }\bibfield  {title} {\enquote {\bibinfo {title} {Stochastic local
  operations with classical communication of absolutely maximally entangled
  states},}\ }\href {\doibase 10.1103/PhysRevA.102.022413} {\bibfield
  {journal} {\bibinfo  {journal} {Phys. Rev. A}\ }\textbf {\bibinfo {volume}
  {102}},\ \bibinfo {pages} {022413} (\bibinfo {year} {2020})}\BibitemShut
  {NoStop}%
\bibitem [{\citenamefont {Jonnadula}\ \emph {et~al.}(2017)\citenamefont
  {Jonnadula}, \citenamefont {Mandayam}, \citenamefont {\ifmmode~\dot{Z}\else
  \.{Z}\fi{}yczkowski},\ and\ \citenamefont {Lakshminarayan}}]{Bhargavi2017}%
  \BibitemOpen
  \bibfield  {author} {\bibinfo {author} {\bibfnamefont {Bhargavi}\
  \bibnamefont {Jonnadula}}, \bibinfo {author} {\bibfnamefont {Prabha}\
  \bibnamefont {Mandayam}}, \bibinfo {author} {\bibfnamefont {Karol}\
  \bibnamefont {\ifmmode~\dot{Z}\else \.{Z}\fi{}yczkowski}}, \ and\ \bibinfo
  {author} {\bibfnamefont {Arul}\ \bibnamefont {Lakshminarayan}},\ }\bibfield
  {title} {\enquote {\bibinfo {title} {Impact of local dynamics on entangling
  power},}\ }\href {\doibase 10.1103/PhysRevA.95.040302} {\bibfield  {journal}
  {\bibinfo  {journal} {Phys. Rev. A}\ }\textbf {\bibinfo {volume} {95}},\
  \bibinfo {pages} {040302} (\bibinfo {year} {2017})}\BibitemShut {NoStop}%
\bibitem [{\citenamefont {Ocneanu}(1988)}]{ocneanu1988quantized}%
  \BibitemOpen
  \bibfield  {author} {\bibinfo {author} {\bibfnamefont {Adrian}\ \bibnamefont
  {Ocneanu}},\ }\bibfield  {title} {\enquote {\bibinfo {title} {Quantized
  groups, string algebras and galois theory for algebras},}\ }\href@noop {}
  {\bibfield  {journal} {\bibinfo  {journal} {Operator algebras and
  applications}\ }\textbf {\bibinfo {volume} {2}},\ \bibinfo {pages} {119--172}
  (\bibinfo {year} {1988})}\BibitemShut {NoStop}%
\bibitem [{\citenamefont {Krishnan}\ and\ \citenamefont
  {Sunder}(1996)}]{krishnan1996biunitary}%
  \BibitemOpen
  \bibfield  {author} {\bibinfo {author} {\bibfnamefont {Uma}\ \bibnamefont
  {Krishnan}}\ and\ \bibinfo {author} {\bibfnamefont {V.~S.}\ \bibnamefont
  {Sunder}},\ }\bibfield  {title} {\enquote {\bibinfo {title} {On biunitary
  permutation matrices and some subfactors of index 9},}\ }\href
  {http://www.jstor.org/stable/2155367} {\bibfield  {journal} {\bibinfo
  {journal} {Transactions of the American Mathematical Society}\ }\textbf
  {\bibinfo {volume} {348}},\ \bibinfo {pages} {4691--4736} (\bibinfo {year}
  {1996})}\BibitemShut {NoStop}%
\bibitem [{\citenamefont {Jones}()}]{jones1999planar}%
  \BibitemOpen
  \bibfield  {author} {\bibinfo {author} {\bibfnamefont {Vaughan F.~R.}\
  \bibnamefont {Jones}},\ }\href@noop {} {\enquote {\bibinfo {title} {Planar
  algebras, {I}},}\ }\Eprint
  {http://arxiv.org/abs/https://arxiv.org/abs/math/9909027}
  {https://arxiv.org/abs/math/9909027} \BibitemShut {NoStop}%
\bibitem [{\citenamefont {Reutter}\ and\ \citenamefont
  {Vicary}(2019)}]{reutter2016biunitary}%
  \BibitemOpen
  \bibfield  {author} {\bibinfo {author} {\bibfnamefont {David~J.}\
  \bibnamefont {Reutter}}\ and\ \bibinfo {author} {\bibfnamefont {Jamie}\
  \bibnamefont {Vicary}},\ }\bibfield  {title} {\enquote {\bibinfo {title}
  {Biunitary constructions in quantum information},}\ }\href@noop {} {\bibfield
   {journal} {\bibinfo  {journal} {Higher Structures}\ }\textbf {\bibinfo
  {volume} {3}},\ \bibinfo {pages} {109--154} (\bibinfo {year} {2019})},\
  \Eprint {http://arxiv.org/abs/https://arxiv.org/abs/1609.07775}
  {https://arxiv.org/abs/1609.07775} \BibitemShut {NoStop}%
\bibitem [{\citenamefont {Benoist}\ and\ \citenamefont
  {Nechita}(2017{\natexlab{a}})}]{Nechita2017}%
  \BibitemOpen
  \bibfield  {author} {\bibinfo {author} {\bibfnamefont {Tristan}\ \bibnamefont
  {Benoist}}\ and\ \bibinfo {author} {\bibfnamefont {Ion}\ \bibnamefont
  {Nechita}},\ }\bibfield  {title} {\enquote {\bibinfo {title} {On bipartite
  unitary matrices generating subalgebra-preserving quantum operations},}\
  }\href {\doibase 10.1016/j.laa.2017.01.020} {\bibfield  {journal} {\bibinfo
  {journal} {Linear Algebra and its Applications}\ }\textbf {\bibinfo {volume}
  {521}},\ \bibinfo {pages} {70–103} (\bibinfo {year}
  {2017}{\natexlab{a}})}\BibitemShut {NoStop}%
\bibitem [{\citenamefont {Kodiyalam}\ \emph {et~al.}(2020)\citenamefont
  {Kodiyalam}, \citenamefont {Sruthymurali},\ and\ \citenamefont
  {Sunder}}]{kodiyalam2020planar}%
  \BibitemOpen
  \bibfield  {author} {\bibinfo {author} {\bibfnamefont {Vijay}\ \bibnamefont
  {Kodiyalam}}, \bibinfo {author} {\bibnamefont {Sruthymurali}}, \ and\
  \bibinfo {author} {\bibfnamefont {V.~S.}\ \bibnamefont {Sunder}},\ }\bibfield
   {title} {\enquote {\bibinfo {title} {Planar algebras, quantum information
  theory and subfactors},}\ }\href {\doibase 10.1142/S0129167X20501244}
  {\bibfield  {journal} {\bibinfo  {journal} {International Journal of
  Mathematics}\ }\textbf {\bibinfo {volume} {31}},\ \bibinfo {pages} {2050124}
  (\bibinfo {year} {2020})}\BibitemShut {NoStop}%
\bibitem [{\citenamefont {Nechita}\ \emph {et~al.}(2021)\citenamefont
  {Nechita}, \citenamefont {Schmidt},\ and\ \citenamefont
  {Weber}}]{nechita2021sinkhorn}%
  \BibitemOpen
  \bibfield  {author} {\bibinfo {author} {\bibfnamefont {Ion}\ \bibnamefont
  {Nechita}}, \bibinfo {author} {\bibfnamefont {Simon}\ \bibnamefont
  {Schmidt}}, \ and\ \bibinfo {author} {\bibfnamefont {Moritz}\ \bibnamefont
  {Weber}},\ }\bibfield  {title} {\enquote {\bibinfo {title} {Sinkhorn
  algorithm for quantum permutation groups},}\ }\href {\doibase
  10.1080/10586458.2021.1926005} {\bibfield  {journal} {\bibinfo  {journal}
  {Experimental Mathematics}\ }\textbf {\bibinfo {volume} {0}},\ \bibinfo
  {pages} {1--13} (\bibinfo {year} {2021})}\BibitemShut {NoStop}%
\bibitem [{\citenamefont {Gutkin}\ and\ \citenamefont
  {Osipov}(2016)}]{gutkin2016classical}%
  \BibitemOpen
  \bibfield  {author} {\bibinfo {author} {\bibfnamefont {Boris}\ \bibnamefont
  {Gutkin}}\ and\ \bibinfo {author} {\bibfnamefont {Vladimir}\ \bibnamefont
  {Osipov}},\ }\bibfield  {title} {\enquote {\bibinfo {title} {Classical
  foundations of many-particle quantum chaos},}\ }\href {\doibase
  10.1088/0951-7715/29/2/325} {\bibfield  {journal} {\bibinfo  {journal}
  {Nonlinearity}\ }\textbf {\bibinfo {volume} {29}},\ \bibinfo {pages} {325}
  (\bibinfo {year} {2016})}\BibitemShut {NoStop}%
\bibitem [{\citenamefont {Fan}\ and\ \citenamefont {Hoffman}(1955)}]{Fan1955}%
  \BibitemOpen
  \bibfield  {author} {\bibinfo {author} {\bibfnamefont {Ky}~\bibnamefont
  {Fan}}\ and\ \bibinfo {author} {\bibfnamefont {A.~J.}\ \bibnamefont
  {Hoffman}},\ }\bibfield  {title} {\enquote {\bibinfo {title} {Some metric
  inequalities in the space of matrices},}\ }\href
  {http://www.jstor.org/stable/2032662} {\bibfield  {journal} {\bibinfo
  {journal} {Proceedings of the American Mathematical Society}\ }\textbf
  {\bibinfo {volume} {6}},\ \bibinfo {pages} {111--116} (\bibinfo {year}
  {1955})}\BibitemShut {NoStop}%
\bibitem [{\citenamefont {Keller}(1975)}]{Keller1975}%
  \BibitemOpen
  \bibfield  {author} {\bibinfo {author} {\bibfnamefont {Joseph~B.}\
  \bibnamefont {Keller}},\ }\bibfield  {title} {\enquote {\bibinfo {title}
  {Closest unitary, orthogonal and hermitian operators to a given operator},}\
  }\href {http://www.jstor.org/stable/2690338} {\bibfield  {journal} {\bibinfo
  {journal} {Mathematics Magazine}\ }\textbf {\bibinfo {volume} {48}},\
  \bibinfo {pages} {192--197} (\bibinfo {year} {1975})}\BibitemShut {NoStop}%
\bibitem [{\citenamefont {Benoist}\ and\ \citenamefont
  {Nechita}(2017{\natexlab{b}})}]{Nechita_2017}%
  \BibitemOpen
  \bibfield  {author} {\bibinfo {author} {\bibfnamefont {Tristan}\ \bibnamefont
  {Benoist}}\ and\ \bibinfo {author} {\bibfnamefont {Ion}\ \bibnamefont
  {Nechita}},\ }\bibfield  {title} {\enquote {\bibinfo {title} {On bipartite
  unitary matrices generating subalgebra-preserving quantum operations},}\
  }\href {\doibase https://doi.org/10.1016/j.laa.2017.01.020} {\bibfield
  {journal} {\bibinfo  {journal} {Linear Algebra and its Applications}\
  }\textbf {\bibinfo {volume} {521}},\ \bibinfo {pages} {70--103} (\bibinfo
  {year} {2017}{\natexlab{b}})}\BibitemShut {NoStop}%
\bibitem [{\citenamefont {Rajchel-Mieldzio\'c}(2022)}]{Grzegorz_thesis}%
  \BibitemOpen
  \bibfield  {author} {\bibinfo {author} {\bibfnamefont {Grzegorz}\
  \bibnamefont {Rajchel-Mieldzio\'c}},\ }\emph {\bibinfo {title} {Quantum
  mappings and designs}},\ \href {\doibase 10.48550/ARXIV.2204.13008} {Ph.D.
  thesis} (\bibinfo {year} {2022})\BibitemShut {NoStop}%
\bibitem [{\citenamefont {Prosen}(2021)}]{prosen2021many}%
  \BibitemOpen
  \bibfield  {author} {\bibinfo {author} {\bibfnamefont {Toma{\v z}}\
  \bibnamefont {Prosen}},\ }\bibfield  {title} {\enquote {\bibinfo {title}
  {Many-body quantum chaos and dual-unitarity round-a-face},}\ }\href {\doibase
  10.1063/5.0056970} {\bibfield  {journal} {\bibinfo  {journal} {Chaos: An
  Interdisciplinary Journal of Nonlinear Science}\ }\textbf {\bibinfo {volume}
  {31}},\ \bibinfo {pages} {093101} (\bibinfo {year} {2021})}\BibitemShut
  {NoStop}%
\bibitem [{\citenamefont {Khaneja}\ \emph {et~al.}(2001)\citenamefont
  {Khaneja}, \citenamefont {Brockett},\ and\ \citenamefont {Glaser}}]{KBG01}%
  \BibitemOpen
  \bibfield  {author} {\bibinfo {author} {\bibfnamefont {Navin}\ \bibnamefont
  {Khaneja}}, \bibinfo {author} {\bibfnamefont {Roger}\ \bibnamefont
  {Brockett}}, \ and\ \bibinfo {author} {\bibfnamefont {Steffen~J.}\
  \bibnamefont {Glaser}},\ }\bibfield  {title} {\enquote {\bibinfo {title}
  {Time optimal control in spin systems},}\ }\href {\doibase
  10.1103/PhysRevA.63.032308} {\bibfield  {journal} {\bibinfo  {journal} {Phys.
  Rev. A}\ }\textbf {\bibinfo {volume} {63}},\ \bibinfo {pages} {032308}
  (\bibinfo {year} {2001})}\BibitemShut {NoStop}%
\bibitem [{\citenamefont {Kraus}\ and\ \citenamefont {Cirac}(2001)}]{KC01}%
  \BibitemOpen
  \bibfield  {author} {\bibinfo {author} {\bibfnamefont {B.}~\bibnamefont
  {Kraus}}\ and\ \bibinfo {author} {\bibfnamefont {J.~I.}\ \bibnamefont
  {Cirac}},\ }\bibfield  {title} {\enquote {\bibinfo {title} {Optimal creation
  of entanglement using a two-qubit gate},}\ }\href {\doibase
  10.1103/PhysRevA.63.062309} {\bibfield  {journal} {\bibinfo  {journal} {Phys.
  Rev. A}\ }\textbf {\bibinfo {volume} {63}},\ \bibinfo {pages} {062309}
  (\bibinfo {year} {2001})}\BibitemShut {NoStop}%
\bibitem [{\citenamefont {Zhang}\ \emph {et~al.}(2003)\citenamefont {Zhang},
  \citenamefont {Vala}, \citenamefont {Sastry},\ and\ \citenamefont
  {Whaley}}]{Zhang2003}%
  \BibitemOpen
  \bibfield  {author} {\bibinfo {author} {\bibfnamefont {Jun}\ \bibnamefont
  {Zhang}}, \bibinfo {author} {\bibfnamefont {Jiri}\ \bibnamefont {Vala}},
  \bibinfo {author} {\bibfnamefont {Shankar}\ \bibnamefont {Sastry}}, \ and\
  \bibinfo {author} {\bibfnamefont {K.~Birgitta}\ \bibnamefont {Whaley}},\
  }\bibfield  {title} {\enquote {\bibinfo {title} {Geometric theory of nonlocal
  two-qubit operations},}\ }\href {\doibase 10.1103/PhysRevA.67.042313}
  {\bibfield  {journal} {\bibinfo  {journal} {Phys. Rev. A}\ }\textbf {\bibinfo
  {volume} {67}},\ \bibinfo {pages} {042313} (\bibinfo {year}
  {2003})}\BibitemShut {NoStop}%
\bibitem [{\citenamefont {Musz}\ \emph {et~al.}(2013)\citenamefont {Musz},
  \citenamefont {Ku\ifmmode~\acute{s}\else \'{s}\fi{}},\ and\ \citenamefont
  {\ifmmode~\dot{Z}\else \.{Z}\fi{}yczkowski}}]{Kus2013}%
  \BibitemOpen
  \bibfield  {author} {\bibinfo {author} {\bibfnamefont {Marcin}\ \bibnamefont
  {Musz}}, \bibinfo {author} {\bibfnamefont {Marek}\ \bibnamefont
  {Ku\ifmmode~\acute{s}\else \'{s}\fi{}}}, \ and\ \bibinfo {author}
  {\bibfnamefont {Karol}\ \bibnamefont {\ifmmode~\dot{Z}\else
  \.{Z}\fi{}yczkowski}},\ }\bibfield  {title} {\enquote {\bibinfo {title}
  {Unitary quantum gates, perfect entanglers, and unistochastic maps},}\ }\href
  {\doibase 10.1103/PhysRevA.87.022111} {\bibfield  {journal} {\bibinfo
  {journal} {Phys. Rev. A}\ }\textbf {\bibinfo {volume} {87}},\ \bibinfo
  {pages} {022111} (\bibinfo {year} {2013})}\BibitemShut {NoStop}%
\bibitem [{\citenamefont {Childs}\ \emph {et~al.}(2003)\citenamefont {Childs},
  \citenamefont {Haselgrove},\ and\ \citenamefont {Nielsen}}]{Childs_2003}%
  \BibitemOpen
  \bibfield  {author} {\bibinfo {author} {\bibfnamefont {Andrew~M.}\
  \bibnamefont {Childs}}, \bibinfo {author} {\bibfnamefont {Henry~L.}\
  \bibnamefont {Haselgrove}}, \ and\ \bibinfo {author} {\bibfnamefont
  {Michael~A.}\ \bibnamefont {Nielsen}},\ }\bibfield  {title} {\enquote
  {\bibinfo {title} {Lower bounds on the complexity of simulating quantum
  gates},}\ }\href {\doibase 10.1103/PhysRevA.68.052311} {\bibfield  {journal}
  {\bibinfo  {journal} {Phys. Rev. A}\ }\textbf {\bibinfo {volume} {68}},\
  \bibinfo {pages} {052311} (\bibinfo {year} {2003})}\BibitemShut {NoStop}%
\bibitem [{\citenamefont {Vanicat}\ \emph {et~al.}(2018)\citenamefont
  {Vanicat}, \citenamefont {Zadnik},\ and\ \citenamefont
  {Prosen}}]{Vanicat_2018}%
  \BibitemOpen
  \bibfield  {author} {\bibinfo {author} {\bibfnamefont {Matthieu}\
  \bibnamefont {Vanicat}}, \bibinfo {author} {\bibfnamefont {Lenart}\
  \bibnamefont {Zadnik}}, \ and\ \bibinfo {author} {\bibfnamefont {Toma{\v z}}\
  \bibnamefont {Prosen}},\ }\bibfield  {title} {\enquote {\bibinfo {title}
  {Integrable {T}rotterization: Local conservation laws and boundary
  driving},}\ }\href {\doibase 10.1103/PhysRevLett.121.030606} {\bibfield
  {journal} {\bibinfo  {journal} {Phys. Rev. Lett.}\ }\textbf {\bibinfo
  {volume} {121}},\ \bibinfo {pages} {030606} (\bibinfo {year}
  {2018})}\BibitemShut {NoStop}%
\bibitem [{\citenamefont {Mandarino}\ \emph {et~al.}(2018)\citenamefont
  {Mandarino}, \citenamefont {Linowski},\ and\ \citenamefont
  {{\.Z}yczkowski}}]{Mandarino2018}%
  \BibitemOpen
  \bibfield  {author} {\bibinfo {author} {\bibfnamefont {Antonio}\ \bibnamefont
  {Mandarino}}, \bibinfo {author} {\bibfnamefont {Tomasz}\ \bibnamefont
  {Linowski}}, \ and\ \bibinfo {author} {\bibfnamefont {Karol}\ \bibnamefont
  {{\.Z}yczkowski}},\ }\bibfield  {title} {\enquote {\bibinfo {title}
  {Bipartite unitary gates and billiard dynamics in the {W}eyl chamber},}\
  }\href {\doibase 10.1103/PhysRevA.98.012335} {\bibfield  {journal} {\bibinfo
  {journal} {Phys. Rev. A}\ }\textbf {\bibinfo {volume} {98}},\ \bibinfo
  {pages} {012335} (\bibinfo {year} {2018})}\BibitemShut {NoStop}%
\bibitem [{\citenamefont {S}\ \emph {et~al.}(2020)\citenamefont {S},
  \citenamefont {Ahmad~Rather},\ and\ \citenamefont
  {Lakshminarayan}}]{ASL2020}%
  \BibitemOpen
  \bibfield  {author} {\bibinfo {author} {\bibfnamefont {Aravinda}\
  \bibnamefont {S}}, \bibinfo {author} {\bibfnamefont {Suhail}\ \bibnamefont
  {Ahmad~Rather}}, \ and\ \bibinfo {author} {\bibfnamefont {Arul}\ \bibnamefont
  {Lakshminarayan}},\ }\href@noop {} {} (\bibinfo {year} {2020}),\ \bibinfo
  {note} {in preparation}\BibitemShut {NoStop}%
\bibitem [{\citenamefont {Keedwell}\ and\ \citenamefont
  {D{\'e}nes}(2015)}]{keedwell2015latin}%
  \BibitemOpen
  \bibfield  {author} {\bibinfo {author} {\bibfnamefont {A~Donald}\
  \bibnamefont {Keedwell}}\ and\ \bibinfo {author} {\bibfnamefont {J{\'o}zsef}\
  \bibnamefont {D{\'e}nes}},\ }\href@noop {} {\emph {\bibinfo {title} {Latin
  squares and their applications}}},\ \bibinfo {edition} {2nd}\ ed.\ (\bibinfo
  {publisher} {Elsevier},\ \bibinfo {year} {2015})\BibitemShut {NoStop}%
\bibitem [{\citenamefont {Bose}\ \emph {et~al.}(1960)\citenamefont {Bose},
  \citenamefont {Shrikhande},\ and\ \citenamefont {Parker}}]{bose1960further}%
  \BibitemOpen
  \bibfield  {author} {\bibinfo {author} {\bibfnamefont {R.~C.}\ \bibnamefont
  {Bose}}, \bibinfo {author} {\bibfnamefont {S.~S.}\ \bibnamefont
  {Shrikhande}}, \ and\ \bibinfo {author} {\bibfnamefont {E.~T.}\ \bibnamefont
  {Parker}},\ }\bibfield  {title} {\enquote {\bibinfo {title} {Further results
  on the construction of mutually orthogonal latin squares and the falsity of
  euler's conjecture},}\ }\href {\doibase 10.4153/CJM-1960-016-5} {\bibfield
  {journal} {\bibinfo  {journal} {Canadian Journal of Mathematics}\ }\textbf
  {\bibinfo {volume} {12}},\ \bibinfo {pages} {189–203} (\bibinfo {year}
  {1960})}\BibitemShut {NoStop}%
\bibitem [{\citenamefont {Musto}\ and\ \citenamefont {Vicary}(2016)}]{MV16}%
  \BibitemOpen
  \bibfield  {author} {\bibinfo {author} {\bibfnamefont {Benjamin}\
  \bibnamefont {Musto}}\ and\ \bibinfo {author} {\bibfnamefont {Jamie}\
  \bibnamefont {Vicary}},\ }\bibfield  {title} {\enquote {\bibinfo {title}
  {Quantum {L}atin squares and unitary error bases},}\ }\href@noop {}
  {\bibfield  {journal} {\bibinfo  {journal} {Quantum Info. Comput.}\ }\textbf
  {\bibinfo {volume} {16}},\ \bibinfo {pages} {1318–1332} (\bibinfo {year}
  {2016})},\ \Eprint
  {http://arxiv.org/abs/https://doi.org/10.48550/arXiv.1504.02715}
  {https://doi.org/10.48550/arXiv.1504.02715} \BibitemShut {NoStop}%
\bibitem [{\citenamefont {Musto}\ and\ \citenamefont {Vicary}(2019)}]{MV19}%
  \BibitemOpen
  \bibfield  {author} {\bibinfo {author} {\bibfnamefont {Benjamin}\
  \bibnamefont {Musto}}\ and\ \bibinfo {author} {\bibfnamefont {Jamie}\
  \bibnamefont {Vicary}},\ }\bibfield  {title} {\enquote {\bibinfo {title}
  {Orthogonality for quantum {L}atin isometry squares},}\ }\href {\doibase
  10.4204/eptcs.287.15} {\bibfield  {journal} {\bibinfo  {journal} {Electronic
  Proceedings in Theoretical Computer Science}\ }\textbf {\bibinfo {volume}
  {287}},\ \bibinfo {pages} {253–266} (\bibinfo {year} {2019})}\BibitemShut
  {NoStop}%
\bibitem [{\citenamefont {Paczos}\ \emph {et~al.}(2021)\citenamefont {Paczos},
  \citenamefont {Wierzbi\ifmmode~\acute{n}\else \'{n}\fi{}ski}, \citenamefont
  {Rajchel-Mieldzio\ifmmode~\acute{c}\else \'{c}\fi{}}, \citenamefont
  {Burchardt},\ and\ \citenamefont {\ifmmode~\dot{Z}\else
  \.{Z}\fi{}yczkowski}}]{Paczos_2021}%
  \BibitemOpen
  \bibfield  {author} {\bibinfo {author} {\bibfnamefont {Jerzy}\ \bibnamefont
  {Paczos}}, \bibinfo {author} {\bibfnamefont {Marcin}\ \bibnamefont
  {Wierzbi\ifmmode~\acute{n}\else \'{n}\fi{}ski}}, \bibinfo {author}
  {\bibfnamefont {Grzegorz}\ \bibnamefont
  {Rajchel-Mieldzio\ifmmode~\acute{c}\else \'{c}\fi{}}}, \bibinfo {author}
  {\bibfnamefont {Adam}\ \bibnamefont {Burchardt}}, \ and\ \bibinfo {author}
  {\bibfnamefont {Karol}\ \bibnamefont {\ifmmode~\dot{Z}\else
  \.{Z}\fi{}yczkowski}},\ }\bibfield  {title} {\enquote {\bibinfo {title}
  {Genuinely quantum solutions of the game sudoku and their cardinality},}\
  }\href {\doibase 10.1103/PhysRevA.104.042423} {\bibfield  {journal} {\bibinfo
   {journal} {Phys. Rev. A}\ }\textbf {\bibinfo {volume} {104}},\ \bibinfo
  {pages} {042423} (\bibinfo {year} {2021})}\BibitemShut {NoStop}%
\bibitem [{\citenamefont {Nechita}\ and\ \citenamefont
  {Pillet}()}]{Nechita_qsudoku}%
  \BibitemOpen
  \bibfield  {author} {\bibinfo {author} {\bibfnamefont {Ion}\ \bibnamefont
  {Nechita}}\ and\ \bibinfo {author} {\bibfnamefont {Jordi}\ \bibnamefont
  {Pillet}},\ }\bibfield  {title} {\enquote {\bibinfo {title} {Sudoq--a quantum
  variant of the popular game},}\ }\href@noop {} {\bibfield  {journal}
  {\bibinfo  {journal} {arXiv:2005.10862}\ }}\Eprint
  {http://arxiv.org/abs/https://arxiv.org/abs/2005.10862}
  {https://arxiv.org/abs/2005.10862} \BibitemShut {NoStop}%
\bibitem [{\citenamefont {Borsi}\ and\ \citenamefont
  {Pozsgay}(2022)}]{borsi2022remarks}%
  \BibitemOpen
  \bibfield  {author} {\bibinfo {author} {\bibfnamefont {M\'arton}\
  \bibnamefont {Borsi}}\ and\ \bibinfo {author} {\bibfnamefont {Bal\'azs}\
  \bibnamefont {Pozsgay}},\ }\bibfield  {title} {\enquote {\bibinfo {title}
  {Construction and the ergodicity properties of dual unitary quantum
  circuits},}\ }\href {\doibase 10.1103/PhysRevB.106.014302} {\bibfield
  {journal} {\bibinfo  {journal} {Phys. Rev. B}\ }\textbf {\bibinfo {volume}
  {106}},\ \bibinfo {pages} {014302} (\bibinfo {year} {2022})}\BibitemShut
  {NoStop}%
\bibitem [{\citenamefont {Chen}\ \emph {et~al.}(2008)\citenamefont {Chen},
  \citenamefont {Duan}, \citenamefont {Ji}, \citenamefont {Ying},\ and\
  \citenamefont {Yu}}]{ChenDuan2007}%
  \BibitemOpen
  \bibfield  {author} {\bibinfo {author} {\bibfnamefont {Jianxin}\ \bibnamefont
  {Chen}}, \bibinfo {author} {\bibfnamefont {Runyao}\ \bibnamefont {Duan}},
  \bibinfo {author} {\bibfnamefont {Zhengfeng}\ \bibnamefont {Ji}}, \bibinfo
  {author} {\bibfnamefont {Mingsheng}\ \bibnamefont {Ying}}, \ and\ \bibinfo
  {author} {\bibfnamefont {Jun}\ \bibnamefont {Yu}},\ }\bibfield  {title}
  {\enquote {\bibinfo {title} {Existence of universal entangler},}\ }\href
  {\doibase 10.1063/1.2829895} {\bibfield  {journal} {\bibinfo  {journal}
  {Journal of Mathematical Physics}\ }\textbf {\bibinfo {volume} {49}},\
  \bibinfo {pages} {012103} (\bibinfo {year} {2008})},\ \Eprint
  {http://arxiv.org/abs/https://doi.org/10.1063/1.2829895}
  {https://doi.org/10.1063/1.2829895} \BibitemShut {NoStop}%
\bibitem [{\citenamefont {Mendes}\ and\ \citenamefont
  {Ramos}(2015)}]{MENDES2015}%
  \BibitemOpen
  \bibfield  {author} {\bibinfo {author} {\bibfnamefont {F.V.}\ \bibnamefont
  {Mendes}}\ and\ \bibinfo {author} {\bibfnamefont {R.V.}\ \bibnamefont
  {Ramos}},\ }\bibfield  {title} {\enquote {\bibinfo {title} {Numerical search
  for universal entanglers in ${C}^3 \otimes{C}^4$ and ${C}^4 \otimes
  {C}^4$},}\ }\href {\doibase https://doi.org/10.1016/j.physleta.2014.11.056}
  {\bibfield  {journal} {\bibinfo  {journal} {Physics Letters A}\ }\textbf
  {\bibinfo {volume} {379}},\ \bibinfo {pages} {289--292} (\bibinfo {year}
  {2015})}\BibitemShut {NoStop}%
\bibitem [{\citenamefont {Sloane}\ and\ \citenamefont {Inc.}(2020)}]{OEIS}%
  \BibitemOpen
  \bibfield  {author} {\bibinfo {author} {\bibfnamefont {Neil J.~A.}\
  \bibnamefont {Sloane}}\ and\ \bibinfo {author} {\bibfnamefont {The
  OEIS~Foundation}\ \bibnamefont {Inc.}},\ }\href
  {http://oeis.org/?language=english} {\enquote {\bibinfo {title} {The on-line
  encyclopedia of integer sequences},}\ } (\bibinfo {year} {2020})\BibitemShut
  {NoStop}%
\bibitem [{\citenamefont {Raissi}\ \emph {et~al.}(2020)\citenamefont {Raissi},
  \citenamefont {Teixid\'o}, \citenamefont {Gogolin},\ and\ \citenamefont
  {Ac\'{\i}n}}]{Acin_AME_2020}%
  \BibitemOpen
  \bibfield  {author} {\bibinfo {author} {\bibfnamefont {Zahra}\ \bibnamefont
  {Raissi}}, \bibinfo {author} {\bibfnamefont {Adam}\ \bibnamefont
  {Teixid\'o}}, \bibinfo {author} {\bibfnamefont {Christian}\ \bibnamefont
  {Gogolin}}, \ and\ \bibinfo {author} {\bibfnamefont {Antonio}\ \bibnamefont
  {Ac\'{\i}n}},\ }\bibfield  {title} {\enquote {\bibinfo {title} {Constructions
  of $k$-uniform and absolutely maximally entangled states beyond maximum
  distance codes},}\ }\href {\doibase 10.1103/PhysRevResearch.2.033411}
  {\bibfield  {journal} {\bibinfo  {journal} {Phys. Rev. Research}\ }\textbf
  {\bibinfo {volume} {2}},\ \bibinfo {pages} {033411} (\bibinfo {year}
  {2020})}\BibitemShut {NoStop}%
\bibitem [{\citenamefont {Kodiyalam}\ and\ \citenamefont
  {Sunder}(2004)}]{VijayK}%
  \BibitemOpen
  \bibfield  {author} {\bibinfo {author} {\bibfnamefont {Vijay}\ \bibnamefont
  {Kodiyalam}}\ and\ \bibinfo {author} {\bibfnamefont {V.~S.}\ \bibnamefont
  {Sunder}},\ }\bibfield  {title} {\enquote {\bibinfo {title} {A complete set
  of numerical invariants for a subfactor},}\ }\href {\doibase
  https://doi.org/10.1016/j.jfa.2003.11.010} {\bibfield  {journal} {\bibinfo
  {journal} {Journal of Functional Analysis}\ }\textbf {\bibinfo {volume}
  {212}},\ \bibinfo {pages} {1--27} (\bibinfo {year} {2004})}\BibitemShut
  {NoStop}%
\bibitem [{\citenamefont {Rico}(2020)}]{rico2020absolutely}%
  \BibitemOpen
  \bibfield  {author} {\bibinfo {author} {\bibfnamefont {Albert}\ \bibnamefont
  {Rico}},\ }\emph {\bibinfo {title} {Absolutely maximally entangled states in
  small system sizes}},\ \href
  {https://diglib.uibk.ac.at/ulbtirolhs/content/titleinfo/5327562} {Master's
  thesis},\ \bibinfo  {school} {University of Innsbruck} (\bibinfo {year}
  {2020})\BibitemShut {NoStop}%
\end{thebibliography}%

\pagebreak

\onecolumngrid

\appendix 
\section{Details about the map in the two-qubit case \label{app:map}}

In the two-qubit case, the Cartan form of any unitary (\ref{eq:cartan}) can be written as

\begin{equation}
    U_0 =     \left( \begin{array}{cccc}
    e^{-ic_3^{(0)}}c_-^{(0)}       & 0 & 0  & -ie^{-ic_3^{(0)}}s_-^{(0)} \\
    0       & e^{ic_3^{(0)}}c_+^{(0)} & -ie^{ic_3^{(0)}}s_+^{(0)} & 0 \\
    0       & -ie^{ic_3^{(0)}}s_+^{(0)} & e^{ic_3^{(0)}}c_+^{(0)} & 0 \\
    -ie^{-ic_3^{(0)}}s_-^{(0)}       & 0 & 0  & e^{-ic_3^{(0)}}c_-^{(0)}  
\end{array} \right) 
\label{eq:carmatrix}
\end{equation}
where,
$$
c_\pm^{(0)} = \cos(c_1^{(0)}\pm c_2^{(0)});\qquad s_\pm^{(0)} = \sin(c_1^{(0)}\pm c_2^{(0)}).
$$

The unitary operator $U$ in its canonical decomposition has four parameters. Let $\alpha_0 = e^{-ic_3^{(0)}}c_-^{(0)}, \beta_0 = -ie^{-ic_3^{(0)}}s_-^{(0)}, \gamma_0 = -ie^{ic_3^{(0)}}s_+^{(0)}, $ and $\delta_0 = e^{ic_3^{(0)}}c_+^{(0)}$.  Then the Eq. (\ref{eq:carmatrix}) can be written as 

\begin{equation}
    U_0 =     \left( \begin{array}{cccc}
    \alpha_0 & 0 & 0 & \beta_0 \\
    0 & \delta_0 & \gamma_0 & 0 \\
    0 & \gamma_0 & \delta_0 & 0 \\
    \beta_0 & 0 & 0 & \alpha_0 
    \end{array} \right).
\label{eq:carabc}
\end{equation}

The $\mathcal{M}_R$ map can now be studied analytically by applying it to the two qubit unitary operators in its Cartan form (\ref{eq:carabc}). Action of linear map $R$ on $U_0$ defined in Eq. (\ref{eq:carabc}) results 
\begin{equation}
    U^R_0 =     \left( \begin{array}{cccc}
    \alpha_0 & 0 & 0 & \delta_0 \\
    0 & \beta_0 & \gamma_0 & 0 \\
    0 & \gamma_0 & \beta_0 & 0 \\
    \delta_0 & 0 & 0 & \alpha_0 
    \end{array}  \right).
\label{eq:UR}
\end{equation} 
The polar decomposition of the matrix $U_0^R$, which is given by $U_0^R = U_1H$, where $U_1$ is unitary and $H = \sqrt{U^{R\dagger} U^R}$ is given by 

\begin{equation}
    H = \frac{1}{2}    \left( \begin{array}{cccc}
    |\alpha_0-\delta_0| + |\alpha_0+\delta_0| & 0 & 0 & -|\alpha_0-\delta_0| + |\alpha_0+\delta_0| \\
    0 & |\beta_0-\gamma_0| + |\beta_0+\gamma_0| & -|\beta_0-\gamma_0| + |\beta_0+\gamma_0| & 0 \\
    0 &  -|\beta_0-\gamma_0| + |\beta_0+\gamma_0| & |\beta_0-\gamma_0| + |\beta_0+\gamma_0| & 0 \\
    -|\alpha_0-\delta_0| + |\alpha_0+\delta_0| & 0 & 0 & |\alpha_0-\delta_0| + |\alpha_0+\delta_0| 
    \end{array} \right).
\label{eq:sqrtp}
\end{equation} 
The unitary $U_1= U_0^R H^{-1}$ is given by 
\begin{equation}
    U_1 =   \left( \begin{array}{cccc} 
    \alpha_0 \alpha_+ + \delta_0 \alpha_- & 0 & 0 & \alpha_0 \alpha_- + \delta_0 \alpha_+ \\
    0 & \beta_0 \beta_+ + \gamma_0 \beta_- & \beta_0 \beta_- + \gamma_0 \beta_+ & 0 \\
    0 & \gamma_0 \beta_+ + \beta_0 \beta_- & \gamma_0 \beta_- + \beta_0 \beta_+ & 0 \\
    \alpha_0 \alpha_- + \delta_0 \alpha_+   & 0 & 0 & \alpha_0 \alpha_+ + \delta_0 \alpha_-  
    \end{array} \right ),
\label{eq:U1}
\end{equation}

where $\alpha_{\pm}$ and $\beta_{\pm}$ are given as 
\begin{equation} 
\begin{split} 
\alpha_{\pm} &=  \frac{|\alpha_0-\delta_0| \pm |\alpha_0+\delta_0|}{2|\alpha_0-\delta_0||\alpha_0+\delta_0|} \\
\beta_{\pm} &=  \frac{|\beta_0-\gamma_0| \pm |\beta_0+\gamma_0|}{2|\beta_0-\gamma_0||\beta_0+\gamma_0|}
\end{split}
\end{equation} 
Note that although $U_0 \in \mathcal{SU}(4)$ but $U_1$ given by Eq.~(\ref{eq:U1}) need not be in $\mathcal{SU}(4)$ in general.
%$U_1$ is again four parameter unitary and lets denote these parameters as $\alpha_1, \beta_1, \gamma_1 $ and $\delta_1$. Then the mapping between $a,b,c,d$ of $U$ is mapped to $a_1, b_1, c_1$ and $d_1$ as 
The mapping between $\alpha_0,\beta_0,\gamma_0,\delta_0$ of $U_0$ and $\alpha_1', \beta_1', \gamma_1', \delta_1'$ of 
$U_1$ is
\begin{align}
\alpha_1' & = \frac{(\alpha_0+\delta_0)|\alpha_0-\delta_0| + (\alpha_0-\delta_0) |\alpha_0+\delta_0|}{2|\alpha_0-\delta_0||\alpha_0+\delta_0|} \nonumber \\
\beta_1' &= \frac{(\alpha_0+\delta_0)|\alpha_0-\delta_0| - (\alpha_0-\delta_0) |\alpha_0+\delta_0|}{2|\alpha_0-\delta_0||\alpha_0+\delta_0|} \nonumber \\
\gamma_1' &= \frac{(\beta_0+\gamma_0)|\beta_0-\gamma_0| - (\beta_0-\gamma_0)|\beta_0+\gamma_0|}{2|\beta_0-\gamma_0||\beta_0+\gamma_0|} \nonumber \\
\delta_1' &= \frac{(\beta_0+\gamma_0)|\beta_0-\gamma_0| + (\beta_0-\gamma_0)|\beta_0+\gamma_0|}{2|\beta_0-\gamma_0||\beta_0+\gamma_0|}.
\label{eq:amap}
\end{align} 
The above set of equations written in a compact form as
\begin{equation}
\begin{pmatrix}
\alpha_1' \\
\beta_1' \\
\gamma_1' \\
\delta_1' 
\end{pmatrix} = 
\begin{pmatrix}
\alpha_+ & 0 & 0 & \alpha_- \\
\alpha_- & 0 & 0 & \alpha_+ \\
0 & \beta_- & \beta_+ & 0 \\
0 & \beta_+ & \beta_- & 0 
\end{pmatrix} 
\begin{pmatrix}
\alpha_0 \\
\beta_0 \\
\gamma_0 \\
\delta_0 
\end{pmatrix},
\end{equation} 
depicts the non-linear nature of the map.

\section{Proofs of the fixed point theorems in the two-qubit case \label{app:fixed_point_thms}}
\subsection{Proof of Theorem 1}
\begin{proof}
Let $U_0$ be a two-qubit gate of the form Eq.~(\ref{eq:carabc0}).
If $U_0$ is fixed point of the $\mathcal{M}_{R}$ map: \beq
\mathcal{M}_{R}[U_0]=U_1=U_0.
\eeq
As this implies that $U_1$ is also in $\mathcal{SU}(4)$, $\chi_{1}=0$. 
Using Eq.~(\ref{eq:abcdmap}), the fixed point condition $\mathcal{M}_{R}[U_0]=U_0$ can be written as
\begin{subequations}
\label{eq:period_one}
\begin{align}
\label{eq:alp_0}
  \alpha_0 & =\frac{1}{2}[(k_+^{(0)}+k_-^{(0)})\alpha_0 + (k_+^{(0)}-k_-^{(0)})\delta_0],\\
\label{eq:bet_0}
  \beta_0 & =\frac{1}{2}[(k_+^{(0)}-k_-^{(0)})\alpha_0 + (k_+^{(0)}+k_-^{(0)})\delta_0],\\
\label{eq:gam_0}
  \gamma_0 & =\frac{1}{2}[(l_+^{(0)}-l_-^{(0)})\beta_0 + (l_+^{(0)}+l_-^{(0)})\gamma_0],\\
\label{eq:del_0}
  \delta_0 & =\frac{1}{2}[(l_+^{(0)}+l_-^{(0)})\beta_0 + (l_+^{(0)}-l_-^{(0)})\gamma_0],
 \end{align}
 \end{subequations}
where $k_{\pm}^{(0)}=1/|\alpha_0 \pm \delta_0|$ and $l_{\pm}^{(0)}=1/|\beta_0 \pm \gamma_0|$. From the unitarity of $U_0$, it follows that Re$(\alpha_0 \beta_0^*)=$Re$(\gamma_0 \delta_0^*)=0$. Multiplying Eq.~(\ref{eq:alp_0}) by $\gamma_0^*$, Eq.~(\ref{eq:gam_0}) by $\alpha_0^*$ in Eq.~(\ref{eq:period_one}) and taking real parts one obtains $k_{+}^{(0)}+k_{+}^{(0)}=l_{+}^{(0)}+l_{+}^{(0)}=2$. Similarly, multiplying Eq.~(\ref{eq:alp_0}) by $\beta_0^*$ and taking the real parts we get
\beq
(k_+^{(0)}-k_-^{(0)})\text{Re}[\delta_0 \beta_0^*]=0.
\label{eq:fixed_point_conditions}
\eeq
Therefore, either $k_+^{(0)}=k_-^{(0)}$ or Re$[\delta_0 \beta_0^*]=0$. For $k_+^{(0)}=k_-^{(0)}$ together with the condition $k_{+}^{(0)}+k_{+}^{(0)}=2$, from Eq.~(\ref{eq:bet_0}), it follows that $\beta_0=\delta_0$. For $U_0$ with $\beta_0=\delta_0 \neq 0$, it is can be shown that $\gamma_0=\pm \alpha_0$ using the unitarity of $U_0$. Therefore, $U_0$ is of the form
\beq
U_0=\left( \begin{array}{cccc}
    \alpha_0 & 0 & 0 & \beta_0 \\
    0 & \beta_0 & \pm \alpha_0  & 0 \\
    0 & \pm \alpha_0 & \beta_0 & 0 \\
    \beta_0 & 0 & 0 & \alpha_0 
    \end{array} \right)\;,
    \label{eq:gen_selfdual}
\eeq
which satisfies $U_0^R=U_0$ and thus is a self-dual unitary. From Eq.~(\ref{eq:bet_0}) and Eq.~(\ref{eq:del_0}) using unitarity of $U_0$ together with $k_{+}^{(0)}+k_{+}^{(0)}=l_{+}^{(0)}+l_{+}^{(0)}=2$, it follows that
$|\beta_0|^2=|\delta_0|^2=$ Re$(\beta_0 \delta_0^*)$. Therefore, the other condition Re$(\beta_0 \delta_0^*)=0$ in Eq.~(\ref{eq:fixed_point_conditions}) is satisfied only when $\beta_0=\delta_0=0$. In this case $U_0$ is of the form
\beq 
U_0 =     \left( \begin{array}{cccc}
    \alpha_0 & 0 & 0 & 0 \\
    0 & 0 & \gamma_0 & 0 \\
    0 & \gamma_0 & 0 & 0 \\
    0 & 0 & 0 & \alpha_0 
    \end{array} \right)\;,
    \label{eq:selfdual}
\eeq
where  $|\alpha_0|= |\gamma_0|=1$ and is also self-dual unitary.
Hence, all period-one fixed points of the $\mathcal{M}_R$ map in the two-qubit case are self-dual.
\end{proof}

Canonical form of the dual-unitaries obtained by setting $c_1=c_2=\pi/4$ in Eq.~(\ref{eq:carmatrix}) is of the form Eq.~(\ref{eq:selfdual}). Self-dual unitaries of the form Eq.~(\ref{eq:gen_selfdual}) are LU equivalent to the canonical form. For example, $(H \otimes H)U_0(H \otimes H)$, $H$ being the Hadamard gate; $H=\frac{1}{\sqrt{2}}\begin{pmatrix}
1 & 1 \\
1 & -1
\end{pmatrix}$, is of the canonical form where $U_0$ is of the form
\[
U_0=\left( \begin{array}{cccc}
    \alpha_0 & 0 & 0 & \beta_0 \\
    0 & \beta_0  & \alpha_0  & 0 \\
    0 & \alpha_0 & \beta_0 & 0 \\
    \beta_0 & 0 & 0 & \alpha_0 
    \end{array} \right)\;.
   \]
\subsection{Proof of Theorem 2} 
\begin{proof}

One way is easy: if $U_0$ is dual-unitary, then $\mathcal{M}_{R}[U_0]=U_0^R$ as $U_0^R$ is unitary and therefore $\mathcal{M}_{R}^2[U_0]=\mathcal{M}_{R}[U_0^R]=U_0$ i.e., $U_0$ is fixed point of the $\mathcal{M}_{R}^2$ map.

The other direction: that all fixed points of $\mathcal{M}_{R}^2$
are dual-unitary, is nontrivial. Analogous to the period-one case the fixed point equation $\mathcal{M}_{R}^2[U_0]=U_0$ can be written in terms of matrix elements. Action of $\mathcal{M}_{R}$ on $U_0$ leads to $U_1$ given by
\[
U_1=\left( \begin{array}{cccc}
    \alpha_1 & 0 & 0 & \beta_1 \\
    0 & \delta_1  & \gamma_1  & 0 \\
    0 & \gamma_1 & \delta_1 & 0 \\
    \beta_1 & 0 & 0 & \alpha_1 
    \end{array} \right)\;,
\]
where
\begin{subequations}
\label{eq:period_two_U1}
\begin{align}
\label{eq:alp_1}
  \alpha_1 & =\frac{1}{2}[(k_+^{(0)}+k_-^{(0)})\alpha_0 + (k_+^{(0)}-k_-^{(0)})\delta_0],\\
\label{eq:bet_1}
  \beta_1 & =\frac{1}{2}[(k_+^{(0)}-k_-^{(0)})\alpha_0 + (k_+^{(0)}+k_-^{(0)})\delta_0],\\
\label{eq:gam_1}
  \gamma_1 & =\frac{1}{2}[(l_+^{(0)}-l_-^{(0)})\beta_0 + (l_+^{(0)}+l_-^{(0)})\gamma_0],\\
\label{eq:del_1}
  \delta_1 & =\frac{1}{2}[(l_+^{(0)}+l_-^{(0)})\beta_0 + (l_+^{(0)}-l_-^{(0)})\gamma_0],
 \end{align}
 \end{subequations}
where $k_{\pm}^{(0)}=1/|\alpha_0 \pm \delta_0|$ and $l_{\pm}^{(0)}=1/|\beta_0 \pm \gamma_0|$ and we have ignored the overall phase as it does not affect the proof. As $\mathcal{M}_{R}^2[U_0]:=\mathcal{M}_{R}[\mathcal{M}_{R}[U_0]]=\mathcal{M}_{R}[U_1]=U_0$, therefore mapping among the matrix elements is given by
\begin{subequations}
\label{eq:period_two_U2}
\begin{align}
\label{eq:alp_2}
 & \alpha_0=\frac{1}{2}[(k_+^{(1)}+k_-^{(1)})\alpha_1 + (k_+^{(1)}-k_-^{(1)})\delta_1],\\
 \label{eq:bet_2}
 & \beta_0=\frac{1}{2}[(k_+^{(1)}-k_-^{(1)})\alpha_1 + (k_+^{(1)}+k_-^{(1)})\delta_1],\\
 \label{eq:gam_2}
 & \gamma_0=\frac{1}{2}[(l_+^{(1)}-l_-^{(1)})\beta_1 + (l_+^{(1)}+l_-^{(1)})\gamma_1],\\
 \label{eq:del_2}
 & \delta_0=\frac{1}{2}[(l_+^{(1)}+l_-^{(1)})\beta_1 + (l_+^{(1)}-l_-^{(1)})\gamma_1],
 \end{align}
 \end{subequations}
where $k_{\pm}^{(1)}=1/|\alpha_1 \pm \delta_1|$ and $l_{\pm}^{(1)}=1/|\beta_1 \pm \gamma_1|$. Using unitarity constraints Re$(\alpha_n \beta_n^*)=$ Re$(\gamma_n\delta_n^*)=0$ ($n=0,1$), we will simplify above set of equations.

%\begin{itemize}
%\item
%Unitarity of $U_0$: Re$(\alpha_0 \beta_0^*)=$ Re$(\gamma_0\delta_0^*)=0$. Multiplying Eq.~(\ref{eq:alp_2}) by $\beta_0^*$ and taking real parts on both sides we get
%\beq
%\frac{\text{Re}[\alpha_1 \beta_0^*]}{\text{Re}[\delta_1 \beta_0^*]}=-\frac{k_+^{(1)}-k_+^{(1)}}{k_+^{(1)}+k_+^{(1)}}.
%\label{eq:alp0_bet0star}
%\eeq
%Similarly, multiplying Eq.~(\ref{eq:gam_2}) by $\delta_0^*$ and taking real parts we get
%\beq
%\frac{\text{Re}[\gamma_1 \delta_0^*]}{\text{Re}[\beta_1 \delta_0^*]}=-\frac{l_+^{(1)}-l_+^{(1)}}{l_+^{(1)}+l_+^{(1)}}.
%\label{eq:gam0_del0star}
%\eeq
%\item
%Unitarity of $U_1$: Re$(\alpha_1 \beta_1^*)=$ Re$(\gamma_1\delta_1^*)=0$. Multiplying Eq.~(\ref{eq:alp_1}) by $\beta_1^*$ and taking real parts on both sides we get
%\beq
%\frac{\text{Re}[\alpha_0 \beta_1^*]}{\text{Re}[\delta_0 \beta_1^*]}=-\frac{k_+^{(0)}-k_+^{(0)}}{k_+^{(0)}+k_+^{(0)}}.
%\label{eq:alp1_bet1star}
%\eeq
%Similarly, multiplying Eq.~(\ref{eq:gam_1}) by $\delta_1^*$ and taking real parts we get
%\beq
%\frac{\text{Re}[\gamma_0 \delta_1^*]}{\text{Re}[\beta_0 \delta_1^*]}=-\frac{l_+^{(0)}-l_+^{(0)}}{l_+^{(0)}+l_+^{(0)}}.
%\label{eq:gam0_del0star}
%\eeq
%\end{itemize}

Multiplying Eq.~(\ref{eq:alp_2}) by $\gamma_1^*$ and taking real parts we get
\beq
\text{Re}[\alpha_0\gamma_1^*]=\frac{1}{2}(k_+^{(1)}+k_-^{(1)})\text{Re}[\alpha_1 \gamma_1^*].
\label{eq:alp0_gam1star}
\eeq
Now, multiplying $\alpha_0$ to the complex conjugate of Eq.~(\ref{eq:gam_1}) and taking real parts we get
\beq
\text{Re}[\alpha_0\gamma_1^*]=\frac{1}{2}(l_+^{(0)}+l_-^{(0)})\text{Re}[\alpha_0 \gamma_0^*].
\label{eq:gam1star_alp0}
\eeq
From Eq.~(\ref{eq:alp0_gam1star}) and Eq.~(\ref{eq:gam1star_alp0}) it follows that
\beq
\frac{\text{Re}[\alpha_1 \gamma_1^*]}{\text{Re}[\alpha_0 \gamma_0^*]}=\frac{l_+^{(0)}+l_-^{(0)}}{k_+^{(1)}+k_-^{(1)}}.
\label{eq:alp_gam_ratio_1}
\eeq
Similarly, mutiplying Eq.~(\ref{eq:alp_1}) by $\gamma_0^*$,  complex conjugate of Eq.~(\ref{eq:gam_1}) by $\alpha_0$, and taking real parts leads to
\beq
\frac{\text{Re}[\alpha_1 \gamma_1^*]}{\text{Re}[\alpha_0 \gamma_0^*]}=\frac{k_+^{(0)}+k_-^{(0)}}{l_+^{(1)}+l_-^{(1)}}.
\label{eq:alp_gam_ratio_2}
\eeq
From Eqs.~(\ref{eq:alp_gam_ratio_1}) and ~(\ref{eq:alp_gam_ratio_2}), it follows that
\beq
\frac{l_+^{(0)}+l_-^{(0)}}{k_+^{(1)}+k_-^{(1)}}=\frac{k_+^{(0)}+k_-^{(0)}}{l_+^{(1)}+l_-^{(1)}}.
\label{eq:alp_gam_ratio_1_2}
\eeq
Multiplying complex conjugate of Eq.~(\ref{eq:bet_2}) by $\beta_1$ and taking real parts
\beq
\text{Re}[\beta_1\beta_0^*]=\frac{1}{2}(k_+^{(1)}+k_-^{(1)})\text{Re}[\delta_1^*\beta_1].
\label{eq:bet0star_bet1}
\eeq
Now, multiplying Eq.~(\ref{eq:bet_1}) by $\beta_0^*$ and taking real parts
\beq
\text{Re}[\beta_1\beta_0^*]=\frac{1}{2}(k_+^{(0)}+k_-^{(0)})\text{Re}[\delta_0 \beta_0^*].
\label{eq:bet1_bet0star}
\eeq
From Eq.~(\ref{eq:bet0star_bet1}) and Eq.~(\ref{eq:bet1_bet0star}) it follows that
\beq
\frac{\text{Re}[\delta_1 \beta_1^*]}{\text{Re}[\delta_0 \beta_0^*]}=\frac{k_+^{(0)}+k_-^{(0)}}{k_+^{(1)}+k_-^{(1)}}.
\label{eq:bet_del_ratio_1}
\eeq
Similarly, mutiplying Eq.~(\ref{eq:del_1}) by $\delta_0^*$,  complex conjugate of Eq.~(\ref{eq:del_2}) by $\delta_1$, and taking real parts we get
\beq
\frac{\text{Re}[\delta_1 \beta_1^*]}{\text{Re}[\delta_0 \beta_0^*]}=\frac{l_+^{(0)}+l_-^{(0)}}{l_+^{(1)}+l_-^{(1)}}.
\label{eq:bet_del_ratio_2}
\eeq
From Eqs.~(\ref{eq:bet_del_ratio_1}) and (\ref{eq:bet_del_ratio_2}) it follows that
\beq
\frac{k_+^{(0)}+k_-^{(0)}}{k_+^{(1)}+k_-^{(1)}}=\frac{l_+^{(0)}+l_-^{(0)}}{l_+^{(1)}+l_-^{(1)}}.
\label{eq:bet_del_ratio_1_2}
\eeq
From Eq.~(\ref{eq:alp_gam_ratio_1_2}) and Eq.~(\ref{eq:bet_del_ratio_1_2}) it is easy to check that
\beq
k_+^{(0)}+k_-^{(0)}=l_+^{(0)}+l_-^{(0)},
k_+^{(1)}+k_-^{(1)}=l_+^{(1)}+l_-^{(1)}
\label{eq:k_l_n}
\eeq
Multiplying Eq.~(\ref{eq:alp_2}) by $\gamma_0^*$ and taking real parts we get
\beq
\begin{split}
\text{Re}[\alpha_0\gamma_0^*]& =\frac{1}{2}\left[(k_+^{(1)}+k_-^{(1)})\text{Re}[\alpha_1\gamma_0^*]+(k_+^{(1)}-k_-^{(1)})\text{Re}[\delta_1\gamma_0^*]\right],\\
& = \frac{(k_+^{(1)}+k_-^{(1)})(k_+^{(0)}+k_-^{(0)})}{4}\text{Re}[\alpha_0\gamma_0^*] \\
& + \frac{(k_+^{(1)}-k_-^{(1)})(l_+^{(1)}-l_-^{(1)})}{4}\text{Re}[\delta_1\beta_1^*] ,
\end{split}
\eeq
where second equation is obtained using Re$[\alpha_1\gamma_0^*]=(k_+^{(0)}+k_-^{(0)})/2$ Re$[\alpha_0\gamma_0^*]$ and Re$[\delta_1\gamma_0^*]=(l_+^{(1)}-l_-^{(1)})/2$ Re$[\delta_1\beta_1^*]$. Using Eq.~(\ref{eq:bet_del_ratio_1}) in above equation one obtains
\beq
\begin{split}
\text{Re}[\alpha_0\gamma_0^*]& = \frac{(k_+^{(1)}+k_-^{(1)})(k_+^{(0)}+k_-^{(0)})}{4}\text{Re}[\alpha_0\gamma_0^*] \\
& + \frac{(k_+^{(1)}-k_-^{(1)})(l_+^{(1)}-l_-^{(1)})}{4} \frac{(k_+^{(0)}+k_-^{(0)})}{(k_+^{(1)}+k_-^{(1)})}\text{Re}[\delta_0\beta_0^*] ,
\end{split}
\label{eq:alp_0_gam_0}
\eeq
Now, multiplying conjugate of Eq.~(\ref{eq:bet_2}) by $\delta_0$ and taking real parts 
\beq
\begin{split}
\text{Re}[\beta_0^*\delta_0]& =\frac{1}{2}\left[(k_+^{(1)}-k_-^{(1)})\text{Re}[\alpha_1^*\delta_0]+(k_+^{(1)}+k_-^{(1)})\text{Re}[\delta_1^*\delta_0]\right],\\
& = \frac{(k_+^{(1)}-k_-^{(1)})(l_+^{(1)}-l_-^{(1)})}{4}\text{Re}[\alpha_1^*\gamma_1] \\
& + \frac{(k_+^{(1)}+k_-^{(1)})(l_+^{(1)}+l_-^{(1)})}{4}\text{Re}[\delta_1^*\beta_1],
\end{split}
\eeq
where second equation is obtained using Re$[\alpha_1^*\delta_0]=(l_+^{(1)}-l_-^{(1)})/2$ Re$[\alpha_1^* \gamma_1]$ and Re$[\delta_1^* \delta_0]=(l_+^{(1)}+l_-^{(1)})/2$ Re$[\delta_1^*\beta_1]$. Using Eqs.~(\ref{eq:alp_gam_ratio_1},~(\ref{eq:bet_del_ratio_1} in above equation we get
\beq
\begin{split}
\text{Re}[\beta_0^*\delta_0]& = \frac{(k_+^{(1)}-k_-^{(1)})(l_+^{(1)}-l_-^{(1)})}{4}\frac{(k_+^{(0)}+k_-^{(0)})}{(k_+^{(1)}+k_-^{(1)})} \text{Re}[\alpha_0 \gamma_0^*] \\
& + \frac{(k_+^{(0)}+k_-^{(0)})(l_+^{(1)}+l_-^{(1)})}{4} \text{Re}[\delta_0^*\beta_0],
\end{split}
\label{eq:bet_0_del_0}
\eeq
From Eqs.~(\ref{eq:alp_0_gam_0}),~(\ref{eq:bet_0_del_0}) and ~(\ref{eq:k_l_n}), we get
\beq
\begin{split}
& k_+^{(1)}=k_-^{(1)}, l_+^{(1)}=l_-^{(1)},\\
& (k_+^{(0)}+k_-^{(0)})(k_+^{(1)}+k_-^{(1)})=(l_+^{(0)}+l_-^{(0)})(l_+^{(1)}+l_-^{(1)})=4.
\label{eq:k_pm_l_pm_1}
\end{split}
\eeq
%\beq
%\begin{split}
%& \frac{\text{Re}[\alpha_0 \beta_1^*]}{\text{Re}[\beta_0 \beta_1^*]}=\frac{k_+^{(1)}-k_+^{(1)}}{k_+^{(1)}+k_+^{(1)}}, \frac{\text{Re}[\alpha_1 \beta_0^*]}{\text{Re}[\delta_1 \beta_0^*]}=-\frac{k_+^{(1)}-k_+^{(1)}}{k_+^{(1)}+k_+^{(1)}}, \\
%& \frac{\text{Re}[\alpha_0 \beta_1^*]}{\text{Re}[\delta_0 \beta_1^*]}=-\frac{k_+^{(0)}-k_+^{(0)}}{k_+^{(0)}+k_+^{(0)}},\frac{\text{Re}[\alpha_1 \beta_0^*]}{\text{Re}[\beta_1 \beta_0^*]}=\frac{k_+^{(0)}-k_+^{(0)}}{k_+^{(0)}+k_+^{(0)}}, \\
% \end{split}
% \label{eq:period_two_constraints}
%\eeq

%For non-zero $\alpha_0, \beta_0,\gamma_0,$ and $\delta_0$ it can be shown that
%\beq
%\frac{k_+^{(1)}-k_+^{(1)}}{k_+^{(1)}+k_+^{(1)}}=\frac{k_+^{(0)}-k_+^{(0)}}{k_+^{(0)}+k_+^{(0)}}, \frac{l_+^{(1)}-l_+^{(1)}}{l_+^{(1)}+l_+^{(1)}}=\frac{l_+^{(0)}-l_+^{(0)}}{l_+^{(0)}+l_+^{(0)}}.
%\eeq
%Therefore, from Eq.~(\ref{period_two_constraints}) it follows that
%\beq
%\text{Re}(\alpha_0 \beta_1^*) \left[\frac{1}{\text{Re}(\beta_0 \beta_1^*)}+\frac{1}{\text{Re}(\delta_0 \beta_1^*)} \right]=0.
%\eeq 
%which implies that either Re$(\alpha_0 \beta_1^*)=0$ or, $\frac{1}{\text{Re}(\beta_0 \beta_1^*)}+\frac{1}{\text{Re}(\delta_0 \beta_1^*)}=0$. Re$(\alpha_0 \beta_1^*)=(k_+^{(1)}-k_-^{(1)})\text{Re}(\delta_1\beta_1^*)=0$ gives
%\beq
%k_+^{(1)}=k_-^{(1)}, \text{Re}(\delta_1\beta_1^*)=0.
%\eeq
A similar calculation of Re$[\alpha_1 \gamma_1^*]$ and Re$[\beta_1 \delta_1^*]$ implies that
\beq
\begin{split}
& k_+^{(0)}=k_-^{(0)}, l_+^{(0)}=l_-^{(0)},\\
& (k_+^{(0)}+k_-^{(0)})(k_+^{(1)}+k_-^{(1)})=(l_+^{(0)}+l_-^{(0)})(l_+^{(1)}+l_-^{(1)})=4.
\label{eq:k_pm_l_pm_2}
\end{split}
\eeq
Using Eqs.~(\ref{eq:k_pm_l_pm_2})-(\ref{eq:k_pm_l_pm_2}) in Eq.~(\ref{eq:period_two_U1}) and Eq.~(\ref{eq:period_two_U2}) it follows that $k_{\pm}^{(n)}=l_{\pm}^{(n)}=1$ where $n=0,1$, as in the period-one case. For  $k_+^{(1)}=k_-^{(1)}=1$, $\alpha_1=\alpha_0, \beta_1=\delta_0,\gamma_1=\gamma_0,$ and $\delta_1=\beta_0$. Therefore, Re$(\alpha_0 \delta_0^*)=$ Re$(\beta_0 \gamma_0^*)=0$ which implies that $U_0^R$ given by
\[
U_0^R=\left( \begin{array}{cccc}
    \alpha_0 & 0 & 0 & \delta_0 \\
    0 & \beta_0  & \gamma_0  & 0 \\
    0 & \gamma_0 & \beta_0 & 0 \\
    \delta_0 & 0 & 0 & \alpha_0 
    \end{array} \right)
\] 
is also unitary (note that $\alpha_0$ and $\gamma_0$ do not change positions under the realignment operation) and thus $U_0$ is dual-unitary. Hence, $\mathcal{M}_{R}^2[U_0]=U_0$ implies that $U_0$ is dual-unitary. 
\end{proof}
We have assumed that $\alpha_0,\beta_0,\delta_0,$ and $\gamma_0$ are all non-zero. It is easy to verify from Eqs.~(\ref{eq:period_two_U1})-~(\ref{eq:period_two_U2}) that if $\alpha_0=0$ then $\gamma_0=0$ and if $\beta_0=0$ then $\delta_0=0$. In the former case seed unitary is dual-unitary of the form 
\beq
U_0^R=\left( \begin{array}{cccc}
    0 & 0 & 0 & \beta_0 \\
    0 & \delta_0  & 0  & 0 \\
    0 & 0 & \delta_0 & 0 \\
    \beta_0  & 0 & 0 & 0 
    \end{array} \right)
\eeq
and it the later case $U_0$ is self-dual unitary of the canonical form Eq.~(\ref{eq:selfdual}). 

For two-qubit gates of the form Eq.~(\ref{eq:carmatrix}) without any restriction on the parameters, the general form of a two-qubit dual-unitary gates obtained from the map is
\[
U_0=\left( \begin{array}{cccc}
    \alpha_0 & 0 & 0 & \beta_0 \\
    0 & \pm \beta_0  & \pm \alpha_0  & 0 \\
    0 & \pm \alpha_0  & \pm \beta_0 & 0 \\
    \beta_0 & 0 & 0 & \alpha_0 
    \end{array} \right)
\] 
 
\section{Details about the map in terms of Cartan parameters \label{app:CartanAlg}}

\subsection{Map in terms of Cartan parameters}
\subsubsection{XXX family}
For $c_1^{(n)}=c_1^{(n)}=c_1^{(n)}=c^{(n)}$, the complex number arguments appearing in Eq.~(\ref{eq:cartan_gen2}) simplify to
\beq
\begin{split}
\theta_{+}^{(n)}= & -\arctan\left[\frac{1-\cos 2\,c^{(n)}}{1+\cos 2\,c^{(n)}} \tan c^{(n)}\right],\\
\theta_{-}^{(n)}= & -\arctan\left[\frac{1+\cos 2\,c^{(n)}}{1-\cos 2\,c^{(n)}} \tan c^{(n)}\right],\\
\phi_+^{(n)}= & -\arctan\left[\frac{1}{\tan c^{(n)}}\right],\\
\phi_-^{(n)}= & \pi-\arctan\left[\frac{1}{\tan c^{(n)}}\right].
\end{split}
\label{eq:comp_arg}
\eeq
Using above equation Eq.~(\ref{eq:cartan_gen2}) simplifies to
\beq
c_1^{(n+1)}=c_3^{(n+1)}=c^{(n+1)}=\frac{\pi}{4}-\frac{1}{4}\arctan\left[\frac{2}{\tan 2\,c^{(n)}}\right],
\eeq
for all $n$ and 
\beq
\begin{split}
c_2^{(n+1)}= & c^{(n+1)} \; \text{for odd $n$},\\
		   = & \frac{\pi}{2}- c^{(n+1)} \; \text{for even $n$}
\end{split}
\eeq
where Cartan coefficients satisfy $0 \leq c_3^{(n)} \leq c_2^{(n)} \leq c_1^{(n)}$ for all $n$.
Thus the map on Cartan parameters is 1-dimensional given by
 \beq
 c^{(n+1)}=\frac{\pi}{4}-\frac{1}{4}\arctan\left[\frac{2}{\tan 2\,c^{(n)}}\right].
 \label{eq:XXXcnmap}
 \eeq
It is easy to check that $c^{*}=\pi/4$ is the fixed point of the map. We show that it is global attractor for all $c^{(0)} \in (0,\pi/4]$ below. 
In terms of $x_n=1/\tan 2\,c^{(n)}$, Eq.~(\ref{eq:XXXcnmap}) becomes
\beq
x_n=\frac{x_{n+1}}{1-x_{n+1}^2},
\eeq
which under rearrangement gives Eq.~(\ref{eq:XXXunmap}). 

To prove convergence we write the map in terms of $x_n$ defined in Eq.~(\ref{eq:xndef}) as
\beq
\frac{x_{n+1}}{x_n}=\frac{2}{1+\sqrt{(2\,x_n)^2+1}}.
\label{eq:unmapconv}
\eeq
As $ 1 < \sqrt{(2\,x_n)^2+1}$ for all $x_n \in (0,\infty)$. Therefore,
$$1+\sqrt{(2\,x_n)^2+1} > 2 \implies \frac{2}{1+\sqrt{(2\,x_n)^2+1}} < 1.$$ 
Hence from Eq. \ref{eq:unmapconv}, $x_{n+1} < x_n$ and explains the contractive nature of the map. The convergence of the map can also be justified in terms of it's Jacobian given by 
$$J_x=\frac{d}{dx}{\left[\frac{2\,x}{1+\sqrt{4\,x^2+1}}\right]}=\frac{2}{1+4\,x^2+\sqrt{1+4\,x^2}},$$
with $J_x < 1 \;\forall \;x \in (0,\infty)$.
The approach to the fixed point $x^*=0$, or, equivalently the approach of $U_n$ to the {\sc swap} gate, is  {\em algebraic} with exponent equal to 1/2 as shown in Fig.~(\ref{fig:XXXMRmap}) for $x_0=1$. Numerical values in Fig.~(\ref{fig:XXXMRmap} obtained from the algorithm proposed in \cite{Kus2013} exactly match the analytical values obtained from Eq.~(\ref{eq:XXXunmap}).
\begin{figure}
\centering
\includegraphics[scale=.45]{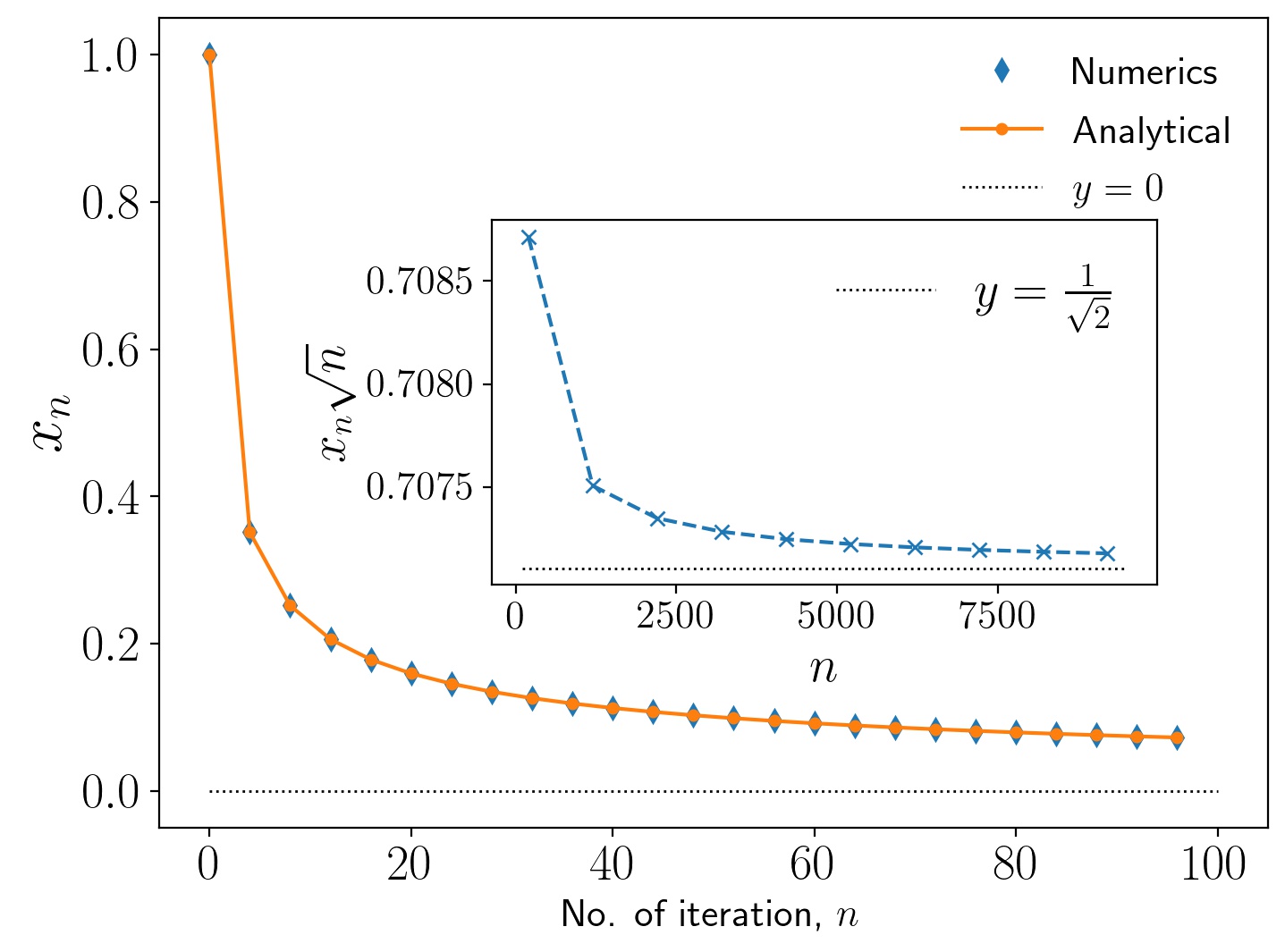}
\caption{XXX case: (Main) Convergence to the respective fixed point $x^*=0$ under the map initiated with $c^{(0)}=\pi/8$, or equivalently $x_0=1$. (Inset) Power law behaviour for the same initial condition as seen from the convergence $x_n \sqrt{n}\rightarrow \frac{1}{\sqrt{2}}$ for large $n$.}
\label{fig:XXXMRmap}
\end{figure}
\subsubsection{{\sc swap-cnot-dcnot} face}
In this case the substitution $c_1^{(n)}=\pi/4$ and assuming $0 \leq c_3^{(n)} \leq c_2^{(n)} \leq c_1^{(n)} \leq \pi/4$, simplifies the complex number arguments in the 3-dimensional map Eq.~(\ref{eq:cartan_gen2}) to
\beq
\begin{split}
\theta_{+}^{(n)}= & -\arctan\left[\tan c_2^{(n)} \tan c_3^{(n)}\right],\\
\theta_{-}^{(n)}= & -\arctan\left[\frac{\tan c_3^{(n)}}{\tan c_2^{(n)}}\right],\\
\phi_+^{(n)}= & -\arctan\left[\frac{1}{\tan c_2^{(n)}\,\tan c_3^{(n)}}\right],\\
\phi_-^{(n)}= & \pi-\arctan\left[\frac{\tan c_2^{(n)}}{\tan c_3^{(n)}}\right].
\end{split}
\label{eq:comp_arg_scd_face}
\eeq
Using above set of equations in Eq.~(\ref{eq:cartan_gen2}) the map on Cartan coefficients reduces to 
\beq
\begin{split}
& c_1^{(n+1)}=\frac{\pi}{4},\\
& c_2^{(n+1)}=\frac{\pi}{4} \pm \frac{1}{2}\arctan\left[\sin 2\,c_3^{(n)}\cot 2\, c_2^{(n)}\right],\\
& c_3^{(n+1)}= \frac{1}{2}\arctan\left[\frac{\tan 2\,c_3^{(n)}}{\sin 2c_2^{(n)}}\right].
\end{split}
\eeq
Thus the map is 2-dimensional and in terms of $y_n=1/\tan^2 2\,c_2^{(n)}$ and $z_n=1/\tan^2 2\,c_3^{(n)}$ the above 2-dimensional map takes a purely algebraic form given by Eq.~(\ref{eq:xnyn_alg_form} ) in the main text.
\subsubsection{{\sc swap-cnot} edge}
In this case the map on Cartan parameters is 1-dimensional; $c_1^{(n+1)}=\pi/4,c_2^{(n+1)}=c_3^{(n+1)}=c^{(n+1)}$, given by
\beq
c^{(n+1)}=\frac{1}{2}\arctan\left[\frac{1}{\cos 2\,c^{(n)}}\right].
\label{eq:XYYmodel}
\eeq
In terms of $t_n=\tan^2(2\,c^{(n)})$, the above map takes a simple linear form
\beq
t_{n+1}=1+t_n,
\eeq
hence $t_n=n-1+t_0$. In terms of $y_n=1/\tan^2(2\,c^{(n)})$ reduces to Eq.~(\ref{eq:XYYunmap}) with exact solution given by Eq.~(\ref{eq:XYYsol}). In this case the approach to the fixed point $y^*=0$ or, equivalently the approach of $U_n$ to the {\sc swap} gate, is {\em algebraic} with exponent equal to 1/2 as shown in Fig.~(\ref{fig:XYYMRmap}).
\begin{figure}
\centering
\includegraphics[scale=.45]{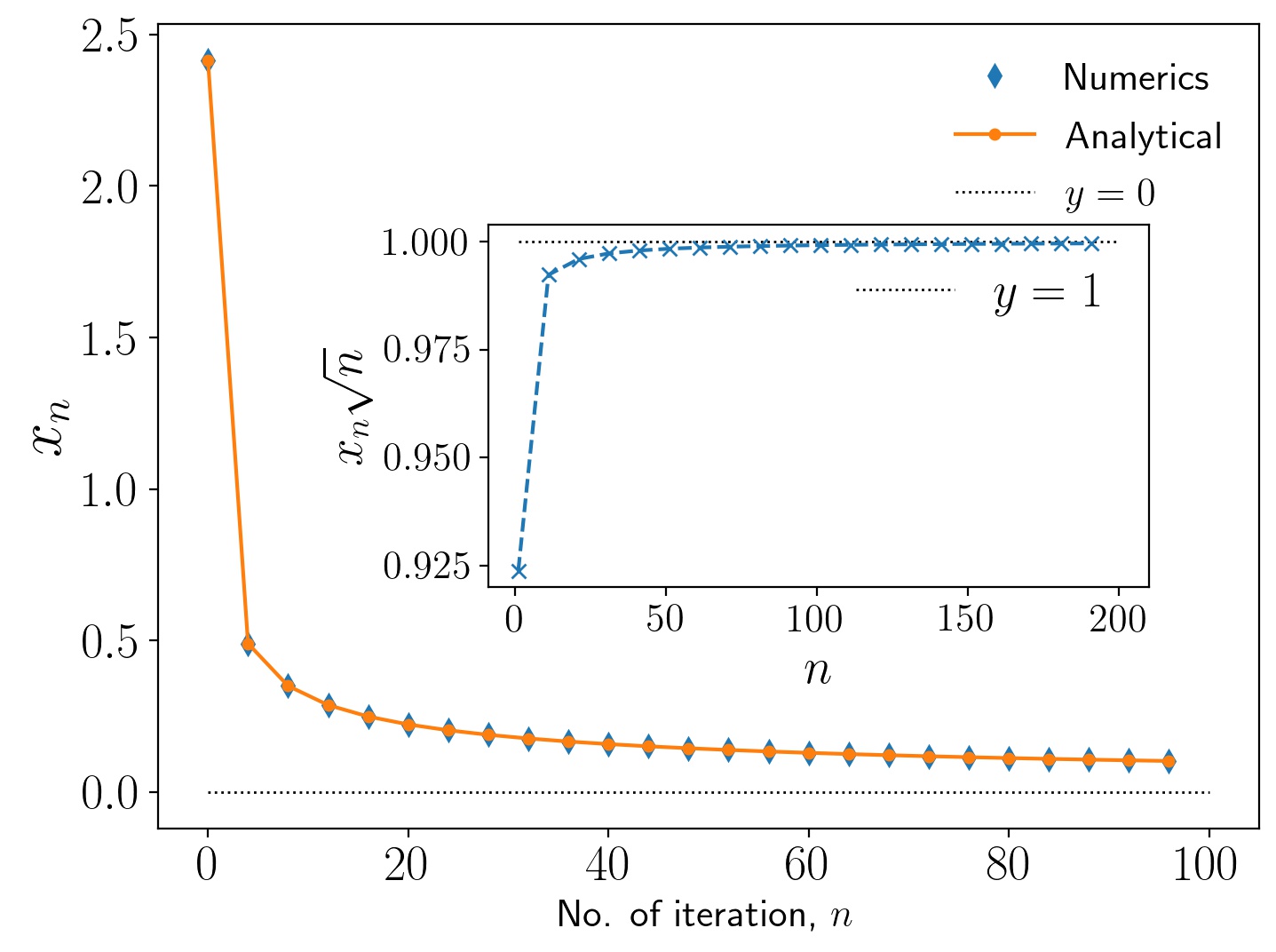}
\caption{{\sc swap-cnot} edge: (Main) Convergence to the respective fixed point $y^*=0$ under the map initiated with $c^{(0)}=\frac{\pi}{16}$. (Inset) Power law behaviour for the same initial condition as seen from the convergence $y_n \sqrt{n}\rightarrow 1 $ for large $n$.}
\label{fig:XYYMRmap}
\end{figure}

\section{2-unitaries with entangled rows and columns \label{app:general_comb_design}}
We write unitary operator $U$ in block form as 
$U=\sum_{i,j=1}^{d}\ket{i}\bra{j} \otimes X_{ij}.$
\begin{enumerate}
\item
Unitarity of $U$: $UU^{\dagger}=U^{\dagger} U=\mathbb{I}_{d^2}$.
\begin{align*}
UU^{\dagger} & =\left(\sum_{i,j=1}^{d} \ket{i} \bra{j} \otimes X_{ij} \right) \left(\sum_{k,l=1}^{d} \ket{k} \bra{l} \otimes X_{kl}\right)^{\dagger},\\
 			& = \sum_{i,k=1}^d \left( \ket{i} \bra{k} \otimes \sum_{j=1}^d X_{ij} X_{kj}^{\dagger} \right),
\end{align*}
$UU^{\dagger}=\mathbb{I}_{d^2}$ gives
\beq
\sum_{j=1}^d X_{ij} X_{kj}^{\dagger}= \delta_{ik}\mathbb{I}_{d}.
\label{eq:UOQLS}
\eeq

\item
Dual-unitarity of $U$: $U^RU^{R \, \dagger}=U^{R \, \dagger} U^R=\mathbb{I}_{d^2}$.
\begin{align*}
U^RU^{R \, \dagger} & =\left(\sum_{ij} \ket{ij} \bra{X_{ij}^*}\right) \left(\sum_{kl} \ket{kl} \bra{X_{kl}^*}\right)^{\dagger},\\
 			& = \sum_{ij} \sum_{kl} \langle X_{kl}|X_{ij} \rangle \ket{ij} \bra{kl},
\end{align*}
$U^RU^{R \, \dagger}=\mathbb{I}_{d^2}$ gives
\beq
\langle X_{kl}|X_{ij} \rangle=\delta_{ik}\delta_{jl},
\label{eq:UROQLS}
\eeq
 i.e, $X_{ij}$'s form an orthonormal operator basis.

\item
T-duality of $U$: $U^{\Gamma} U^{\Gamma \, \dagger}=\mathbb{I}_{d^2}$.
\begin{align*}
U^{\Gamma} U^{\Gamma \, \dagger}& =\left(\sum_{i,j=1}^{d} \ket{j} \bra{i} \otimes X_{ij} \right) \left(\sum_{k,l=1}^{d} \ket{l} \bra{k} \otimes X_{kl}\right)^{\dagger},\\
 			& = \sum_{j,l=1}^d \left( \ket{j} \bra{l} \otimes \sum_{i=1}^d X_{ij} X_{il}^{\dagger} \right),
\end{align*}
$U^{\Gamma} U^{\Gamma \, \dagger}=\mathbb{I}_{d^2}$ gives
\beq
\sum_{i=1}^d X_{ij} X_{il}^{\dagger}= \delta_{jl}\mathbb{I}_{d}.
\label{eq:UGOQLS}
\eeq
\end{enumerate}

The above conditions are equivalent to those presented in Ref.~\cite{SRatherAME46} and Ref.~\cite{rico2020absolutely} in terms of single qudit reduced density matrices (marginals) of a two-qudit pure state; $X_{ij} \mapsto \ket{X_{ij}}=(X_{ij} \otimes \mathbb{I}) \ket{\Phi}$, where $\ket{\Phi}$ is the generalized Bell state. It is easy to see that the marginal with respect to the first qudit is  given by $X_{ij} X_{ij}^{\dagger}$. Thus conditions (1) and (2) above involve orthonormality of sums of single qudit marginals in each row and each column respectively. A $d \times d$ arrangement of $d^2$-dimensional vectors (two-qudit quantum states) $\ket{X_{ij}}$ which satisfy Eqs.~(\ref{eq:UOQLS})--(\ref{eq:UGOQLS}) form an orthogonal quantum Latin square (OQLS). This generalizes the notion of OLS for general 2-unitary operators not necessarily permutations or those for which single qudit marginals are projectors where each $\ket{X_{ij}}$ is a product state. The original definitions of OQLS \cite{GRMZ_2018,MV19} are fragile in the sense they do not work when $\ket{X_{ij}}$ are entangled. Note that the entanglement of two-qudit states $\ket{X_{ij}}$'s change when unitary $U=\sum_{i,j=1}^{d}\ket{i}\bra{j} \otimes X_{ij}$, is multiplied by local unitary transformations.%which is not preserved under local unitary transformations.

%When $\ket{X_{ij}}$ are unentangled, for example for 2-unitary permutations, then each single qudit marginal is a one dimensional projector. When $\ket{X_{ij}}$ are entangled, which is mostly the case for 2-unitaries obtained from numerical algorithms \cite{SAA2020}, then 

\section{Fixed point of $\mathcal{M}_R^2$ map that is not dual-unitary \label{app:nondualFP}}
Here we give an example of a fixed point of the $\mathcal{M}_R^2$ map in $d^2=9$ which is not dual-unitary. The unitary is given by
\beq
U_{\text{nd}}=
\begin{pmatrix}
1 &  0     &  0     &  0     &  0     &  0     &
         0     &  0     &  0     \\
0 &  0     & \sqrt{3}/2 &  1/2    &  0     &  0     &
         0     &  0     &  0     \\
         
0 &  0     &  0     &  0     &  0     & \sqrt{3}/2 &
        -1/2    &  0     &  0     \\
        
0 &  0     &  0     &  0     &  0     &  0     &
         0     &  1     &  0     \\
         
0 &  1     &  0     &  0     &  0     &  0     &
        0     &  0     &  0     \\
        
0 &  0     &  0     &  0     &  1     &  0     &
         0     &  0     &  0     \\
         
0 &  0     &  0     &  0     &  0     &  0     &
        0     &  0     &  1     \\
        
0 &  0     & -1/2    & \sqrt{3}/2 &  0     &  0     &
         0     &  0     &  0     \\
         
0 &  0     &  0     &  0     &  0     &  1/2    &
        \sqrt{3}/2 &  0     &  0
\end{pmatrix}.
\eeq
Action of $\mathcal{M}_R$ map on $U_{\text{nd}}$ results in 
\begin{align}
U_{\text{nd}}'& :=\mathcal{M}_R[U_{\text{nd}}] \\
&=\begin{pmatrix}
\sqrt{3}/2 &  0     &  0     &  0     &  0     &  1/2    &
         0     &  0     &  0     \\
0     &  0     &  0     &  0     &  0     &  0     &
         0     &  0     &  1     \\
0     &  0     &  1/2    &  0     &  0     &  0     &
        -\sqrt{3}/2 &  0     &  0     \\
0     &  0     &  0     &  0     &  1     &  0     &
         0     &  0     &  0     \\
0     &  0     &  0     &  0     &  0     &  0     &
        0     &  1     &  0     \\
0     &  1     &  0     &  0     &  0     &  0     &
         0     &  0     &  0     \\
1/2    &  0     &  0     &  0     &  0     & -\sqrt{3}/2 &
        0     &  0     &  0     \\
0     &  0     &  0     &  1     &  0     &  0     &
         0     &  0     & -0     \\
0     &  0     &  \sqrt{3}/2 &  0     &  0     &  0     &
        1/2    &  0     &  0  
\end{pmatrix},
\end{align}

such that $\mathcal{M}_R[U_{\text{nd}}']=U_{\text{nd}}$ i.e. $U_{\text{nd}}$ (subscript 'nd' is used to emphasize it is not dual-unitary) is a fixed point of the $\mathcal{M}_R^2$ map
\[
\mathcal{M}_R^2[U_{\text{nd}}]:=\mathcal{M}_R[\mathcal{M}_R[U_{\text{nd}}]]=\mathcal{M}_R[U_{\text{nd}}']=U_{\text{nd}}.
\]
Interestingly, $U_{\text{nd}}$ and $U_{\text{nd}}'$ are LU equivalent,
\beq
U_{\text{nd}}=(u_1 \otimes u_2) U_{\text{nd}}' (v_1 \otimes v_2),
\eeq

where
$u_1=\begin{pmatrix}
-\sqrt{3}/2 & 0 & -1/2 \\
0 & 1 & 0 \\
-1/2 & 0 & \sqrt{3}/2 
\end{pmatrix}, \,  u_2=v_1=\begin{pmatrix}
1 & 0 & 0 \\
0 & 0 & 1 \\
0 & -1 & 0
\end{pmatrix}$, and $v_2=\begin{pmatrix}
-1 & 0 & 0 \\
0 & 1 & 0 \\
0 & 0 & -1
\end{pmatrix}$.

$U_{\text{nd}}$ ($U_{\text{nd}}'$) is not dual-unitary; $U_{\text{nd}}^RU_{\text{nd}}^{R \, \dagger} \neq \mathbb{I}$, having three distinct Schmidt values given by $\left\lbrace 1+\frac{\sqrt{3}}{2}\, ,\, 1\, ,\, 1-\frac{\sqrt{3}}{2} \right\rbrace,$ with each value repeated three times. As a consequence of LU equivalence $U_{\text{nd}}$ and $U_{\text{nd}}'$ have the same entangling power and gate-typicality.
%It is interesting to mention that $U_1=\mathcal{M}_R[U_{\text{nd}}]$ is such that singular values of both $U_1^R$ and $U_1^{\Gamma}$ are same as that of $U_{\text{nd}}^R$ and $U_{\text{nd}}^{\Gamma}$ respectively and suggests that these may be LU equivalent in that case the fixed point

\end{document}